\documentclass[11 pt]{article}

\usepackage[utf8]{inputenc}
\usepackage{fullpage}
\usepackage{booktabs} 
\usepackage{graphicx}
\usepackage{caption}
\usepackage{subfig}
\usepackage{graphics}
\usepackage{wrapfig}
\usepackage{float}
\usepackage{amsmath}
\usepackage{amsthm}
\usepackage{amssymb}
\usepackage{soul}
\usepackage[margin= 1 in]{geometry}
\usepackage{marginnote}
\usepackage{color}
\usepackage{physics}
\usepackage{pgf,tikz}
\usetikzlibrary{arrows,arrows.meta}
\usepgflibrary{arrows}

\usepackage[style=numeric]{biblatex}
\AtEveryBibitem{
    \clearfield{note}
    \clearfield{issn}
    \clearfield{isbn}
    \clearfield{doi}
    \clearfield{url}
}

\usepackage[toc,page,header]{appendix}
\usepackage{minitoc}

\usepackage{float}
\usepackage{amsmath}
\usepackage{amsthm}
\usepackage{amssymb}
\usepackage{physics}

\usepackage[ruled]{algorithm2e} 

\usepackage{appendix}

\usepackage{soul}
\usepackage[margin= 1 in]{geometry}

\usepackage{graphicx}
\usepackage{caption}
\usepackage{subfig}
\usepackage{graphics}
\usepackage{wrapfig}
\usepackage{float}
\usepackage{pgf,tikz}
\usetikzlibrary{arrows,arrows.meta}
\usepgflibrary{arrows}

\newtheorem{theorem}{Theorem}[section]
\newtheorem{defn}{Definition} 
\newtheorem{lemma}{Lemma}[section]
\newtheorem{cor}{Corollary} 
\newtheorem{claim}{Claim}[section]
\newtheorem{obs}{Observation} 
\newtheorem{asspt}{Assumption} 

\newcommand{\RJC}[1]{{\color{blue}RJC: #1}}

\newcommand{\IA}[1]{{\color{red}IA: #1}}
\newcommand{\hide}[1]{}

\newcommand{\eps}{\epsilon}

\newcommand{\nmax}{n_{\max}}

\newcommand{\ra}{\rightarrow}
\newcommand{\nbar}{\overline{n}}

\newcommand{\sigmabar}{\overline{\sigma}}
\newcommand{\xbar}{\overline{x}}

\newcommand{\Lbar}{\overline{L}}

\newcommand{\ceil}[1]{\lceil #1 \rceil}
\newcommand{\floor}[1]{\lfloor #1 \rfloor}

\title{Stable Matching: Choosing Which Proposals to Make}
\author{Ishan Agarwal \thanks{New York University. This work was supported in part by NSF Grant CCF-1909538.} \and Richard Cole \footnotemark[1]}

\begin{document}

\maketitle

\begin{abstract}

To guarantee all agents are matched in general, the classic Deferred Acceptance algorithm 
needs complete preference lists.
In practice, preference lists are short, yet stable matching still works well.
This raises two questions:
\begin{itemize} 
    \item Why does it work well? 
    \item Which proposals should agents include in their preference lists?
\end{itemize}

We study these questions in a model, introduced by Lee~\cite{Lee16}, with preferences based on correlated cardinal utilities: these utilities are based on common public ratings of each agent together with individual private adjustments.
Lee showed that for suitable utility functions, in large markets, with high probability, for most agents, all stable matchings yield similar valued utilities.
By means of a new analysis, we strengthen Lee's result, showing that in large markets, with high probability, for \emph{all} but the agents with the lowest public ratings,
all stable matchings yield similar valued utilities. We can then deduce that for \emph{all} but the agents with the lowest public ratings, each agent has an easily identified length $O(\log n)$ preference list that includes all of its stable matches,
addressing the second question above.
We note that this identification uses an initial communication phase.

We extend these results to settings where
 the two sides have unequal numbers of agents,  to many-to-one settings, e.g.\ employers and workers, and we also show the existence of an $\epsilon$-Bayes-Nash equilibrium in which every agent makes relatively few proposals. These results all rely on a new technique for sidestepping the conditioning between the tentative matching events that occur over the course of a run of the Deferred Acceptance algorithm.
 We complement these theoretical results with an experimental study. 
\end{abstract}

\section{Introduction}
\label{sec::intro}

Consider a doctor applying for residency positions.
Where should she apply? To the very top programs for her specialty?
Or to those where she believes she has a reasonable chance of success (if these differ)?
And if the latter, how does she identify them?
%
We study these questions in the context of Gale and Shapley's deferred acceptance (DA) algorithm~\cite{GS}.
It is well-known that in DA the optimal strategy for the proposing side is to list their choices in order of preference. However, this does not address which choices to list.

The DA algorithm is widely used to compute matchings in real-world applications:
the National Residency Matching Program (NRMP), which matches future residents to hospital programs~\cite{RP99}; university admissions programs which match students to programs, e.g.\ in Chile~\cite{RLPC21},
school choice programs, e.g.\ for placement in New York City's high schools~\cite{APR05},
the Israeli psychology Masters match~\cite{HRS17}, and no doubt many others (e.g.~\cite{GNKR19}). 

Recall that each agent provides the mechanism a list of its possible matches in preference order,
including the possibility of ``no match'' as one of its preferences.
These mechanisms promise that the output will be a stable matching with respect to the submitted preference lists.
In practice, preference lists are relatively short.
This may be directly imposed by the mechanism or could be a reflection of the
costs---for example, in time or money---of determining these preferences. Note that a short preference list is implicitly stating that the next preference after the listed
ones is ``no match''.

Thus it is important to understand the impact of short preference lists. Roth and Peranson observed that the NRMP data showed that preference lists were short compared to the number of programs and that these preferences yielded a single stable partner for most
participants; 
we note that this single stable partner could be the ``no match'' choice, and in fact this is the outcome 
for a constant fraction of the participants.
They also confirmed this theoretically for the simplest model of uncorrelated random preferences; namely that with the preference lists truncated to the top O(1) preferences, almost all agents have a unique stable partner.
Subsequently, in~\cite{IM15} the same result was obtained in the more general popularity model which allows for correlations among different agents' preferences; in their model, the first side---men---can have arbitrary preferences; on the second side---women---preferences are selected by weighted random choices, the weights representing the ``popularity'' of the different choices. These results were further extended by Kojima and Parthak
in~\cite{KP09}.

The popularity model does not capture behavior in settings 
where bounds on the number of proposals lead to proposals
being made to plausible partners, i.e.\ partners with
whom one has a realistic chance of matching.
One way to capture such settings is by way of tiers~\cite{ABKS19},
also known as block correlation~\cite{CKN13}.
Here agents on each side are partitioned into tiers, with all agents in a higher tier preferred to agents in a lower tier, and with uniformly random preferences within a tier. Tiers on the two sides may have different sizes.
If we assign tiers successive intervals of ranks equal to their size,
then, in any stable matching, the only matches will be between agents in tiers whose rank intervals overlap.

A more nuanced way of achieving these types of preferences 
bases agent preferences on cardinal utilities; for each side, these utilities are functions of an underlying common assessment of the other side, together with idiosyncratic individual adjustments for the agents on the other side.
These include the separable utilities defined by Ashlagi, Braverman, Kanoria and Shi
in~\cite{ABKS19},
and another class of utilities
introduced by Lee
in~\cite{Lee16}. This last model will be the focus of our study. 

To make this more concrete, we review a simple special case of Lee's model, the \emph{linear
separable model}. Suppose that there are $n$ men and $n$ women seeking to match with each other.
Each man $m$ has a public rating $r_m$, a uniform random draw from $[0,1]$. These ratings can be viewed
as the women's joint common assessment of the men. In addition, each woman $w$ has an individual adjustment, which we call a score,
$s_w(m)$ for man $m$, again a uniform random draw from $[0,1]$. All the draws are independent. Woman $w$'s utility for man $m$ is given
by $\tfrac 12 [r_m + s_w(m)]$; her full preference list has the men in decreasing utility order. The men's utilities are defined similarly.

Lee stated that rather than being assumed, short preference lists should arise from the model; this appears to have been a motivation for the model he introduced. A natural first step would be to show that for some or all stable matchings, the utility of each agent can be well-predicted, for this would then allow the agents to limit themselves to the proposals achieving such a utility. Lee proved an approximate version of this statement, namely
that with high probability (w.h.p., for short) most agents obtain utility within a small $\eps$ of an easily-computed individual benchmark. However, this does not imply that agents can restrict their proposals to a reduced utility range. (See the paragraph preceding Definition~\ref{defn::loss_linear} for the specification of the benchmarks.)

Our work seeks to resolve this issue.
We obtain the following results.
Note that in these results, when we refer to the bottommost agents, we mean when ordered by decreasing public rating. 
Also, we let the term loss mean the difference between an agent's benchmark utility and their achieved utility.
\begin{enumerate}
\item
We show that in the linearly separable model, for any constant $c>0$, with probability $1-1/n^c$, in every stable matching, apart from a sub-constant $\sigma$ fraction of the bottommost agents, \emph{all} the other agents obtain utility equal to an easily-computed individual benchmark $\pm\eps$, where $\eps$ is also sub-constant.

We show that both $\sigma,\eps=\widetilde{\Theta}(n^{-1/3})$.\footnote{The $\widetilde{\Theta}(\cdot)$ notation means up to a poly-logarithmic term; here $\sigma,\eps=\Theta((n/\ln n)^{-1/3})$.}
As we will see, this implies, w.h.p., that for all the agents other than the bottommost $\sigma$ fraction, 
each agent has $\Theta(\ln n)$ possible edges
(proposals) that could be in any stable matching, 
namely the proposals that provide both agents utility within $\eps$ of their benchmark.
Furthermore, we show our bound is tight: with fairly high probability, there is no matching, let alone stable matching, providing every agent a partner if the values of $\eps$ and $\sigma$ are reduced by a suitable constant factor. 

An interesting consequence of this 
lower bound on the agents' utilities is that the agents can readily identify a moderate sized subset of the edge set to which they can safely restrict their applications. More precisely, any woman $w$ outside the bottommost $\sigma$ fraction, knowing only her own public rating, the public ratings of the men, and her own private score for each man, can determine a preference list of length $\widetilde{\Theta}(n^{1/3})$ which, w.h.p, will yield the same result 
as her true full-length list. 
Our analysis also shows that if $w$ obtained the men's private scores for these proposals, 
then w.h.p.\ she could safely limit herself to a length $O(\ln n)$ preference list. 
\item
The above bounds apply not only to the linearly separable model, but to a significantly more general bounded derivative model (in which derivatives of the utility functions are bounded).
\item
The result also immediately extends to settings with unequal numbers of men and women.
Essentially, our analysis shows that the loss for an agent is small if there is a $\sigma$ fraction of agents of lower rank on the opposite side.
Thus even on the longer side, w.h.p., the topmost $n(1-\sigma)$ agents all obtain utility close to their benchmark, where $n$ is the size of the shorter side.
This limits the ``stark effect of competition''~\cite{AKL17}---namely that the agents on the longer side are significantly worse off---to a lower portion of the agents on the longer side.
\item
The result extends to the many-to-one setting, in which agents on one side seek multiple matches.
Our results are given w.r.t.\ a parameter $d$, the number of matches that each agent on the ``many'' side desires. For simplicity, we
assume this parameter is the same for all these agents.
In fact, we analyze a more general many-to-many setting.
\item
A weaker result with arbitrarily small $\sigma, \eps=\Theta(1)$ holds when
there is no restriction on the derivatives of the utility functions, which we call the general values model. Again, we show this bound cannot be improved in general.
This setting is essentially the general setting considered by Lee~\cite{Lee16}. He had shown there was a $\sigma$ fraction of agents who might suffer larger losses; our bound identifies
this $\sigma$ fraction of agents
as 
the bottommost agents.
\item In the bounded derivative model, with slightly stronger constraints on the derivatives, we also show the existence of an $\epsilon$-Bayes-Nash equilibrium in which no agent proposes more than $O(\ln^2 n)$ times and all but the bottommost $O((\ln n/n)^{1/3})$ fraction of the agents make only the $O(\ln n)$ proposals identified in (1) above. Here $\eps=\Theta(\ln n/n^{1/3})$.
\end{enumerate}

These results all follow from a lemma showing that, w.h.p., each non-bottommost agent has at most a small loss.
In turn, the proof of this lemma relies on a new technique which sidesteps the conditioning inherent to runs of DA in these settings.

\paragraph*{Experimental results} Much prior work has been concerned with preference lists that have a constant bound on their length.
For moderate values of $n$, say $n\in[10^3,10^6]$, $\ln n$ is quite small, so our $\Theta(\ln n)$ bound may or may not be sufficiently small in practice for this range of $n$.
What matters are the actual constants hidden by
the $\Theta$ notation, which our analysis
does not fully determine.
To help resolve this, we conducted a variety of simulation experiments.

We have also considered how to select the agents to include in the preference lists, when seeking to maintain a constant bound on their lengths, namely a bound that, for the values of $n$ we considered, was smaller than the $\Theta(\ln n)$ bound determined by the above simulations; again, our investigation
was experimental.

\paragraph*{Other Related work}
The random preference model was introduced by Knuth~\cite{Knuth76} (for a version in English see~\cite{Knuth96}), and subsequently extensively analyzed~\cite{Pittel89,KMRP90,Pittel92,Mertens05,PSV07,Pittel18,Kupfer20}.
In this model, each agent's preferences are an independent uniform random permutation of the agents on the other side.
An important observation was that when running the DA algorithm, the proposing side 
obtained a match of rank $\Theta(\ln n)$ on the average,
while on the other side the matches had rank $\Theta(n/\ln n)$.

A recent and unexpected observation in~\cite{AKL17} was the
``stark effect of competition'': that in the random preferences model the short side, whether it was the proposing side or not, was the one to enjoy the $\Theta(\ln n)$ rank matches.
Subsequent work showed that
this effect disappeared with short preference lists
in a natural modification of the random preferences model~\cite{KMQ21}.
Our work suggests yet another explanation for why this effect may not be present: it does not require that short preference lists be imposed as an external constraint, but rather that the preference model generates few edges that might ever be in a stable matching.

The number of edges present in any stable matching has also been examined for a variety of settings. When preference lists are uniform the expected number of stable pairs is $\Theta(n\ln n)$~\cite{Pittel92}; when they are arbitrary on one side and uniform on the other side, the
expected number is $O(n\ln n)$~\cite{KMRP90}. This result
continues to hold when preference lists
are arbitrary on the men’s side and are generated from general popularities on the women’s side~\cite{GMM19}. Our analysis shows that in the linear separable model (and more generally in the bounded derivative setting) the expected number of stable pairs is also $O(n\ln n)$.

Another important issue is the amount of communication needed to identify who to place on one's preference lists
when they have bounded length. In general, the cost is $\Omega(n)$ per agent (in an $n$ agent market)~\cite{GNOR15},
but in the already-mentioned separable model of Ashlagi et al.\ \cite{ABKS19} this improves to $\widetilde{O}(\sqrt{n})$
given some additional constraints, and further improves to $O(\ln^4n)$ in a tiered separable market~\cite{ABKS19}.
We note that for the bounded derivatives setting, with high probability, the communication cost will be
$O(n^{1/3}\ln^{2/3}n)$ for all agents except the bottommost $\Theta(n^{2/3}\ln^{1/3} n)$,
for whom the cost can reach $O(n^{2/3}\ln^{1/3} n)$.

Another approach to selecting which universities to apply to was considered by Shorrer who devised a dynamic program to compute the optimal choices for students assuming universities had a common ranking of students~\cite{Shorrer19}.

\paragraph*{Roadmap} 
In Section~\ref{sec::prelim}
we review some standard material.
In Section~\ref{sec::linear} we state our main result in two parts: Theorem~\ref{thm::basic-linear-result}, which bounds the losses in the setting of the linear model, and 
Theorem~\ref{lem::acceptable-suffice}, which shows it suffices to limit preference lists
to a small set of edges. We prove these theorems in Sections~\ref{sec::sketch-linear-result} and~\ref{sec::fewer-props}, respectively. We also
present some numerical simulations for the linear separable model in Section~\ref{sec:simulations}
We conclude 
with a brief discussion of open problems in Section~\ref{sec:discussion}.

Following this, in the appendices, we formally state and prove all the other results alluded to in the introduction and we also
present further numerical simulations for the linear separable model.
A complete summary of their content 
is given in Appendix~\ref{sec::appendix-overview}.

\hide{

The analysis for the lower bounds in more general models is based on the key ideas involved in proving Theorems \ref{thm::basic-linear-result} and \ref{lem::acceptable-suffice}. The upper bounds require different, though relatively simple arguments \RJC{I would not call the lower bounds in the general model ``quite simple''}. The $\epsilon$-Bayes-Nash Equilibrium result uses ideas from our analysis for Theorem \ref{thm::basic-linear-result} as well as several other ideas, resulting in quite an involved proof. The appendices conclude by presenting numerical simulations for the linear separable model.  \IA{The second paragraph shouldn't be in the roadmap but I think we should have it somewhere within the first 15 pages.} \RJC{Maybe move some of this to Appendix A.}
}

\hide{In Section~\ref{sec::prelim}
we review some standard material.
In Section~\ref{sec::linear} we state our main result in the setting of the linear model and sketch its proof in Section~\ref{sec::sketch-linear-result}. 
We conclude the first ten pages with a brief discussion of open problems in Section~\ref{sec:discussion}.
 
Next, 
in Section~\ref{sec::utility_models}, we define some more general utility models. In Section~\ref{sec::results} we state our results for these models,
and in Section~\ref{sec::analysis} we sketch their proofs.
In Section~\ref{sec:simulations} we present numerical simulations for the linear separable model. We end the second ten pages with a discussion of other related work.

Finally, we present the full proofs of our results in the appendix.}
 
\section{Preliminaries} \label{sec::prelim}

\subsection{Stable Matching and the Deferred Acceptance (DA) Algorithm}

Let $M$ be a set of $n$ men and $W$ a set of $n$ women. Each man $m$ has an ordered list of women that represents his preferences, i.e.\ if a woman $w$ comes before a woman $w'$ in $m$'s list, then $m$ would prefer matching with $w$ rather than $w'$. 
The position of a woman $w$ in this list is called $m$'s ranking of $w$. Similarly each woman $w$ has a ranking of her preferred men\footnote{Throughout this paper, we assume that each man $m$ (woman $w$) ranks all the possible women (men), i.e.\ $m$’s ($w$’s) preference list is complete.}. The stable matching task is to pair (match) the men and women in such a way that no two people prefer each other to their assigned partners. More formally:

\begin{defn}[Matching]
A matching is a pairing of the agents in $M$ with the agents in $W$.
It comprises a bijective function $\mu$ from $M$ to $W$, and its inverse
$\nu=\mu^{-1}$, which is a bijective function from $W$ to $M$. 
\end{defn}

\begin{defn}[Blocking pair]
A matching $\mu$ has a blocking pair $(m,w)$ if and only if: 
\begin{enumerate}
    \item $m$ and $w$ are not matched: $\mu(m)\neq w$.
    \item $m$ prefers $w$ to his current match $\mu(m)$.
    \item $w$ prefers $m$ to her current match $\nu(w)$. 
\end{enumerate}
\end{defn}

\begin{defn}[Stable matching]
A matching $\mu$ is stable if it has no blocking pair.
\end{defn}
\begin{algorithm}[t]
\SetAlgoNoLine
Initially, all the men and women are unmatched.\\
\While{some woman $w$ with a non-empty preference list is unmatched}{ 
   let $m$ be the first man on her preference list\;
   \If{$m$ is currently unmatched}
      {tentatively match $w$ to $m$.}
      {\eIf{$m$ is currently matched to $w'$, and $m$ prefers $w$ to $w'$}
      {make $w'$ unmatched and tentatively match $w$ to $m$.}
      {remove $m$ from $w$'s preference list.}
    }
 }
\caption{\textsf{Woman Proposing Deferred Acceptance (DA) Algorithm.}}
\label{alg:DA}
\end{algorithm}

Gale and Shapley \cite{GS} proposed the seminal deferred acceptance (DA) algorithm for the stable matching problem. We present the woman-proposing DA algorithm (Algorithm $1$); the man-proposing DA is symmetric.
The following facts about the DA algorithm are well known. We state them here without proof and we shall use them freely in our analysis. 

\begin{obs}
\hspace{0.0001in}
\begin{enumerate}
    \item \label{obs::DA_terminates_in stable_match}
DA terminates and outputs a stable matching.
\item\label{obs::DA_invariant_of_proposal_order}
The stable matching generated by DA is independent of the order in which the unmatched agents on the proposing side are processed.
\item\label{obs::DA_optimality}
Woman-proposing DA is woman-optimal, i.e.\ each woman is matched with the best partner she could be matched with in any stable matching.
\item \label{obs::DA_pessimality}
Woman-proposing DA is man-pessimal, i.e.\ each man is matched with the worst partner he could be matched with in any stable matching.
\end{enumerate}
\end{obs}

\subsection{Useful notation and definitions}

There are $n$ men and $n$ women.
In all of our models, each man $m$ has a utility $U_{m,w}$ for the woman $w$, and each woman $w$ has a utility $V_{m,w}$ for the man $m$. These utilities are defined as
\begin{align*}
    &U_{m,w}=U(r_w,s_m(w)), \text{ and}\\
    &V_{m,w}=V(r_m,s_w(m)),
\end{align*}
where $r_m$ and $r_w$ are common public ratings, $s_m(w)$ and $s_w(m)$ are private scores specific to the pair $(m,w)$, and $U(\cdot,\cdot)$ and $V(\cdot,\cdot)$ are continuous and strictly increasing functions from $\mathbb{R}^2_+$ to $\mathbb{R}_+$.
The ratings are independent uniform draws from $[0,1]$ as are the scores.

In the \emph{Linear Separable Model},
each man
$m$ assigns each woman $w$ a utility of $U_{m,w}=\lambda\cdot r_w+(1-\lambda)\cdot s_m(w)$, where $0< \lambda < 1$ is a constant.
The women's utilities for the men are defined analogously as $V_{m,w}=\lambda\cdot r_m+(1-\lambda)\cdot s_w(m)$.
All our experiments are for this model.

We let $\{m_1, m_2, \ldots, m_n\}$ be the men in descending order of their public ratings and $\{w_1, w_2, \ldots, w_n\}$ be a similar ordering of the women. 
We say that $m_i$ has public rank $i$, or rank $i$ for short,
and similarly for $w_i$.
We also say that $m_i$ and $w_i$ are \emph{aligned}.
In addition, we often want to identify the men or women in an interval of public ratings.
Accordingly, we define $M(r,r')$ to be the set of men with public ratings in the range $(r,r')$,
and $M[r,r']$ to be the set with public ratings in the range $[r,r']$; we also use the notation
$M(r,r']$ and $M[r,r')$ to identify the men with ratings in the corresponding semi-open intervals.
We use an analogous notation, with $W$ replacing $M$, to refer to the corresponding sets of women.

We will be comparing the achieved utilities in stable matchings to the following benchmarks: the rank $i$ man has as benchmark $U(r_{w_i},1)$,
the utility he would obtain from the combination of the rank $i$ woman's public rating and the highest possible private score;
and similarly for the women.
Based on this we define the loss an agent faces as follows.

\begin{defn} [Loss]
\label{defn::loss_linear} 
Suppose man $m$ and woman $w$ both have rank $i$.
The \emph{loss} $m$ sustains from a match of utility $u$ is defined to be $U(r_w,1) - u$.
The loss for women is defined analogously.
\end{defn}

In our analysis we will consider a complete bipartite graph whose two sets of vertices correspond to the men and women, respectively. For each man $m$ and woman $w$, we view the possible matched pair $(m,w)$ as an edge in this graph. Thus, throughout this work, we will often refer to edges being proposed, as well as edges satisfying various conditions.

\hide{
There are $n$ men and $n$ women. 
Each man $m$ has a public rating $r_m$ and each women $w$ has a public rating $w$. These are uniform draws from $[0,1]$.

The utilities of the agents for each other are based on the public rankings and individual adjustments called scores.
Each woman $w$ has a score $s_w(m)$ for each man $m$;  likewise, each man $m$ has a score $s_m(w)$ for each woman $w$.
The scores are draws from $[0,1]$.}

\section{Upper Bound in The Linear Separable Model}
\label{sec::linear}

To illustrate our proof technique for deriving upper bounds, we begin by stating and proving our upper bound result for the special case of the linear separable model with $\lambda= \tfrac 12$. 

\begin{theorem}\label{thm::basic-linear-result}
In the linear separable model with $\lambda=1/2$, when there are $n$ men and $n$ women,
for any given constant $c>0$, for large enough $n$, with probability at least $1-n^{-c}$, in every stable matching, for every $i$, with $r_{w_i}\ge \sigmabar\triangleq 3\Lbar/2$, agent $m_i$ suffers a loss of at most $\Lbar$, where $\Lbar=(16(c+2)\ln n/n)^{1/3}$,
and similarly for the agents $w_i$.
\end{theorem}

In words, w.h.p., all but the bottommost agents (those whose aligned agents have 
public rating less than $\sigmabar$) suffer a loss of no more than $\Lbar$. This is a special case of our basic upper bound for the bounded utilities model (Theorem \ref{thm:basic-bdd-deriv-result}).

One of our goals is to be able to limit the number of proposals the proposing side needs to make.
We identify the edges that could be in some stable matching, calling them acceptable edges.
Our definition is stated generally so that it covers all our results; accordingly we replace the terms $\overline{L}$ and $\overline{\sigma}$ in Theorem~\ref{thm::basic-linear-result} with parameters $L$ and $\sigma$.

\begin{defn} [Acceptable edges] 
Let $0<\sigma <1$ and $0<L<1$ be two parameters.
An edge $(m_i,w_j)$
is $(L,\sigma)$-\emph{man-acceptable} either if it provides $m_i$ utility at least $U(r_{w_i},1)-L$, or if $m_i\in M[0,\sigma)$.
The definition of $(L,\sigma)$-woman-acceptable is symmetric. Finally, $(m_i,w_j)$ is $(L,\sigma)$-\emph{acceptable} if it is both $(L,\sigma)$-man and $(L,\sigma)$-woman-acceptable.
\end{defn}

To prove our various results, we choose $L$ and $\sigma$ so that w.h.p.\ the edges in every stable matching are $(L,\sigma)$-\emph{acceptable}.
We call this high probability event $\mathcal E$.
We will show that if $\mathcal E$ occurs, then running DA on the set of acceptable edges, or any superset of the acceptable edges
obtained via loss thresholds, 
produces the same stable matching as running DA on the full set of edges.

\begin{theorem}\label{lem::acceptable-suffice}
If $\mathcal E$ occurs, then running woman-proposing DA with the edge set restricted to the acceptable edges or to any superset of the acceptable edges obtained via loss thresholds (including the full edge set) result in the same stable matching.

\end{theorem}

The implication is that w.h.p.\ a woman can safely restrict her proposals to her acceptable edges, or to any overestimate of this set of edges obtained by her setting an upper bound on the loss she is willing to accept.
There is a small probability--- at most $n^{-c}$---that this may result in a less good outcome, which can happen only if
$\mathcal E$ does not occur.
Note that Theorem~\ref{lem::acceptable-suffice} applies to every utility model we consider.
Then,
w.h.p., every stable matching gives each woman $w$, whose aligned agent $m$ has public rating $r_m\ge\sigmabar=\Omega((\ln n/n)^{1/3})$, a partner with public rating in the range $[r_m - 2\Lbar, r_m+\tfrac 52 \Lbar]$ (see Theorem~\ref{thm::cone-bdd-deriv} in Appendix~\ref{sec::cone-size}).
The bound $r_m-2\Lbar$ is a consequence of the bound on the woman's loss; the bound $r_m+\tfrac52\Lbar$ is a consequence of the bound on the men's losses. 
An analogous statement applies to the men. 

This means that if we are running woman-proposing DA, each of these women might as well limit her proposals to her woman-acceptable edges, which is at most the men with public ratings in the range $r_m \pm \Theta(\Lbar)$ for whom she has private scores of at least $1 - \Theta(\Lbar)$.
In expectation, this yields $\Theta(n^{1/3}(\ln n)^{2/3})$ men to whom it might be worth proposing.
It also implies that a woman can have a gain of at most $\Theta(\Lbar)$ compared to her target utility.

If, in addition, each man can inexpensively signal the women who are man-acceptable to him,
then the women can further limit their proposals to just those men providing them with a signal; 
this reduces the expected number of proposals these women can usefully make to just $\Theta(\ln n)$. 

\section{Sketch of the Proof of Theorem~\ref{thm::basic-linear-result}}
\label{sec::sketch-linear-result}

\begin{figure}[thb]
\begin{center}
\begin{tikzpicture}
\begin{scope}[>=latex]
\draw (-0.2,3) to (-0.2,1.55);
\node at (-0.2,1.485) {$\circ$};
\node at (-2.8,1.485) {rating $r_{m_i}$, man $m_i$};
\draw[->] (-1.15,1.485) to (-0.45,1.485);
\node at (3,1.485) {$\circ$};
\node at (4.1,1.485) {woman $w_i$};
\draw[dashed] (0.2,1.485) to (2.9,1.485);
\draw (-0.2,1.46) to (-0.2,0.5);
\draw[<->] (-0.0,3) to (-0.0,0.5);
\node at (1.5,2.0) {$M_i=[r_{m_i}-\alpha,1]$};
\draw[<->] (-0.4,1.475) to (-0.4,0.5);
\node at (-2.5,1.0) {$h_i$ men, rating range $\alpha$};
\draw (3,3) to (3,1.55);
\draw (3,1.46) to (3,-0.95);
\draw[<->] (3.2,1.475) to (3.2,-1);
\node at (4.1,0.25) {$\ell_i$ women};
\node at (0, 3.3) {men};
\node at (3, 3.3) {women};
\node at (-3.2,0.5) {cutoff $r_{m_i} - \alpha$};
\draw[->] (-1.9,0.5) to (-0.45,0.5);
\draw[<->] (5.0,3) to (5.0,-1);
\node at (6.4,1.0) {$\widetilde{W}_i=W[r_{\overline{w}_i},1]$};
\node at (3,-1) {$\circ$};
\node at (3.2,-1.3) {woman $\overline{w}_i=w_{i+\ell_i}$};
\end{scope}
\end{tikzpicture}
\end{center}
\caption{
} 
\label{fig::key_bounded}
\end{figure}

We begin by outlining the main ideas used in our analysis. 
Our goal is to show that when we run woman proposing DA, w.h.p.\ each man
receives a proposal that gives him a loss of at most $L$ (except possibly for men among the bottommost $\Theta(nL)$).
As the outcome is the man-pessimal stable matching, this means that w.h.p., in all stable matchings,
these men have a loss of at most $L$. By symmetry, the same bound holds for the women.

Next, we provide some intuition for the proof of this result. See Fig.~\ref{fig::key_bounded}. 
Our analysis uses 3 parameters
$\alpha, \beta, \gamma = \Theta(L)$. Let $m_i$ be a non-bottommost man.
We consider the set of men with public rank at least $r_{m_i}-\alpha$: 
$M_i = M[r_{m_i}-\alpha,1]$.
We consider a similar, slightly larger set of women: $\widetilde{W}_i=W[r_{w_i}-3\alpha,1]$.
Now we look at the best proposals by the women in $\widetilde{W}_i$, i.e.\ the ones they make first.
Specifically, we consider the proposals that give these women utility
at least $V(r_{m_i}-\alpha,1)$, proposals that are therefore guaranteed
to be to the men in $M_i$.
Let $\big|M_i\big|= i + h_i$ and $\big|W_i\big|= i + \ell_i$.
In expectation, $\ell_i-h_i= 2\alpha n$.
Necessarily, at least $\ell_i -h_i+1$ women in $M_i$ cannot match with men in $M_i\setminus\{m_i\}$.
But, as we will see, these women all have probability at least $\beta$ of
having a proposal to $m_i$ which gives them utility at least $V(r_{m_i}-\alpha,1)$.
These are proposals these women must make before they make any proposals to men
with public rating less than $r_{m_i}-\alpha$.
Furthermore, for each of these proposals, $m_i$ has probability at least
$\gamma$ of having a loss of $L$ or less.
Thus, in expectation, $m_i$ receives at least $2\alpha\beta\gamma n$ proposals which
give him a loss of $L$ or less.

We actually want a high-probability bound.
So we choose $\alpha, \beta,\gamma$ so that $\alpha\beta\gamma n\ge c \log n$ for a suitable constant $c>0$,
and then apply a series of Chernoff bounds.
There is one difficulty.
The Chernoff bounds requires the various proposals to be independent.
Unfortunately, in general, this does not appear to be the case.
However, we are able to show that the failure probability for our setting is at most
the failure probability in an artificial setting in which the events are independent,
which yields the desired bound.

We now embark on the actual proof.

We formalize the men's rating cutoff with the notion of DA stopping at public rating $r$.

\begin{defn}[DA stops]
The women \emph{stop at public rating $r$} if, in each woman's preference list, all the edges with utility less than
$V(r,1)$
are removed. 
The women \emph{stop at man $m$} if, in each woman's preference list, all the edges following her edge to $m$ are removed.
The women \emph{double cut at man $m$ and public rating $r$}, if they each stop at $m$ or $r$, whichever comes first.
Men stopping and double cutting are defined similarly. Finally, an edge is said to \emph{survive} the cutoff if it is not removed by the stopping.
\end{defn}

To obtain our bounds for man $m_i$, we will have the women double cut at rating $r_{m_i}-\alpha$ and at man $m_i$, where $\alpha>0$ is a parameter we will specify later.

Our upper bounds in all of the utility models depend on a parameterized key lemma (Lemma \ref{lem::key_bounded}) stated shortly.
This lemma concerns the losses the men face in the woman-proposing DA; a symmetric result applies to the women.
The individual theorems follow by setting the parameters appropriately.
Our key lemma uses three parameters:
$\alpha,\beta,\gamma >0$. 
To avoid rounding issues, we will choose $\alpha$ so that $\alpha n$ is an integer.
The other parameters need to satisfy the following constraints.
\begin{align}
\label{eqn::beta-constraint}
   &\text{for $r\ge \alpha$:}\hspace*{0.5in} V(r-\alpha,1) \le V(r,1-\beta)\\
    \label{eqn::gamma-constraint}
    &\text{for $r\ge 3\alpha$:}\hspace*{0.5in} U(r,1) - U(r-3\alpha,1-\gamma)\le L
    \end{align}
Equation~\eqref{eqn::beta-constraint} relates the range of private values that will yield a woman
an edge to $m_i$ that survives the cut at $r_{m_i} - \alpha$, or equivalently the probability of having such an edge.
Observation~\ref{obs::towards-man-accept} below, shows that Equation~\eqref{eqn::gamma-constraint} identifies
the range of $m_i$'s private values for proposals from $\widetilde{W}_i$ that yield him a loss of at most $L$
(for we will ensure the women in $\widetilde{W}_i$ have public rating at least $r_{w_i} - 3\alpha$).

\begin{obs}
\label{obs::towards-man-accept}
Consider the proposal from woman $w$ to the rank $i$ man $m_i$.
Suppose the rank $i$ woman $w_i$ has rating $r_{w_i}\ge 3\alpha$.
If $w$ has public rating $r\ge r_{w_i}-3\alpha$ and $m_i$'s private score for $w$ is at least $1-\gamma$,
then $m_i$'s utility for $w$ is at least $U(r_{w_i}-3\alpha,1-\gamma)\ge  U(r_{w_i},1)-L$.
\end{obs}

In the linear separable model with $\lambda = \tfrac 12$, we set $\alpha=\beta=\gamma$ and $L=2\alpha$.

The next lemma determines the probability that man $m_i$ receives a proposal 
causing him a loss of at most $L$.
The lemma calculates this probability in terms of the parameters we just defined. Note that the result does not depend on the utility functions $U(\cdot,\cdot)$ and $V(\cdot,\cdot)$ being linear. In fact, the same lemma applies to much more general utility models which we also study (see Section~\ref{sec::utility_models}) and it is the crucial tool we use in all our upper bound proofs.

In what follows, to avoid heavy-handed notation, by
$r_{m_i}-\alpha$ we will mean $\max\{0,r_{m_i} - \alpha\}$.

In order to state our next result crisply, we define the following Event ${\cal E}_i$.
It concerns a run of woman-proposing DA with double cut at the rank $i$ man $m_i$ 
and at public rating $r_{m_i} - \alpha$. 
Let $h_i=\big|M[r_{m_i} - \alpha,r_{m_i})\big|$,
$\ell_i=\big|W[r_{w_i}-3\alpha,r_{w_i})\big|$,
and $\overline{w}_i$ be the woman with rank $i+\ell_i$. See Figure \ref{fig::key_bounded} for an illustration of these definitions.
Event ${\cal E}_i$ occurs if $r_{w_i} \ge 3\alpha$ and
between them the $i+\ell_i$ women in $W[r_{{w}_i}-3\alpha,1]$ make at least one proposal to $m_i$ that 
causes him a loss of at most $L$.


Finally we define Event $\cal E$: it happens if ${\cal E}_i$ occurs for all $i$ such that $r_{w_i}\ge 3\alpha$.

\begin{lemma}\label{lem::key_bounded}
Let $\alpha>0$ and $L>0$ be given, and suppose that $\beta$ and $\gamma$ satisfy
\eqref{eqn::beta-constraint} and \eqref{eqn::gamma-constraint}, respectively.
Then, Event $\cal E$ occurs with probability at least $1-p_f$,
where the failure probability 
$$p_f=n\cdot \exp(-\alpha(n-1)/12)
+n\cdot \exp(-\alpha (n-1)/24)
    +n\exp(-\alpha\beta n/8)
    +n\cdot \exp(-\alpha\beta\gamma n/2).$$
\hide{
${cal E}$ is the occurrence of 
For Event ${\cal E}_i$, 
woman-proposing DA is run with double cut at the rank $i$ man $m_i$ 
and public rating $r_{m_i} - \alpha$. 
Let $h_i=\big|M[r_{m_i} - \alpha,r_{m_i}])\big|$,
$\ell_i = h_i + \alpha n -1$, and
$\overline{w}_i$ be the woman with rank $i+\ell_i$.
Event ${\cal E}_i$ occurs if
$\overline{w}_i$ exists and
between them the $i+\ell_i$ women in $W[r_{\overline{w}_i},1]$ make at least one proposal to $m_i$ that provides him utility at least $U(r_{w_i},1)-L$.
}
\end{lemma}
\hide{
We are considering an interval of men descending
from man $m_i$ covering a public rating range of length
$\alpha$. Suppose there are $h_i$ men in this interval.
We are also considering an interval of $\ell_i$ women 
descending from $w_i$, where $\ell_i =h_i+\alpha n -1$.
(See Fig.~\ref{fig::key_bounded}.)
}
The following simple claim notes that the men's loss when running the full DA is no larger than when running double-cut DA.

\begin{claim}
\label{clm::full-DA-better}
Suppose a woman-proposing double-cut DA at man $m_i$ and rating $r_{m_i}-\alpha$ is run, and suppose $m_i$ incurs a loss of $L$.
Then in the full run of woman-proposing DA, $m_i$ will incur a loss of at most $L$.
\end{claim}
\begin{proof}
Recall that when running the women-proposing DA the order in which unmatched women are processed does not affect the outcome.
Also note that as the run proceeds, whenever a man's match is updated, the man obtains an improved utility.
Thus, in the run with the full edge set we can first use the edges used in the double-cut DA and then proceed with 
the remaining edges.
Therefore if in the double-cut DA $m_i$ has a loss of $L$,
in the full run $m_i$ will also have a loss of at most $L$.
\end{proof}

To illustrate how this lemma is applied, we now prove
Theorem~\ref{thm::basic-linear-result}.
Note that $\overline{L}$ is the value of $L$ used in this theorem.
Our other results use other values of $L$.

\begin{proof} (Of Theorem~\ref{thm::basic-linear-result})~
By Lemma~\ref{lem::key_bounded}, in the double-cut DA, 
for all $i$ with $r_{w_i}\ge 3 \alpha$,
$m_i$ obtains a match giving him loss at most $\Lbar$,
with probability at least
$1-n\cdot \exp(-\alpha(n-1)/12)
-n\cdot \exp(-\alpha n/24)
    -n\exp(-\alpha^2 n/8)
    -n\cdot \exp(-\alpha^3 n/2)$.

By Claim~\ref{clm::full-DA-better}, $m_i$ will incur a loss of at most $\Lbar$ in the full run of woman-proposing DA with at least as large a probability.
But this is the man-pessimal match. Consequently, in every stable match,
$m_i$ has a loss of at most $\Lbar$. 
By symmetry, the same bound applies to each woman $w_i$ such that $r_{m_i}\ge 3\alpha$.

We choose $\Lbar = [16(c+2)\ln n/n]^{1/3}$. Recalling that $\alpha=\Lbar/2$, we see that for large enough $n$ the probability bound, over all the men and women, is at most $1-n^{-c}$.
The bounds $r_{w_i}\ge 3\alpha$ and $r_{m_i}\ge 3\alpha$ imply we can set $\sigmabar=3\alpha =\tfrac 32\Lbar$.
\end{proof}

\begin{proof} (Of Lemma \ref{lem::key_bounded}.)
We run the double-cut DA in two phases, defined as follows. 
Recall that $h_i=\big|M[r_{m_i}-\alpha,r_{m_i})\big|$ and 
$\ell_i=\big|W[r_{w_i}-3\alpha,r_{w_i})\big|$. Note that women with rank at most $i+\ell_i$ have public rating at least $r_{w_i}-3\alpha$.
\\
\emph{Phase 1}. Every unmatched woman with rank at most
$i+\ell_i$ 
keeps proposing until her next proposal
is to $m_i$,
or she runs out of proposals.\\
\emph{Phase 2}. Each unmatched women makes her next proposal, if any, which will be a proposal to $m_i$.

\hide{
In this sketch we will be using the term with high probability (w.h.p.) to mean
with probability at least $1-n^{c+2}$.
}

Our analysis is based on the following four claims. The first two are simply observations that w.h.p.\ the number of agents with
public ratings
in a given interval is close to the
expected number. We defer the proofs to the appendix.

A critical issue in this analysis is to make sure the conditioning induced by the successive steps of the analysis does not affect the independence needed for subsequent steps. To achieve this, we use the Principle of Deferred Decisions, only instantiating
random values as they are used.
Since each successive bound uses a different collection of random variables this does not present a problem.

\begin{claim}
\label{clm::span-man-interval}
Let ${\mathcal B}_1$ be the event that for some $i$,
$h_i \ge \tfrac 32 \alpha (n-1)$.
${\mathcal B}_1$ occurs with probability at most $n\cdot \exp(-\alpha(n-1)/12)$.
The only randomness used in the proof are the choices of the men's public ratings.
The same bound applies to the women.
\end{claim}
\begin{proof} (Sketch.)
As $\text{E}[h_i] =\alpha(n-1)$,
w.h.p.,
$h_i < \tfrac 32 \alpha (n-1)$.
This claim uses a Chernoff bound with the randomness coming from
the public ratings of the men.
\end{proof}

\begin{claim}
\label{clm::span-woman-interval}
Let ${\mathcal B}_2$ be the event that for some $i$, $\ell_i \le \tfrac 52 \alpha (n-1)$.
${\mathcal B}_2$ occurs with probability at most $n\cdot \exp(-\alpha (n-1)/24)$.
The only randomness used in the proof are the choices of the women's public ratings.
The same bound applies to the men.
\end{claim}
\begin{proof} 
This is very similar to the proof of Claim~\ref{clm::span-man-interval}.
\hide{
(Sketch.)
As $\text{E}[\ell_i] =\alpha(n-1)$,
w.h.p.,
$\ell_i \ge \tfrac 52 \alpha (n-1)/24$.
This claim uses a Chernoff bound 
with the randomness coming from the public ratings of the women.
}
\end{proof}

\begin{claim}
\label{clm::many-man-acceptable-props}
Let ${\mathcal B}_3$ be the event that between them, the women with
rank at most $i+\ell_i$ make fewer than
$\tfrac 12 \alpha\beta n$ Step 2 proposals to $m_i$.
If events ${\mathcal B}_1$ and ${\mathcal B}_2$
do not occur, then ${\mathcal B}_3$ occurs with probability at most $\exp(-\alpha\beta n/8)$.
The only randomness used in the proof are the choices of the women's private scores.
\end{claim}
This bound uses the private scores of the women
and employs a novel argument given below to sidestep the conditioning
among these proposals.

\begin{claim}
\label{clm::one-man-accept-prop}
If none of the events ${\mathcal B}_1$, ${\mathcal B}_2$, or
${\mathcal B}_3$
occur, then
at least one of the Step 2 proposals to $m_i$ will
cause him a loss of at most $L$ with probability at least
$1 -(1-\gamma)^{\alpha\beta n/2}\ge 1- \exp(-\alpha\beta\gamma n/2)$.
The only randomness used in the proof are the choices of the men's private scores.
\end{claim}
\begin{proof}
Note that each Phase 2 proposal is from a woman $w$ with rank at most $i+\ell_i$. 
As already observed, her public rating is at least $r_{w_i} -3\alpha$. 
Recall that man $m_i$'s utility for $w$ equals
$U(r_w,s_{m_i}(w))\ge U(r_{w_i}-3\alpha,s_{m_i}(w))$.
To achieve utility at least $U(r_{w_i},1)-L \le U(r_{w_i}-3\alpha,1-\gamma)$ (using \eqref{eqn::gamma-constraint}) it suffices to have
$s_{m_i}(w) \ge 1-\gamma$, which happens with
probability $\gamma$.
Consequently, utility at least $U(r_{w_i},1)-L$ is
achieved with probability at least $\gamma$.

For each Phase 2 proposal these probabilities are independent as they reflect $m_i$'s private
scores for each of these proposals.
%
Therefore the probability that there is no proposal providing $m_i$ a loss of at most $L$ is at most 
\begin{align*}
    \big(1-\gamma\big)^{\alpha\beta n/2}\le \exp(\alpha\beta\gamma n/2).
\end{align*}
\end{proof}

\noindent
Concluding the proof of Lemma \ref{lem::key_bounded}:
The overall failure probability summed over all $n$ choices of $i$
is
\begin{align*}
n\cdot \exp(-\alpha(n-1)/12)
+n\cdot \exp(-\alpha (n-1)/24)
    +n\exp(-\alpha\beta n/8)
    +n\cdot \exp(-\alpha\beta\gamma n/2).
\end{align*}
\end{proof}

\begin{proof} (Of Claim~\ref{clm::many-man-acceptable-props}.)
%
First, we simplify the action space by viewing the decisions as being made on a discrete utility space, as
specified in the next claim, proved in the appendix.
\begin{claim}
\label{clm::discrete-dist}
For any $\delta>0$, there is a discrete utility space in which
for each woman the probability of selecting $m_i$ is only increased,
and the probability of having any differences
in the sequence of actions in the original continuous setting and the discrete setting is at most $\delta$.
\end{claim}

We represent the possible computations of the double-cut DA in this discrete setting using a tree $T$. Each woman will be going through her possible utility values in decreasing order, with the possible actions of the various women being interleaved in the order given by the DA processing.
Each node $u$ corresponds to a woman $w$ processing her next utility value. The possible choices at this utility are each represented by an edge descending from $u$. These choices are: 

i. Proposing to some man 
(among those men $w$ has not yet proposed to); or 

ii. ``no action''. This corresponds 
to $w$ making no proposal achieving the current utility.

We observe the following important structural feature of tree $T$.
Let $S$ be the subtree descending from the edge corresponding to woman $w$ proposing to $m_i$;
in $S$ there are no further actions of $w$, i.e.\ no nodes at which $w$ makes a choice, because the double cut DA cuts at the proposal to $m_i$.

The assumption that ${\cal B}_1$ and ${\cal B}_2$ do not occur means that for all $i$, $h_i< \tfrac 32 \alpha(n-1)$ and
$\ell_i> \tfrac 52 \alpha(n-1)$, and therefore $\ell_i - h_i > \alpha(n-1)$.

At each leaf of $T$, up to $i+h_i -1$ 
women will have been matched with someone other than $m_i$.
The other women either finished with a proposal to $m_i$ or both failed to match and did not propose to $m_i$. Let $w$ be a woman in the latter category.
Then, 
on the path to this leaf, $w$ will have traversed edges corresponding to a choice at each discrete utility in the range $[V(r_{m_i}-\alpha,1),V(1,1)]$.

We now create an extended tree, $T_x$, by adding a subtree at each leaf; this subtree will correspond to pretending there were no matches; the effect is that 
each women will take an action at all their remaining utility values in the range
$[V(r_{m_i}-\alpha,1),V(1,1)]$,
except that in the sub-subtrees descending from edges that correspond to some woman $w$ selecting $m_i$, $w$ has no further actions.
For each leaf in the unextended tree, the probability
of the path to that leaf is left unchanged.
The probabilities of the paths in the extended tree are then calculated by multiplying the path probability in the unextended tree with the probabilities of each woman's choices in the extended portion of the tree.

Next, we create an artificial mechanism $\mathcal M$ that acts on tree $T_x$. The mechanism $\mathcal M$
is allowed to put $i+h_i -1$ 
``blocks'' on each path; blocks can be placed at internal nodes. A block names a woman $w$ and corresponds to her matching (but we no longer think of the matches as corresponding to the outcome of the edge selection; they have no meaning beyond making
all subsequent choices by this woman be the ``no action'' choice).

DA can be seen as choosing to place up to $i+h_i -1$ 
blocks at each of the nodes corresponding to a leaf of $T$.
$\mathcal M$ will place its blocks so as to minimize the probability $p$ of paths with at least $\tfrac12 \alpha\beta n$ women choosing edges to $m_i$.
Clearly $p$ is a lower bound on the probability
that the double-cut DA makes at least $\tfrac12 \alpha\beta n$ proposals in Step 2.
Given a choice of blocks we call the resulting probability of having fewer than $\tfrac12 \alpha\beta n$ women choosing edges to $m_i$ the \emph{blocking probability}.

\begin{claim}
\label{clm::many-props-despite-blocking}
The probability that $\mathcal M$ makes at least $\tfrac 12\alpha\beta n$ proposals to $m_i$ is at least $1-\exp(-\alpha\beta n/8)$.
\end{claim}
\hide{
\begin{proof}
We will show that the most effective blocking strategy is to block all but $\alpha n$ women before they have made any choices.
Then, the remaining $\alpha n$ women all have independent probability at least $\beta$ of making a proposal to $m_i$.
In expectation, there are at least $\alpha\beta n$ proposals to $m_i$, and therefore, by a Chernoff bound,
at least $\tfrac 12 \alpha\beta n$ proposals with probability at least $1-\exp(-\alpha\beta n/8)$.

The claim about the best blocking strategy follows
by showing that at a node $x$ such that the only blocks below it are at its children, the blocks, possibly modified, can be moved to $x$. This requires a case analysis, deferred to the appendix. Iterating the argument demonstrates the claim.
\end{proof}
}

\begin{cor}
\label{cor::many-props-to-m-mtm}
The probability that the double-cut DA makes at least $\tfrac 12\alpha\beta n$ proposals to $m_i$ is at least $1-\exp(-\alpha\beta n/8)$.
\end{cor}
\begin{proof}
For any fixed $\delta$, by Claim~\ref{clm::many-props-despite-blocking}, the probability that $\mathcal M$ makes at least $\tfrac 12\alpha\beta n$ proposals to $m_i$ is at least $1-\exp(-\alpha\beta n/8)$.
By construction, the probability is only larger for the double-cut DA in the discrete space.

Therefore, by Claim~\ref{clm::full-DA-better}, the probability that the double-cut DA makes at least $\tfrac 12\alpha\beta n$ proposals to $m_i$ in the actual continuous space is at least $1-\exp(-\alpha\beta n/8)-\delta$, and this holds
for any $\delta>0$, however small.
Consequently, this probability is at least $1-\exp(-\alpha\beta n/8)$.
\end{proof}

\begin{proof} (Of Claim~\ref{clm::many-props-despite-blocking}.)
We will show that the most effective blocking strategy is to block as many women as possible
before they have made any choices. This leaves at least $(i+\ell_i) - (i-1 +h_i) \ge 1 +\alpha (n-1) \ge \alpha n$ women unmatched.
Then, as we argue next, each of these remaining at least $\alpha n$ women $w$ has independent probability at least $\beta$ that their proposal to $m_i$ is cutoff-surviving.
To be cutoff-surviving, it suffices that
$V(r_{m_i},s_w(m_i))\ge V(r_{m_i}-\alpha,1)$.
But we know  by \eqref{eqn::beta-constraint} that $V(r_{m_i}-\alpha,1)\le V(r_{m_i},1-\beta)$, and therefore it suffices that
$s_w(m_i) \ge 1-\beta$, which occurs with probability $\beta$.

Consequently, in expectation, there are at least $\alpha\beta n$ proposals to $m_i$, and therefore, by a Chernoff bound,
at least $\tfrac 12 \alpha\beta n$ proposals with probability at least $\exp(-\alpha\beta n/8)$.

We consider the actual blocking choices made by $\mathcal M$ and modify them bottom-up in a way that only reduces the probability of there being $\tfrac 12 \alpha\beta n$ or more proposals to $m_i$.

Clearly, ${\mathcal M}$ can choose to block the same maximum number of women on every path as it never hurts to block more women (we allow the blocking of women who have already proposed to $m_i$ even though it does not affect the number of proposals to $m_i$).

Consider a deepest block at some node $u$ in the tree,
and suppose $b$ women are blocked at $u$.
Let $v$ be a sibling of $u$.
As this is a deepest block, there will be no blocks
at proper descendants of $u$, and furthermore as there
are the same number of blocks on every path,
$v$ will also have $b$ blocked women.

Observe that if there is no blocking in a subtree, then the probability that a woman makes a proposal to $m_i$ is independent of the outcomes for the other women.
Therefore the correct blocking decision at node $u$ is to block the $b$ women with the highest probabilities of otherwise making a proposal to $m_i$,
which we call their \emph{proposing probabilities}; the same is true at each of its siblings $v$.

Let $x$ be $u$'s parent. Suppose the action at node $x$ concerns woman $\widetilde{w}_x$.
Note that the proposing probability for any woman 
$w\ne \widetilde{w}_x$ is the same at $u$ and $v$ because
the remaining sequence of actions for woman $w$ is the same at nodes $u$ and $v$, and as they are independent of the actions of the other women, they yield the same probability of selecting $m_i$ at some point.

We need to consider a number of cases.

\smallskip
\noindent
{\bf Case 1}. $w$ is blocked at every child of $x$. \\
Then we could equally well block $w$ at node $x$.

\noindent
{\bf Case 2}. At least one woman other than $\widetilde{w}_x$ is blocked at some child of $x$.\\
Each such blocked woman $w$ has the same proposing
probability at each child of $x$.
Therefore by choosing to block the women with the highest proposing
probabilities, we can ensure that at each node either
$\widetilde{w}_x$ plus the same $b-1$ other women
are blocked, or these $b-1$ woman plus the same
additional woman $w'\ne \widetilde{w}_x$ are blocked.
In any event, the blocking of the first $b-1$ women can be moved to $x$.

\noindent
{\bf Case 2.1}. $\widetilde{w}_x$ is not blocked at any child of $x$.\\
Then the remaining identical blocked woman at each child of $x$ can be moved to $x$.

\noindent
{\bf Case 2.2}. $\widetilde{w}_x$ is blocked at some child of $x$ but not at all the children of $x$.\\
Notice that we can avoid blocking $\widetilde{w}_x$  at the child $u$ of $x$ corresponding to selecting $m_i$, 
as the proposing probability for $\widetilde{w}_x$ after it has selected $m_i$ is $0$, so blocking any other women would be at least as good.
Suppose that $w\ne \widetilde{w}_x$ is blocked at node $u$.

Let $v$ be another child of $x$ at which $\widetilde{w}_x$ is blocked.
Necessarily, $p_{v,\widetilde{w}_x}$, the proposing probability for $\widetilde{w}_x$ at node $v$, is at least the proposing probability $p_{v,w}$ for $w$ at node $v$ (for otherwise $w$ would be blocked at node $v$); also, $p_{v,w}$ equals the proposing probability for $w$ at every child of $x$ including $u$; in addition, 
$p_{v,\widetilde{w}_x}$ equals the proposing probability for $\widetilde{w}_x$ at every child of $x$ other than $u$.
It follows that $w$ is blocked at $u$ and $\widetilde{w}_x$ can be blocked at every other child of $x$.
But then blocking $\widetilde{w}_x$ at $x$ only reduces the proposing probability.

Thus in every case one should move the bottommost blocking decisions at a collection of sibling nodes to a single blocking decision at their parent.
\end{proof}

\end{proof}

\section{Making Fewer Proposals}
\label{sec::fewer-props}
We identify a sufficient set of edges that contains all stable matchings, and on which the DA algorithm produces the same outcome as when it runs on the full edge set.

\begin{defn}[Viable edges]
An edge $(m,w)$ is \emph{man-viable} if, according to $m$'s preferences, $w$ is at least as good as the woman he is matched to in the man-pessimal stable match. Woman-viable is defined symmetrically. An edge is \emph{viable} if it is both man and woman-viable. 
$E_{v}$ is the set of all viable edges.
\end{defn}

\begin{lemma}\label{lem::match_occurs}
Running woman-proposing DA with the edge set restricted to $E_v$ and with any superset obtained via loss thresholds, including the full edge set, results in the same stable matching.
\end{lemma}

\begin{proof}
Suppose a new stable matching, $S$, now exists in the restricted edge set: it could not be present when using the full edge set, therefore there must be a blocking edge $(m,w)$ in the full edge set. But neither $m$ nor $w$ would have removed this edge when forming their restricted edge set since for both of them it is better than an edge they did not remove (the edge they are matched with in $S$).

It follows that w.h.p.\ the set of stable matchings is the same when using $E_v$ 
(or any super set of it generated by truncation with larger loss thresholds) and the whole set. Thus woman-proposing DA run on the restricted edge set will yield the same stable matching as on the full edge set.


\end{proof}

\begin{proof} (Of Theorem~\ref{lem::acceptable-suffice}.)
If $\mathcal E$ occurs, the set of acceptable edges contains all the viable edges. 
Furthermore, the acceptable edges are defined by means of loss thresholds.
The result now follows from Lemma~\ref{lem::match_occurs}.
\end{proof}

For some of the very bottommost agents, almost all edges may be acceptable. However, in the bounded derivatives model, with slightly stronger constraints on the derivatives, we also show (see Appendix~\ref{sec::eps-BN-equil-full-match}) the existence of an $\eps$-Bayes-Nash equilibrium in which all but a bottom $\Theta((\ln n/n)^{1/3})$ fraction of agents use only $\Theta(\ln n)$ edges, and all agents propose using at most $\Theta(\ln^2 n)$ edges, with $\eps=O(\ln n/n^{1/3})$.
\section{Numerical Simulations}
\label{sec:simulations}

We present several simulation results which are complementary to our theoretical results.
Throughout this section, we focus on the linear separable model.

\subsection{NRMP Data}

We used NRMP data to motivate some of our choices of parameters for our simulations. The NRMP provides extensive summary data~\cite{NRMP}. We begin by discussing this data.

Over time, the number of positions and applicants has been growing. We mention some numbers for 2021. There were over 38,000 positions available and a little over 42,000 applicants.
The main match using the DA algorithm (modified to allow for couples, who comprise a little over 5\% of the applicants) filled about 95\% of the available positions.
The NRMP also ran an aftermarket, called SOAP, after which about 0.5\% of the positions remained unfilled.

The positions cover many different specialities.
These specialities vary hugely in the number of positions available, with the top 11, all of size at least 1,000, accounting for 75\% of the positions.
In addition, about 75\% of the doctors apply to only one speciality.
We think that as a first approximation, w.r.t.\ the model we are using, it is reasonable to view each speciality as a separate market. Accordingly, we have focused our simulations on markets with 1,000--2,000 positions (though the largest speciality in the NRMP data had over 9,000 positions).

On average, doctors listed 12.5 programs in their preference lists, hospital programs listed 88 doctors, and the average program size was 6.5 (all numbers are approximate). While there is no detailed breakdown of the first two numbers, it is clear they vary considerably over
the individual doctors and hospitals. For our many-to-one simulations we chose to use a fixed size for the hospital programs. Our simulations cause the other two numbers to vary over the individual doctors and programs because the public ratings and private scores are chosen by a random process.

\subsection{Numbers of Available Edges}

The first question we want to answer is how long do the preference lists need to be in order to have a high probability of including all acceptable edges, for all but the bottommost agents?

We chose bottommost to mean the bottom 20\% of the agents, based on where the needed length of the preference lists started to increase in our experiments for $n=\text{1,000--2,000}$.

We ran experiments with $\lambda =0.5,0.67,0.8$,
corresponding to the public rating having respectively equal, twice, and four times the weight of the private scores in their contribution to the utility.
We report the results for $\lambda=0.8$. 
The edge sets were larger for smaller values of $\lambda$, but the results were qualitatively the same.
We generated 100 random markets and determined the smallest value of $L$ that ensured all agents were matched in all 100 markets. $L=0.12$ sufficed.
In Figure~\ref{fig:edges_avg_2000}, we show results by decile of women's rank (top 10\%, second 10\%, etc.),
specifically the average length of the preference list and the average number of edges proposed by a woman in woman-proposing DA, over these 100 randomly generated markets. We also show the max and min values 
over the 100 runs; these can be quite far from the average value.
Note that the min values in Figure~\ref{fig:edges_avg_2000}(a) are close to the max values in Figure~\ref{fig:edges_avg_2000}(b), which suggests that being on the proposing side does not significantly reduce the value of $L$ that the women could use compared to the value the men use.
We also show data for a typical single run in Figure~\ref{fig:edges_typ}.

\begin{figure}[htb]
    \centering
    \subfloat[Number of edges in the acceptable edge set, per woman, by decile; average in blue with circles, minimum in red with stars. ($n=\text{2,000}$, $\lambda=0.8$, $L=0.12$.)]{	\includegraphics[width=0.47\linewidth,height=4cm]{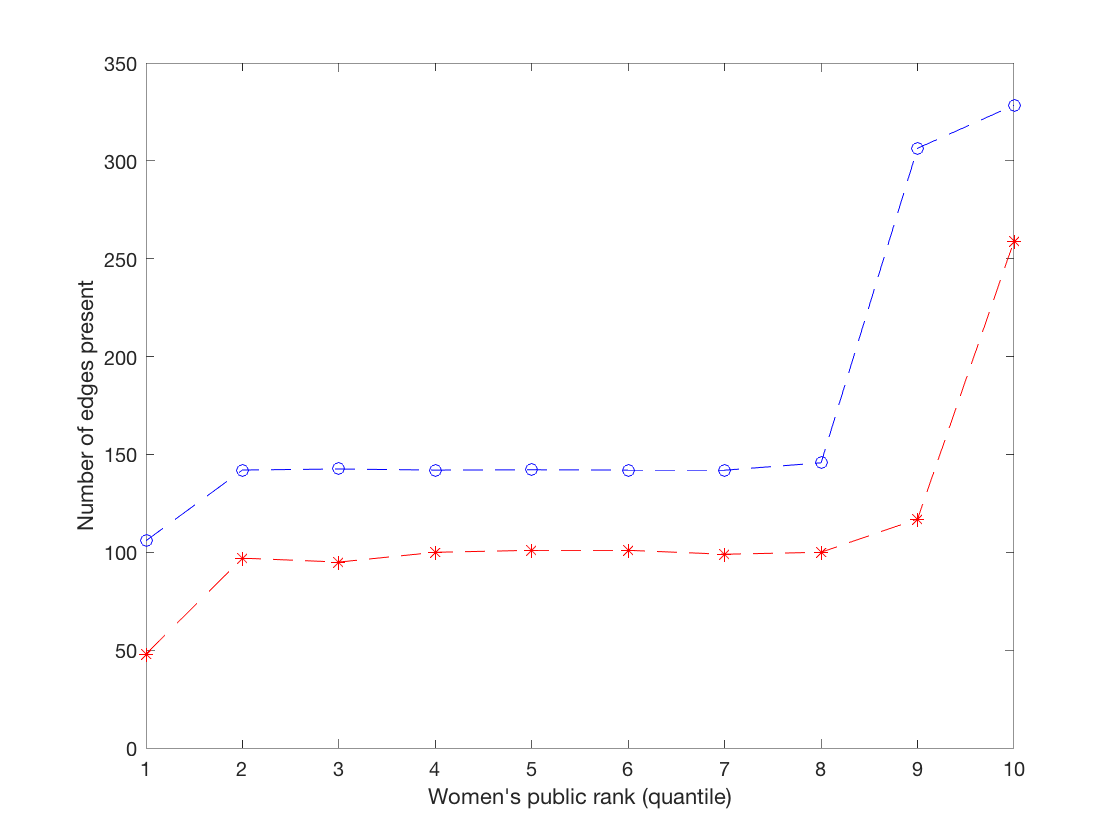}}
    \qquad
    \subfloat[Number of edges in the acceptable edge set proposed during the run of DA, per women, by decile; average in blue with circles, maximum in red with stars.]{\includegraphics[width=0.47\linewidth,height=4cm]{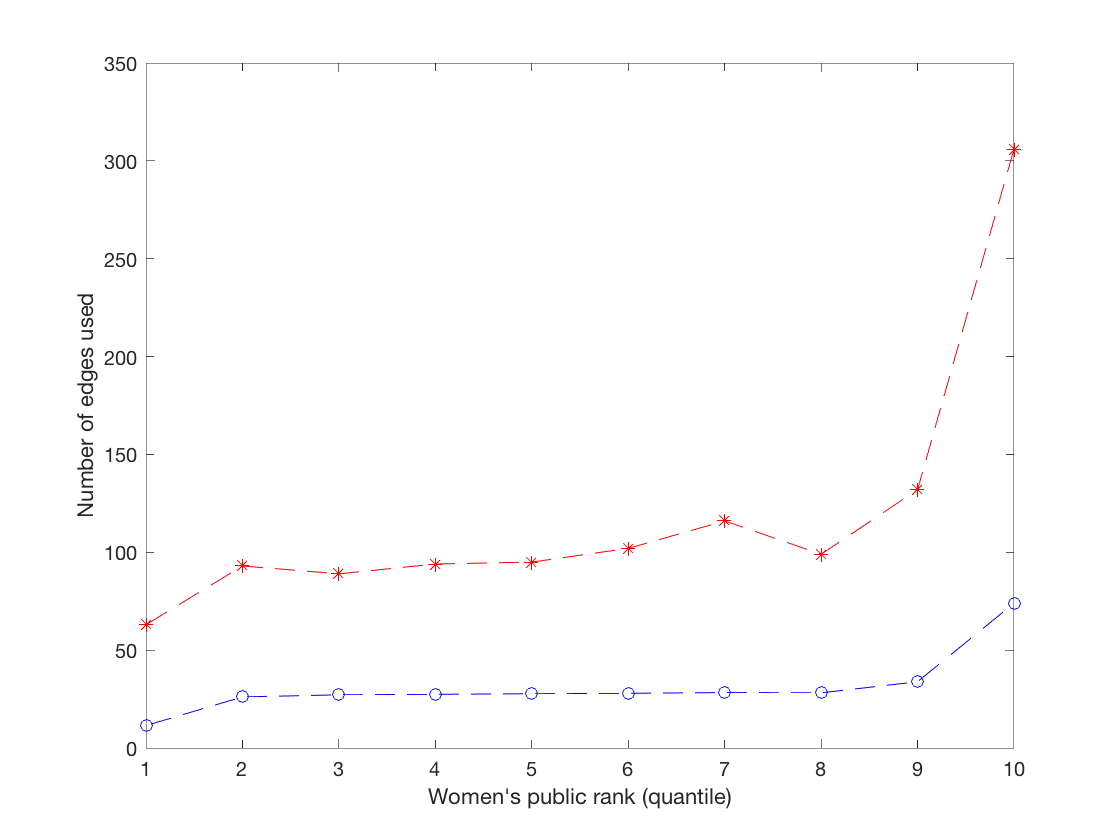}}
    \caption{One-to-one case: summary statistics.}
\label{fig:edges_avg_2000}
\end{figure}

\begin{figure}
    \centering
    \subfloat[Number of edges in the acceptable edge set for each woman.]{	\includegraphics[width=0.47\linewidth,height=4cm]{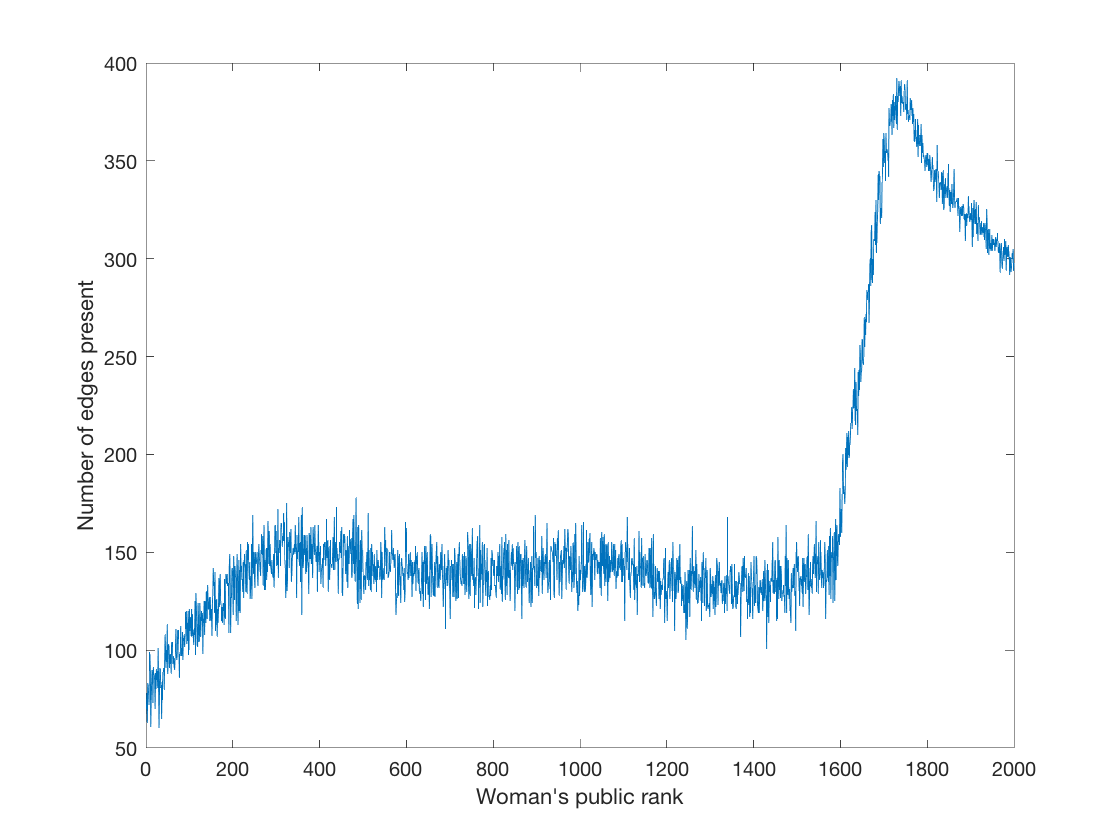}}
    \qquad
    \subfloat[Number of edges in the acceptable edge set  proposed by each woman.
    ]{\includegraphics[width=0.47\linewidth,height=4cm]{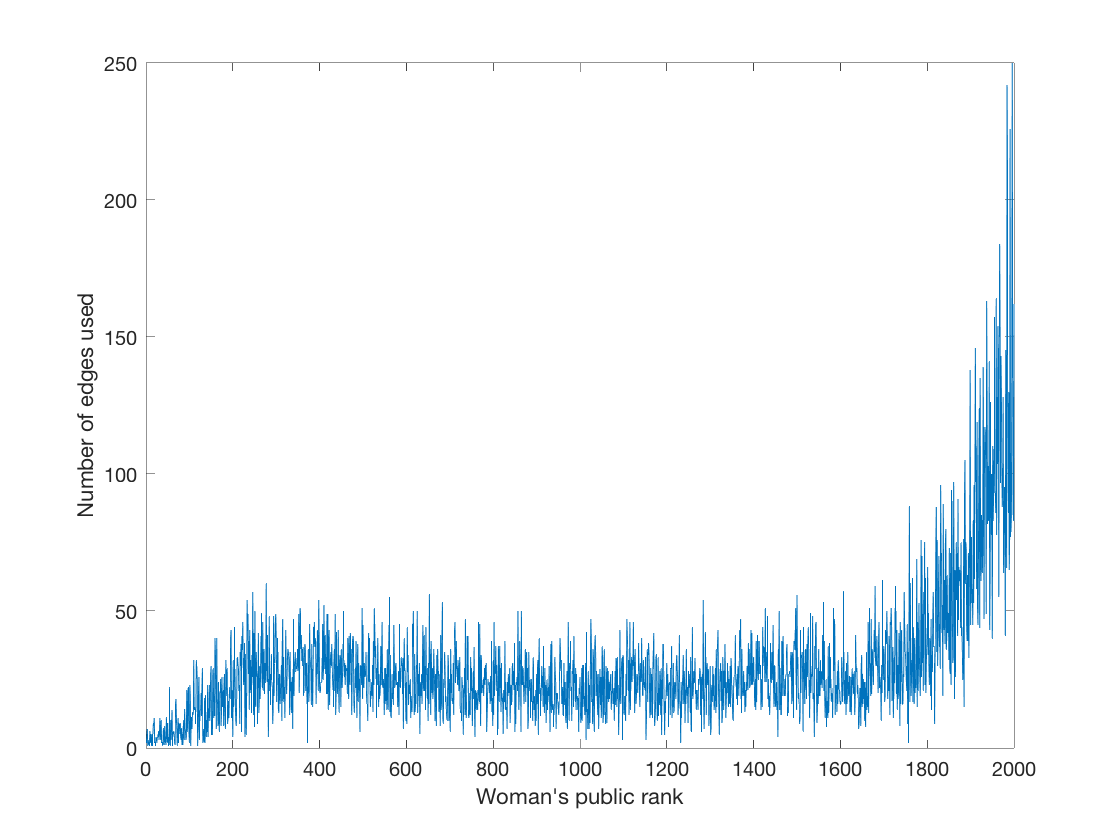}}
    \caption{One-to-one case: a typical run.}
\label{fig:edges_typ}
\end{figure}

We repeated the simulation for the many-to-one setting. In Figure~\ref{fig:edges_typ_8}, we show the results for 2000 workers and 250 companies, each with 8 positions. Now, on average, a typical worker (i.e.\ among the top 80\%) has an average preference list length of 55 and makes 7 proposals.

\begin{figure}
    \centering
    \subfloat[Many to One Setting: Number of edges in the acceptable edge set per worker, by decile; average in blue with circles, minimum in red with stars. ($n_w=\text{2,000}$, $d=8$, $\lambda=0.8$, $L_c = 0.14$, $L_w=0.24$.)]{	\includegraphics[width=0.47\linewidth,height=4cm]{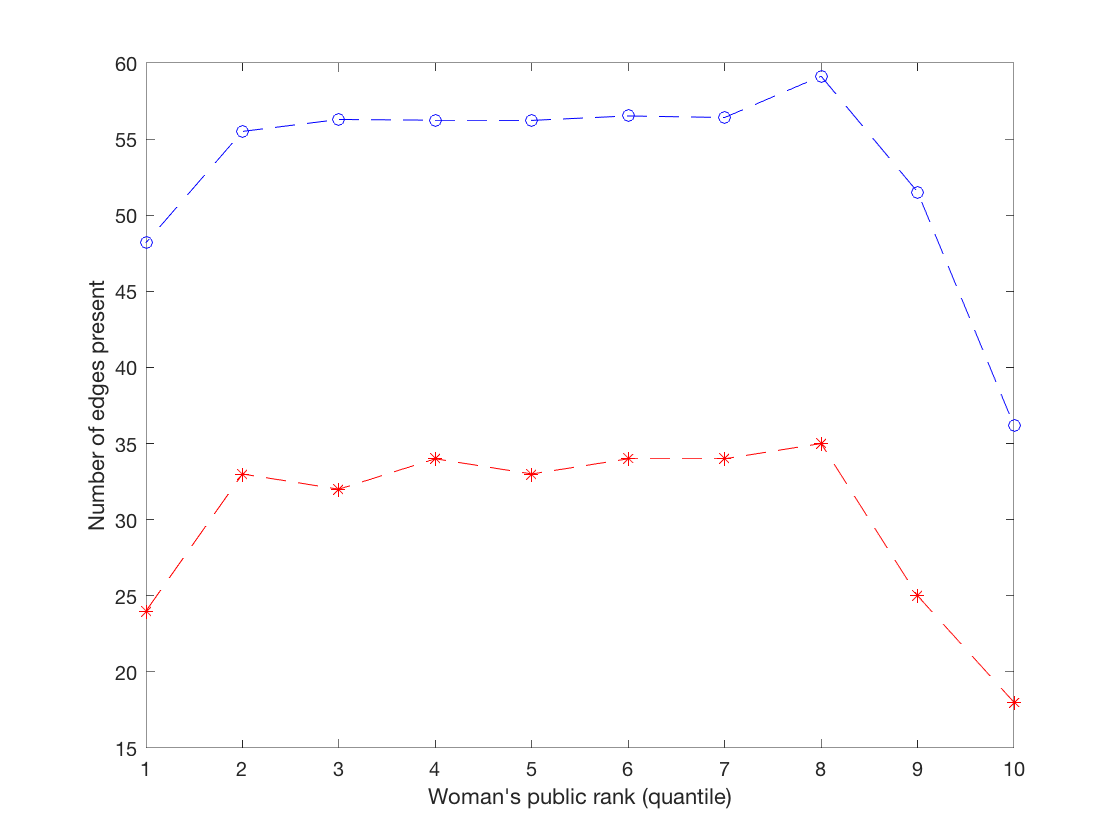}}
    \qquad
    \subfloat[Number of edges in the acceptable edge set proposed during the run of DA, per worker, by decile; average in blue with circles, maximum in red with stars. 
    ]{\includegraphics[width=0.47\linewidth, height=4cm]{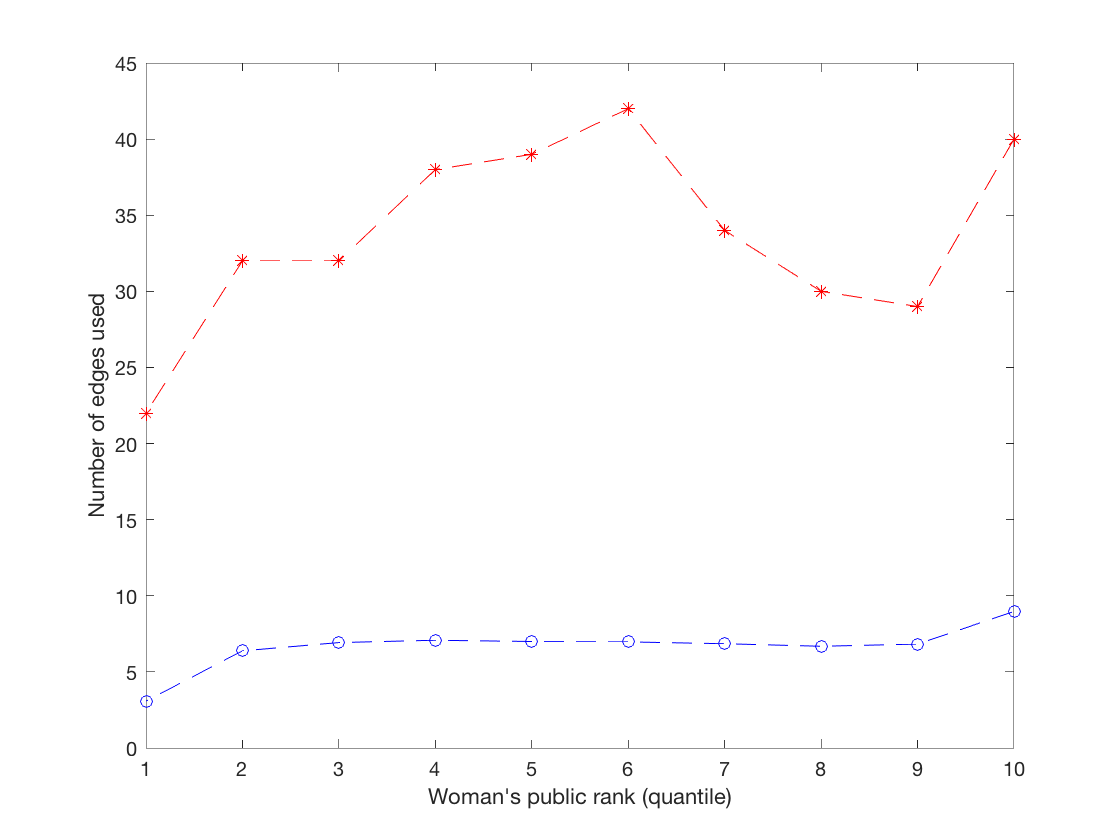}}
    \caption{Many to One Setting.}
\label{fig:edges_typ_8}
\end{figure}
The one-to-one results show that for non-bottommost agents, the preference lists have length 150 on the average, 
while women make 30 proposals on the average (these numbers are slightly approximate). What is going on?
We believe that the most common matches provide a small loss or gain ($\Theta(n^{-1/3})$ in our theoretical bounds) as opposed to the maximum loss possible ($\Theta(n^{-1/3}\ln^{1/3}n)$ in our theoretical bounds),
as is indicated by our distribution bound on the losses
(see item 4 in Section~\ref{sec::extending-results}).
The question then is where do these edges occur in the preference list, and the answer is about one fifth of the way through (for one first has the edges providing a gain, which only go to higher up agents on the opposite side, and then one has the edges providing a loss, and these go both up and down).
However, a few of the women will need to go through most of their list, as indicated by the fact that the max and min lines (for example in Figure \ref{fig:edges_typ_8}) roughly coincide.

This effect can also be seen in the many-to-one experiment but it is even more stark on the worker's side. The reason is that the number of companies with whom a worker $w$ might match which are above $w$, based on their public ratings alone, is $\Theta(L_c n_c)$, while
the number below $w$ is $\Theta(L_w n_c)$, a noticeably larger number. (See Appendix~\ref{sec::cone-size} for a proof of these bounds.)
The net effect is that there are few edges that provide $w$ a gain, and so the low-loss edges, which are the typical matches, are reached even sooner in this setting.

Now we turn to why the number of edges in the available edge set per woman changes at the ends of the range.
There are two factors at work.
The first factor is due to an increasing loss bound as we move toward the bottommost women, which increases the sizes
of their available edge sets.
The second factor is due to public ratings.
For a woman $w$ the range of men's public ratings for its acceptable edges is
$[r_m-\Theta(\Lbar),r_m+\Theta(\Lbar)]$, where $m$ is aligned with $w$. 
But at the ends a portion of this range will be cut off, 
reducing the number of acceptable edges, with the effect more pronounced for low public ratings. 
Because $\lambda=0.8$, initially, as we move to lower ranked women, 
the gain due to increasing the loss bound dominates the loss due to a reduced public rating range, 
but eventually this reverses. Both effects can be clearly seen in Figure~\ref{fig:edges_typ}(a), for example.

\subsection{Unique Stable Partners}

Another interesting aspect of our simulations is that they showed that most agents have a unique stable partner.
This is similar to the situation in the popularity model when there are short preference lists, but here this result appears to hold with full length preference lists.
In Figure~\ref{fig:gap}, we show the outcome on a typical run and averaged over 100 runs, for $n=\text{2,000}$ in the one-to-one setting.
We report the results for the men, but as the setting is symmetric they will be similar for the women.
On the average, among the top 90\% of agents by rank,
0.5\% (10 of 1,800) had more than one stable partner, and among the remainder another 2\% had multiple stable partners (40 of 200).

Also, as suggested by the single run illustrated in Figure~\ref{fig:gap}(a), the pair around public rank 1,600 and the triple between 1,200 and 1,400 have multiple stable partners which they can swap (or exchange via a small cycle of swaps) to switch between different stable matchings. This pattern is typical for the very few men with multiple stable partners outside the bottommost region.

\begin{figure}[htb]
    \centering
    \subfloat[Public ranks of men with multiple stable partners in a typical run.]{	\includegraphics[width=0.47\linewidth]{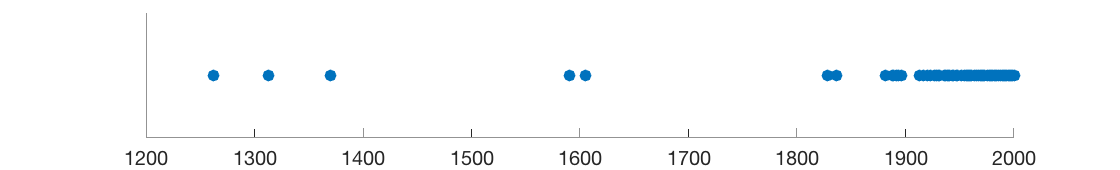}}
    \qquad
    \subfloat[Average numbers of men with multiple stable partners, by decile.]{\includegraphics[width=0.47\linewidth,height=4cm]{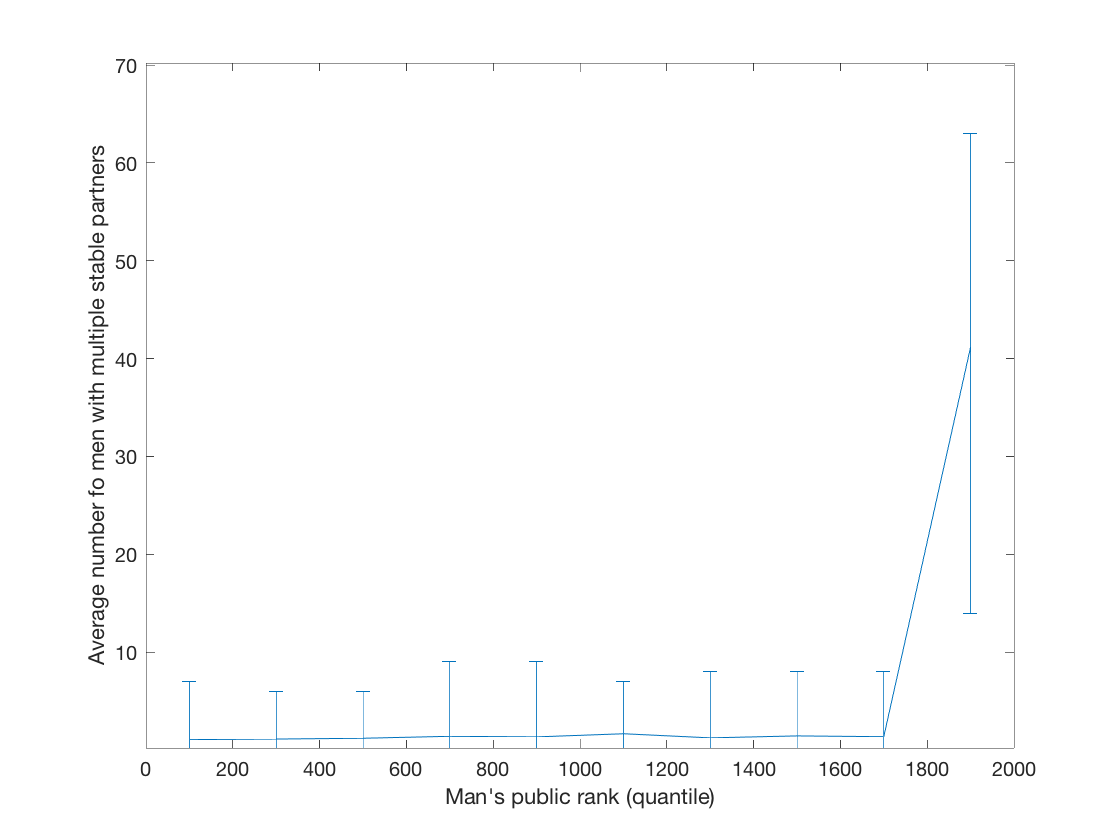}}
    \caption{Unique stable partners, one-to-one setting.}
\label{fig:gap}
\end{figure}

\subsection{Constant Number of Proposals}
\label{sec:const-no-props}

Our many-to-one experiments suggest that the length of the preference lists needed by our model are larger than those observed in the NRMP data. In addition, even though there is a simple rule for identifying these edges, in practice the communication that would be needed to identify these edges may well be excessive. In light of this it is interesting to investigate what can be done when the agents have shorter preference lists.

We simulated a strategy where the workers' preference lists contain only a constant number of edges. We construct an \textit{Interview Edge Set} which contains the edges $(w,c)$ satisfying the following conditions:
\begin{enumerate}
    \item Let $r_w$ and $r_c$ be the public ratings of $w$ and $c$ respectively. Then $|r_w-r_c|\le p$. 
    \item The private score $w$ has for $c$ as well as the private score of $c$ for $w$ are both greater than $q$. 
\end{enumerate}
We choose the parameters $p$ and $q$ so as to have 15 edges per agent on average.  Many combinations of $p$ and $q$ would work. We chose a pair that caused relatively few mismatches.
We then ran worker proposing DA on the Interview Edge Set.

One way of identifying these edges is with the following communication protocol: the workers signal
the companies which meet their criteria (the workers' criteria); the companies then reply to those workers who meet their
criteria. In practice this would be a lot of communication on the workers's side, and therefore it may
be that an unbalanced protocol where the workers use a larger $q_w$ as their private score cutoff and the companies a correspondingly smaller $q_c$ is more plausible. Clearly this will affect the losses each side incurs when there is a match, but we think it will have no effect on the non-match probability, and as non-matches are the main source of losses, we believe our simulation is indicative.
We ran the above experiment with $p=0.19$ and $q=0.60$, with the company capacity being $8$. Figure \ref{fig:constant}(a) shows the locations of unmatched workers in a typical run of this experiment while \ref{fig:constant}(b) shows the average numbers of unmatched workers per quantile (of public ratings) over $100$ runs. We observe that the number of unmatched workers is very low (about 1.5\% of the workers) and most of these are at the bottom of the public rating range. 

Figure \ref{fig:constant}(c) compares the utility obtained by the workers in the match obtained by running worker-proposing DA on the Interview Edge Set to the utility they obtain in the worker-optimal stable match. We observe that only a small number of workers have a significantly worse outcome when restricted to the Interview Edge Set. 

\begin{figure}[h!]
    \centering
    \subfloat[Public ranks of unmatched workers in a typical run.]{	\includegraphics[width=0.29\linewidth]{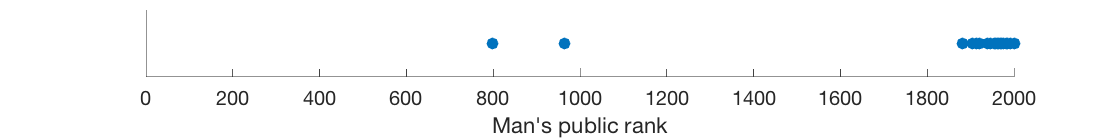}}
    \qquad
    \subfloat[Average numbers of unmatched workers by public rating decile.]{\includegraphics[width=0.29\linewidth]{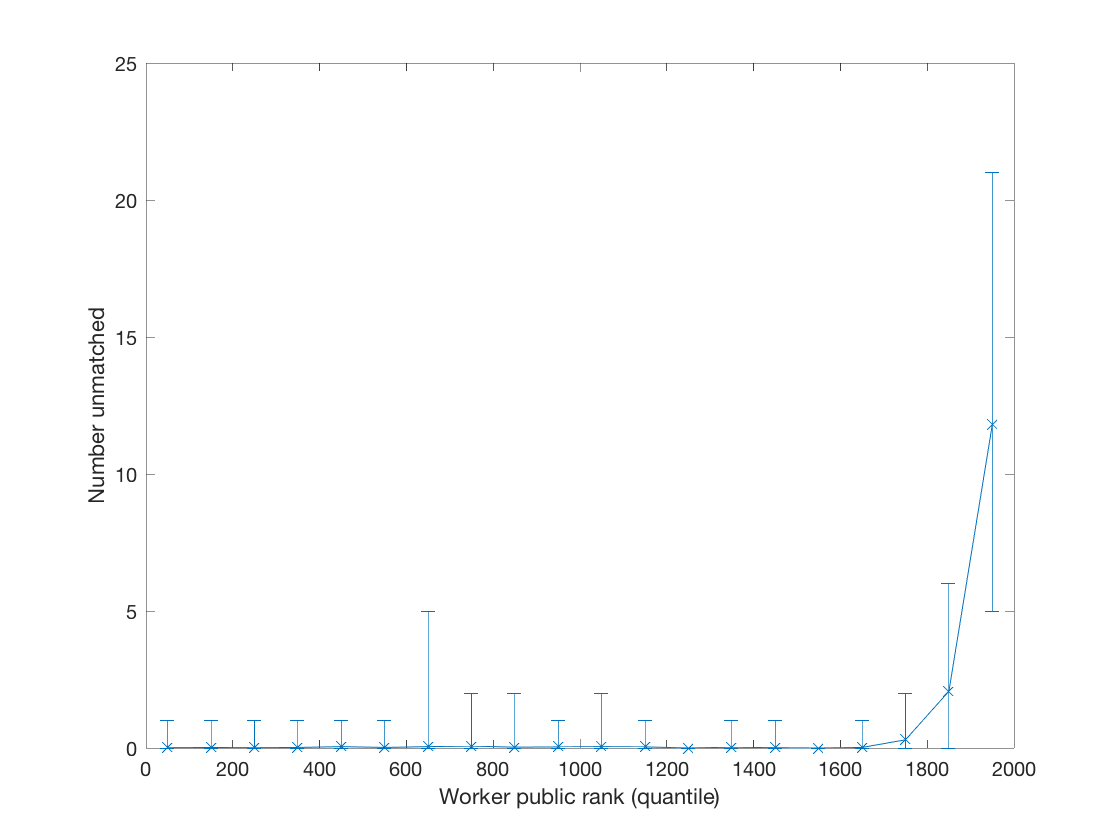}}
    \qquad
    \subfloat[Distribution of workers' utilities with worker-proposing DA: $\text{(full edge set result)}$ $- \text{(Interview edge set result)}$ ]{\includegraphics[width=0.29\linewidth]{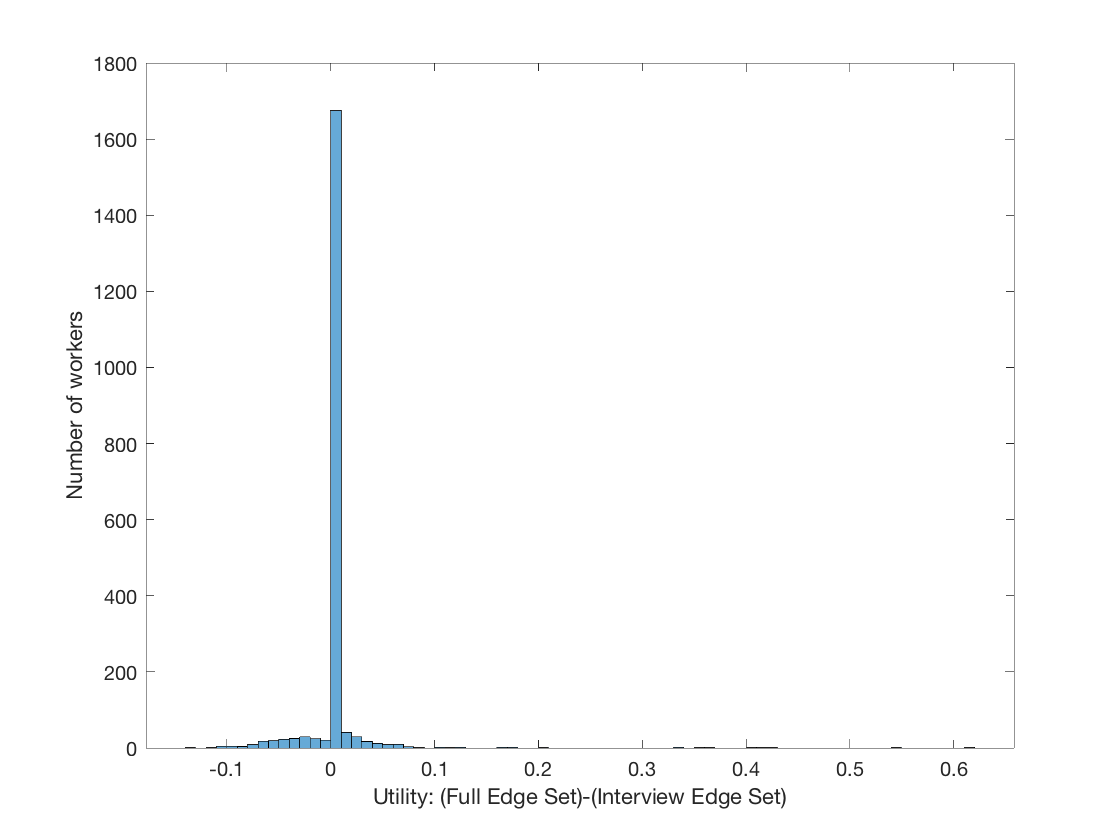}}
    \caption{Constant number of proposals.}
    
\label{fig:constant}
\end{figure}


\section{Discussion and Open Problems}
\label{sec:discussion}

Our work shows that in the bounded derivatives model, apart from a sub-constant fraction of the
agents, each of the other agents has $O(\ln n)$ easily identified edges on their preference list which cover all
their stable matches w.h.p. 

As described 
in Section~\ref{sec:simulations},
our experiments for the one-to-one setting yield a need for what appear to be impractically large preference lists. 
While the results in the many-to-one setting are more promising, even here
the preference lists appear to be on the large side.
Also, while our rule for identifying the edges to include is simple,
in practice it may well require too much communication to identify these edges.
At the same time, our outcome is better than what is achieved in practice: we obtain a complete match
with high probability, whereas in the NRMP setting a small but significant percentage of positions are left unfilled.
Our conclusion is that it remains important to understand how to effectively select smaller sets of edges.

In the popularity model, it is reasonable for each agent to simply select their favorite partners. But in the current setting, which we consider to be more realistic, it would be an ineffective strategy, as it would result in most agents remaining unmatched.
Consequently, we believe the main open issue is to characterize what happens when the number of edges $k$ that an agent can list
is smaller than the size of the allowable edge set.
We conjecture that following a simple protocol for selecting edges to list, such as the one
we use in our experiments (see Section~\ref{sec:const-no-props}), will lead to an $\eps$-Bayes-Nash equilibrium, where $\eps$ is a decreasing function of $k$. Strictly speaking, as the identification of allowable edges requires communication,
we need to consider the possibility of strategic communication,
and so one would need to define a notion of $\eps$-equilibrium akin to a Subgame Perfect equilibrium.
We conjecture that even with this, it would still be an $\eps$-equilibrium.

Finally, it would be interesting to resolve whether the experimentally observed near uniqueness of the stable matching
for non-bottom agents is a property of the linear separable model.
We conjecture that in fact it also holds in the bounded derivatives model.

\newpage
\appendix

\section{Overview of the Appendices}
\label{sec::appendix-overview}
Here we provide an overview of the appendices that follow.

Appendix~\ref{sec::main_lemma_full_proof} provides the omitted proofs of claims from the main body of the paper.
Appendix~\ref{sec::utility_models} defines all the utility models we consider, Appendix~\ref{sec::results} states the remaining results, and Appendix~\ref{sec::analysis} gives sketch proofs of all these results.
Complete analyses are given in the following appendices: The upper bound results are shown in
Appendix~\ref{sec::upper_full_proofs}; they are based on the key ideas involved in proving Theorems~\ref{thm::basic-linear-result} and~\ref{lem::acceptable-suffice}.
The lower bounds are shown in Appendix~\ref{sec::lower_full_proofs}.
The analysis of the $\epsilon$-Nash equilibrium is given in Appendix~\ref{sec::eps-BN-equil-full-match}; this
uses ideas from our analysis of Theorem~\ref{thm::basic-linear-result} as well as several other ideas, resulting in a quite involved proof. 
Finally, further experimental results are described in Appendix~\ref{sec::more_numerics}. 

\section{Missing Proofs}
\label{sec::main_lemma_full_proof}

\begin{proof} (Of Claim~\ref{clm::span-man-interval}.)
We prove the bound for an arbitrary man $m$ with public rating $r_m$.
The expected number $n_x$ of men other than $m$ in $M[r_m-\alpha,r_m]$ is $\alpha(n-1)$.
This bound depends on the independent random choices of the men's public ratings.
Thus, by a Chernoff bound, 
\begin{align*}
    \Pr\big[n_x \ge  \tfrac 32 \alpha(n-1)]  \le \exp(\alpha(n-1)/12).
\end{align*}
Now, we apply a union bound to all $n$ men to obtain the stated result.
\end{proof}

\begin{proof} (Of Claim~\ref{clm::span-woman-interval}.)
We prove the bound for an arbitrary woman $w$ with public rating $r_w\ge 3\alpha$.
The expected number $n_y$ of women other than $w$ in $W[r_w-3\alpha,r_m]$ is $3\alpha(n-1)$.
This bound depends on the independent random choices of the women's public ratings.
Thus, by a Chernoff bound, 
\begin{align*}
    \Pr\big[n_y \le  \tfrac 52 \alpha(n-1)]  \le \exp(\alpha(n-1)/24).
\end{align*}
Now, we apply a union bound to all $n$ women to obtain the stated result.
\hide{

We prove the bound for $w_i$.
Note that $\ell_i =h_i +\alpha n-1$,
and as ${\mathcal B}_1$ does not occur, by
Claim~\ref{clm::span-man-interval},
$h_i \le  \tfrac 32 \alpha (n-1)$.
Next, we bound the rating range $\tau$ such that
$\big|W[r_{w_i}-\tau,r_{w_i}]\big|\ge \ell_i$ with high probability.

For any $\tau>0$, the expected number $n_{w_i}$ of woman other than $w_i$
in $W[r_{w_i}-\tau,r_{w_i}]$ is $\tau(n-1)$.
Thus, by a Chernoff bound, 
\begin{align*}
    \Pr\big[n_{w_i} \le \tau(n-1) - \tfrac 12 \alpha(n-1)\big]  \le \exp(-\alpha^2(n-1)/8\tau).
\end{align*}
We want to ensure $n_{w_i} \ge \ell_i$.
This holds if 
\begin{align*}
  \tau(n-1) - \tfrac 12 \alpha(n-1) \ge   \tfrac 32 \alpha (n-1) +\alpha n-1.
\end{align*}
$\tau = 3\alpha$ suffices.
This yields a probability bound of $\exp(-\alpha(n-1)/24)$.
A union bound over all $n$ women yields the result.
}
\end{proof}

\begin{proof}
(Of Claim~\ref{clm::discrete-dist}.)
We do this in such a way that for each woman the probability of selecting $m_i$ is only increased,
and the probability of having any differences
in the sequence of actions in the original continuous setting and the discrete setting is at most $\delta$.
We detail how to construct this discrete utility space
in the appendix. The space depends on
$\delta$, which can be arbitrarily small. 
For each man $m$ we partition the interval $[V(r_m,0),V(r_m,1)]$ of utilities it can provide
into the following $z$ subintervals: $[V(r_m,0),V(r_m,1/z)),[V(r_m,1/z),V(r_m,2/z)),\ldots,[V(r_m,(z-2)/z),V(r_m,(z-1)/z)),[V(r_m,(z-1)/z),V(r_m,1)]$.
Note that the probability that woman $w$'s edge to $m$ occurs in any one subinterval is $1/z$.
Over all $n$ men this specifies $n(z-1)$ utility values that are partitioning points.
Now, for each man $m$, we partition the interval $[V(r_m,0),V(r_m,1)]$ about all $n(z-1)$ of these points,
creating $n(z-1)+1$ subintervals.
The values at these partition points 
plus the endpoint 
$V(r_m,0)$ are the discrete utilities available to the women for evaluating man $m$, obtained by rounding down her actual utility.

Consider a single interval $I=[V(r_m,a),V(r_m,b))$ and an arbitrary woman $w$.
Let $p^{I,c}_j$ be the probability
that in the original continuous private score setting,
the probability exactly one man $m_j$ provides her a utility in $I$,
let $p^{I,c}_{\text{none}}$ be
the probability no one provides her a utility in $I$, and let 
$\overline{p}^{I,c}$ be the probability that two or more men provide her a utility in $I$.
Note that $\overline{p}^{I,c}\le n(n-1)/2z^2$.
In the discrete setting, we remove the possibility
of making two proposals and increase the probability of selecting man $m_i$ by this amount: the probability of selecting man $m_j\ne m_i$ alone, with private score $a$ will be $p^{I,c}_j$, the
probability of selecting no one will be
$p^{I,c}_{\text{none}}$, while the probability of selecting man $m_i$ with private score $a$ becomes $p^{I,d}_i=p^{I,c}_{i}+\overline{p}^{I,c}$.

Recall that in the run of double-cut DA, each woman repeatedly makes the next highest utility proposal.
We view this as happening as follows. For each successive discrete utility value, woman $w$ has the following choices.

i. she selects some man 
to propose to (among the men $w$ she has not yet proposed to); or 

ii. she takes ``no action''. This corresponds 
to $w$ making no proposal achieving the current utility.
\\
Every run of DA in the continuous setting that does not have a woman selecting two men over the course of a single
utility interval will result in the identical run in the discrete setting in terms of the order in which each
woman proposes to the men.
Thus, the probability
that in the discrete setting $w$'s action in terms of who she selects and in what order differs from her actions in the continuous setting is at most $n^3/2z\triangleq \delta/n$ (because, in each possible computation, $w$ makes at most $nz$ choices, and for each choice the probability difference is at most $n^2/2z^2)$.
Furthermore, the probability of selecting man $m_i$ is only increased.
So over all $n$ women, the probability of anything
changing is at most $\delta$.
Clearly, $\delta$ can be made arbitrarily small.
\end{proof}

\section{More General Models}
\label{sec::utility_models}
\subsection{Utility Models}
\hide{
It is natural to view the men and women's ordered preference lists as arising from the fact that the agents differ regarding their perceived utility values for their possible partners. As already discussed, this topic has been 
explored in the literature, in particular in settings where these utility values are thought of as being drawn, either explicitly or implicitly, from various distributions. 
We now define the utility models and settings we consider in this work.
}

\paragraph*{The General Utilities Model}
There are $n$ men and $n$ women. Each man $m$ has a utility $U_{m,w}$ for the woman $w$, and each woman $w$ has a utility $V_{m,w}$ for the man $w$. These utilities are defined as
\begin{align*}
    &U_{m,w}=U(r_w,s_m(w)), \text{ and}\\
    &V_{m,w}=V(r_m,s_w(m)),
\end{align*}
where $r_m$ and $r_w$ are common public ratings, $s_m(w)$ and $s_w(m)$ are private scores specific to the pair $(m,w)$, and $U(\cdot,\cdot)$ and $V(\cdot,\cdot)$ are continuous and strictly increasing functions from $\mathbb{R}^2_+$ to $\mathbb{R}_+$.

The public ratings and private scores are drawn independently from distributions with positive density functions with bounded support on $\mathbb{R}_+$. We assume without loss of generality that all public ratings and private scores are drawn uniformly and independently from $[0,1]$ since there is always a change of variables that transforms them into uniform draws while transforming the utility functions monotonically.

$U$ and $V$ are not explicitly assumed to be bounded. However they are continuous and for the purpose of our analysis we can restrict the domain of $U$ and $V$ to the product of the bounded supports of our ratings and score distributions. These restricted $U$ and $V$ are continuous functions on a compact set and hence are bounded. Now, WLOG, by scaling appropriately, we can assume the range of $U$ and $V$ are both $[0,1]$.

\paragraph*{The Bounded Derivatives Model}
We add a notion of bounded derivatives to the general utilities model.

\begin{defn} \label{def::bdd-deriv}
A function $f(x,y):{\mathbb R}^2\rightarrow {\mathbb R}^+$ has \emph{$(\rho,\mu)$-bounded derivatives} if for all $(x,y)\in {\mathbb R}^2$,
\begin{align*}
   &\rho \le\pdv{f}{x}\Big/\pdv{f}{y}~~\text{\rm --- the ratio bound};\\ 
&\pdv{f}{x}\leq \mu~~\hspace*{0.32in}\text{\rm --- the first derivative bound}.
\end{align*}
\end{defn}
Note that this definition implies $\pdv{f}{y}$ is upper bounded by $\mu/\rho$.

\smallskip

In the bounded derivatives model, the utility functions $U$ and $V$ are restricted to having $(\rho,\mu)$-bounded derivatives, for some constants $\rho, \mu > 0$. In the linear separable model, which is a special case of this model, $\mu= \lambda$ and $\rho=\lambda/(1-\lambda)$.

\subsection{Other Generalizations}

\paragraph*{Unequal numbers of men and women} We generalize the above models to allow for $n$ women and $p$ men, where $n$ and $p$ need not be equal.
Suppose that $n\le p$. It is then convenient to change the public rating ranges to be $[0,p/n]$ for the men and $[p/n -1,p/n]$ for the women. We proceed symmetrically when $n>p$. We will keep the private score range at $[0,1]$. The effect of this change is to ensure that with
high probability the top $n$ public ratings for the men cover approximately the same range as the women's public ratings.

\paragraph*{Many-to-one matchings} The stable matching problem has also been studied in the setting of many-to-one matchings. For example, in the setting of employees and employers, often employers want to hire multiple employees. For this setting, we will refer to the two sides as 
companies and workers. 
Also, we will focus on the bounded derivatives setting.

There are $n_c$ companies and $n_w$ workers.
Each company has $d$ positions, meaning that it wants to match with $d$ workers.
%
Each worker can be hired by only one company. The total capacity of all the companies exactly matches the number of workers, i.e.\ $n_c\cdot d=n_w$.\footnote{Our results generalize easily to the case in which the number of workers differs from the number of available positions. We omit the details.}

$r_c$ will denote the public rating of company $c$, and $r_w$ the public rating of worker $w$. Worker $w$ has private score $s_w(c)$ for company $c$, and company $c$ has private score $s_c(w)$ for worker $w$. 
$U(r_w,s_c(w))$ denotes the utility company $c$ has for worker $w$, and $V(r_c,s_w(c))$ denotes the utility worker $w$ has for company $c$.

To define the loss in the many-to-one setting, we need to define a non-symmetric notion of alignment of workers and companies.

\begin{defn} [Alignment]
Suppose company $c$ has rank $i$ (as per its public rating). Let $w$ be the worker of rank $d\cdot i$ (also as per its public rating). Then $w$ is aligned with $c$.
Likewise, suppose worker $w'$ has rank $j$.
Let $c'$ be the company with rank $\ceil{j/d}$.
Then $c'$ is aligned with $w'$.
\end{defn}

\begin{defn} [Loss, cont.]
\label{defn::loss-many-to-one} 
Let $c$ be a company and let $w$ be aligned with $c$.
The \emph{loss} $c$ sustains from a match of utility $u$ is defined to be $U(r_w,1) - u$.
Similarly, let $w'$ be a worker and let $c'$ be aligned with $w'$.
The \emph{loss} $w'$ sustains from a match of utility $u$ is defined to be $V(r_{c'},1) - u$.
\end{defn}
\section{Results}
\label{sec::results}
\paragraph*{The Bounded Derivatives Model}
We begin by stating our basic result for this model.
\begin{theorem}\label{thm:basic-bdd-deriv-result}
In the bounded derivatives model, when there are $n$ men and $n$ women,
for any given constant $c>0$, for large enough $n$, with probability at least $1-n^{-c}$, in every stable match,
for every $i$, if $r_{w_i}\ge \sigmabar\triangleq 3\Lbar/4\mu$, 
agent $m_i$ suffers a loss of at most $\Lbar$, where $\Lbar=\Theta((\ln n/n)^{1/3})$,
and similarly for the agents $w_i$.
\end{theorem}
Note that w.h.p., the public ratings of aligned agents are similar.

In words, w.h.p., all but the bottommost agents (those whose aligned agent has public rating less than $\sigmabar$) suffer a loss of no more than $\Lbar$. We call this high probability outcome $\mathcal E$.

By Theorem \ref{lem::acceptable-suffice}, the implication is that w.h.p.\ a woman can safely restrict her proposals to her acceptable edges, or to any overestimate of this set of edges obtained by her setting an upper bound on the loss she will accept from a match.
There is a small probability--- at most $n^{-c}$---that this may result in a less good outcome, namely the probability
that $\mathcal E$ does not occur.

Then,
w.h.p., every stable match gives each woman $w$, whose aligned agent $m$ has public rating $r_m\ge\sigmabar=\Omega((\ln n/n)^{1/3})$, a partner with public rating in the range $[r_m - \Lbar/\mu, r_m+ 5/4 \Lbar/\mu]$ (see appendix \ref{sec::cone-size}). An analogous statement applies to the men. 

This means that if we are running woman-proposing DA, each of these women might as well limit her proposals to her woman-acceptable edges, which is at most the men with public ratings in the range $r_m \pm \Theta(\Lbar)$ for whom she has private scores of at least $1 - \Theta(\Lbar)$. In expectation, this yields $\Theta(n^{1/3}(\ln n)^{2/3})$ men to whom it might be worth proposing. It also implies that a woman can have a gain of at most $\Theta(\Lbar)$ compared to her target utility.

If, in addition, each man can inexpensively signal the women who are man-acceptable to him, then the women can further limit their proposals to just those men providing them with a signal; in the case of accurate signals, this reduces the expected number of proposals these women can usefully make to just $\Theta(\ln n)$.

Our next result provides a distribution bound on the losses. It states that for most agents, the losses are at most
$\Theta(\Lbar/(\ln n)^{1/3})$, with a geometrically decreasing number of agents facing larger losses.

\begin{theorem}\label{thm::distr-bound}
In the bounded derivatives model, when there are $n$ men and $n$ women,
for any given constant $c>0$, for large enough $n$, with probability at least $1-n^{-c}$, in every stable match,
among the agents whose aligned partner has public score at least $\sigmabar\triangleq 3\Lbar/4\mu$, 
at most $2n\cdot\exp(-(c+2)\ln n/2^{3h})$ men suffer a loss of more
than $\Lbar/2^h$, for integer $h$ with
$\frac 13 \log\Big(\frac{(c+2)\ln n}{\ln[n/3(c+2)\ln n]}\Big) \le h \le \frac 13 \log[(c+2)\ln n]$,
and likewise for the women.
\end{theorem}

We now generalize Theorem~\ref{thm:basic-bdd-deriv-result} to possibly unequal numbers of men and women, and also state what can be said for agents with low public ratings.
\begin{theorem}\label{thm::low_loss_bounded}
Suppose there are $p$ men and $w$ woman, with $p\ge n$.
Let $t\ge 1$ be a parameter.
In the bounded derivatives model,
for any given constant $c>0$, for large enough $n$, with probability at least $1-n^{-c}$,
in every stable match, every agent, except possibly the 
men whose aligned agents have public rating less than $\tfrac{p-n}{n}+\tfrac{\sigmabar}{t}$ and the women
whose aligned agents have public rating less than $\tfrac{\sigmabar}{t}$, suffers a loss of at most $\Lbar t^2$, where $\Lbar =\Theta((\ln n/n)^{1/3})$
and $\sigmabar = 3\Lbar/4\mu$.
\end{theorem}
Note that when $L=1$ (i.e.\ 100\% loss), $t=\Theta((n/\ln n)^{1/6})$ and therefore $\sigmabar/t = \Theta((\ln n)/n)^{1/2})$,
providing a lower bound on the range for which this result bounds the loss.

Setting $t=1$ and $p=n$ yields Theorem~\ref{thm:basic-bdd-deriv-result}.

The implication is similar to that for Theorem~\ref{thm:basic-bdd-deriv-result},
but as $t$ increases, i.e., for women whose aligned agents have increasingly low public ratings,
the bound on the number of proposals she can usefully make grows by roughly a $t^2$ factor.

\paragraph*{$\eps$-Bayes-Nash Equilibrium}

\begin{defn} \label{def::strong-bdd-deriv}
A function $f(x,y):{\mathbb R}^2\rightarrow {\mathbb R}^+$ has \emph{$(\rho_{\ell},\rho_u,\mu_{\ell},\mu_u)$-bounded derivatives} if for all $(x,y)\in {\mathbb R}^2$,
\begin{align*}
   &\text{\rm The ratio bound:}~~\rho_{\ell} \le\pdv{f}{x}\Big/\pdv{f}{y}\le \rho_u.\\ 
&\text{\rm The first derivative bound:}~~\mu_{\ell}\le \pdv{f}{x}\leq \mu_u.
\end{align*}
Then $f$ is said to have the \emph{strong bounded derivative} property. Note that in the linearly separable model,
$\rho_{\ell}=\rho_u$ and $\mu_{\ell}=\mu_u$.
\end{defn}

Let $t\geq 1$ be a parameter and $\sigmabar= \Theta([\ln n/n]^{1/3})$. Define. $L^m_t\triangleq U(r_w,1) - U(r_w-\sigmabar t^2,1)$ and
$L^w_t\triangleq V(s_m,1) - V(r_m-\sigmabar t^2,1)$. For this to be meaningful when $r_w-\sigmabar t^2 < 0$, we extend the definition of $U$ to this domain as follows. For $s<0$, $\pdv{U(r,s)}{r}=\mu_{\ell}$ and
$\pdv{U(r,s)}{s}=\rho_{\ell}$.
We proceed analogously to handle the case that
$r_m-\sigmabar t^2 < 0$. 
Define parameters $\sigma_m = \beta/ n^{1/3}$ and $\sigma_w=\nu/n^{1/3}$, where $\beta>1$ and $\nu<1$ are constants. We then define $t_m= \sigmabar/\sigma_m$ and $t_w= \sigmabar/\sigma_w$.
Note that in the strongly bounded derivatives model, $L^m_{t_m} \le \Theta\Big(\frac{\mu_u}{\beta^2}\cdot \frac{\ln n}{n^{1/3}}\Big)$ and $L^w_{t_w} \le \Theta\Big(\frac{\mu_u}{\nu^2}\cdot \frac{\ln n}{n^{1/3}}\Big)$.

\begin{theorem}\label{thm::eq-BN}
Let $\eps=\Theta(1/n^{1/3})$. There are constants $\beta>1$ and $\nu<1$ such that in the strongly bounded derivatives model, there exists an $\eps$-Bayes-Nash equilibrium where, with probability a least $1-n^{c}$, agents with public ratings greater than $\sigmabar$ make at most $\Theta(\ln n)$ proposals and all agents make at most $\Theta(\ln^2 n)$ proposals. Furthermore, in this equilibrium, with probability a least $1-n^{c}$, every man has a loss of at most $L^m_{t_m}$, and every woman $w$ has a loss of at most $L^w_{t_w}$.
\end{theorem}

\paragraph*{The General Utilities Model}
\label{sec::general}
\begin{theorem}\label{thm::low_loss_general}
Let $0<\eps<1$, $0 < \sigma < 1$, and $c>0$ be constants.
In the general utilities model, for large enough $n$,
with probability at least $1-\exp(-\Theta(n))$, 
in every stable matching, every agent, 
except possibly those whose aligned agents have public rating less than $\sigma$, suffers a loss of at most $\epsilon$.
\end{theorem}

Clearly the smaller $\eps$, the smaller the ranges of public ratings and private scores that can yield acceptable proposals; however, there does not appear to be a simple functional relationship between $\eps$ and the sizes of these ranges in this general model.

\paragraph*{The Many to One Setting}
Next, we state our many-to-one result, expressing it in terms of the $n_w$ workers and $n_c$ companies, each having $d$ positions. 
We now have possibly different bounds

\begin{align*}
\Lbar_w &= \Theta([(\ln n_w)/n_c]^{1/3})=\Theta([(d\ln n_w)/n_w]^{1/3}), ~~~~\text{and} \\
   \Lbar_c &=\left\{\begin{array}{ll}
                    \Theta((\max\{d,\ln n_w\}/n_w)^{1/3})\hspace*{0.5in} & 
                   d= O((n_w/\ln n_w)^{2/3})\\
                   \Theta((d\ln n_w)/n_w) & d =\Omega( n_w/\ln n_w)^{2/3})
                 \end{array}
         \right.
\end{align*}
on the losses for non-bottommost workers and companies.
Analogous to the one-to-one case, we define $\sigmabar_c = 3\Lbar_c/4\mu$ and $\sigmabar_w = 3\Lbar_w/4\mu$, the public rating thresholds below which these loss bounds need not hold.

\begin{theorem}
\label{thm::many-to-one}
Let 
$\Lbar_c$, $\Lbar_w$, $\sigmabar_c$ and $\sigmabar_w$ be as defined above.
Suppose that $d=O((n/\ln n)^{2/3})$.
Then, for any given constant $k>0$, with probability at least $1-n^{-k}$,
in every stable match, every company, except possibly those whose aligned agent has public rating less
than $\sigmabar_w$, suffers a loss of at most $\Lbar_c$,
and every worker, except possibly those whose aligned agent has public rating less
than $\sigmabar_c$, suffers a loss of at most $\Lbar_w$.
\end{theorem}

\paragraph*{Lower Bounds}
The next two theorems show that the bounded derivative result is tight in two senses.
First, we show that the bound $\Lbar$ on the loss is tight up to a constant factor. 
\begin{theorem}
\label{thm::Lbar-bound-tight}
In the linear separable model with $\lambda = \tfrac 12$,
if $n\ge 32,000$ and $L= \tfrac 18 (\ln n/n)^{1/3}$,
then with probability at least $\tfrac 14 n^{-1/8}$ there is no perfect matching, let alone stable matching,
in which every agent with public rating $\tfrac 32 L$ or larger suffers a loss of at most $L$. (Here $\mu=\tfrac 12$, so $\tfrac 32 L=3L/4\mu.)$
\end{theorem}

Next, we show that to obtain sub-constant losses in general, one needs constant bounds on the derivatives. 
We first define a notion of a sub-constant function, which we use to specify sub-constant losses.

\begin{defn}[Sub-constant function]
A function $f(x):{\mathbb R}\rightarrow {\mathbb R}^+$ is sub-constant if for every choice of constant $c>0$, there exists an $\xbar$ such that for all $x \ge\xbar$, $f(x) \le c$.
\end{defn}

\begin{theorem}\label{thm:constant-deriv-bdd-needed}
Let $f:N\rightarrow {\mathbb R}^+$ 
be a continuous, strictly decreasing sub-constant function,
and let $\delta,\sigma \in (0,1)$ be constants. 
Then, in the following two cases, there exist continuous and strictly increasing utility functions $U(.,.)$ and $V(.,.)$ such that for some $\nbar>0$,
for all $n\ge \nbar$, with probability at least $1 - \delta$, in every perfect matching, some rank $i$ man $m_i$ or woman $w_i$ with public rating at least $\sigma$ receives utility less than $U(r_{w_i},1)-f(n)$ or $V(r_{m_i},1)-f(n)$, respectively.

\smallskip

i.  $U(.,.)$ and $V(.,.)$ have derivatives w.r.t.\ their second variables that are bounded by a constant, but for (at least) one of which the derivative w.r.t.\ their first variable is not bounded by any constant.

\smallskip

ii. $U(.,.)$ and $V(.,.)$ 
have derivatives w.r.t.\ their first variables that are bounded by a constant, but for (at least) one of which the derivative w.r.t.\ their second variable is not bounded by any constant.
\end{theorem}
\section{Proof Sketches for the Remaining Results}
\label{sec::analysis}

In Section \ref{sec::linear}, we proved Theorem \ref{thm:basic-bdd-deriv-result} for the special case of the linear separable model with $\lambda=1/2$. We will now briefly outline how we extend the analysis to the bounded derivative model and the general utilities model, as well as to the case where the number of men and women is unequal and the setting of many-to-one matchings.
The full analyses can be found in Appendix~\ref{sec::upper_full_proofs}.

We will also briefly discuss our construction of an $\epsilon$-Bayes-Nash equilibrium in the bounded derivatives model as well as sketch our lower bound proofs in both the bounded derivatives and the general utility models.
The full proofs can be found in Appendices~\ref{sec::eps-BN-equil-full-match} and~\ref{sec::lower_full_proofs},
respectively.

\subsection{Extending the Upper Bound Result}
\label{sec::extending-results}

\noindent
1. Weaker bounds on the losses for agents with lower ranks. \\
This is obtained by reducing $\alpha$ to $\alpha/t$, where $t>1$, and replacing $\Lbar$ by $L=  4\alpha t^2$. The only change occurs in recalculating the loss probability.

\smallskip
\noindent
2. Unequal numbers of men and women. \\
The critical condition for the bound on $m_i$'s loss is $r_{w_i}\ge 3\alpha$. This simply states that there is a range
of $3\alpha$ ratings below $w_i$. But this statement is independent of how many agents there are on each side.
Similarly, the bound on $w_i$'s loss requires that 
there be a range
of $3\alpha$ ratings below $m_i$.
So all one has to do is rephrase these conditions
in terms of $p$ and $n$, the numbers of men and women, respectively.

\smallskip
\noindent
3. The bounded derivatives model. \\
It suffices to scale the values of $\alpha$, $\beta$, $\gamma$ and $L$ to take account of the bounded derivative property so as to ensure that Equations \eqref{eqn::beta-constraint} and \eqref{eqn::gamma-constraint} still hold.
As we shall see, setting $\beta=\alpha \rho$, $\gamma=\alpha\rho$ and $L= 4\alpha\mu$ suffices.

\smallskip\noindent
4. The many-to-one result.\\
We actually analyze the many-to-many setting. The main issue is that a company (replacing a man in the previous argument) seeks $d_c$ matches rather than 1 and a worker seeks $d_w$ matches. 
We need to restate Lemma~\ref{lem::key_bounded}, for now the alignment we seek is between positions sought by the workers
and provided by the companies, rather than between men and women.

However, the significant change occurs in deducing the theorem, for now we need to determine the probability that a company receives $d_c$ matches. The remaining changes are due to replacing $n$, the number of men and of women, with $n_c$ and $n_w$, the numbers of companies and workers, respectively.

\smallskip\noindent
5. A distribution bound on the losses.\\
By reducing both $\alpha$ and $L$ by a factor $s>1$,
we increase the failure probability for a single agent from $n^{-(c+1)}$ to $n^{-(c+1)/s^3}$.
This implies, for example, that in expectation half the
agents have a loss of $O(1/n^{1/3})$. In fact an analysis along the lines of observation (3) in the sketch proof shows that this bound holds with high probability.


\subsection{Extensions to More General Models}
1. The bounded derivative setting.\\
The only places we use the bounds on the derivatives are to determine $\beta$, $\gamma$, and $L$ satisfying \eqref{eqn::beta-constraint} and \eqref{eqn::gamma-constraint}.
As we show in Appendix~\ref{sec::general},
$\beta=\gamma=\alpha\rho$, and $L=4\alpha/\mu$ suffice.

\smallskip\noindent
2. The general utilities model.\\
Given constants $\eps, \sigma>0$, we need to choose $\alpha, \beta, \gamma>0$ and $n$ 
large enough so that \eqref{eqn::beta-constraint} and  \eqref{eqn::gamma-constraint} are satisfied.
The existence of such constant valued $\alpha$, $\beta$, and $\gamma$ follows using the fact that $U$ and $V$, the utility functions, are continuous and strictly increasing.

\subsection{$\eps$-Bayes-Nash Equilibrium}
In the bounded derivative model, with slightly stronger constraints on the derivatives, we also show the existence of an $\epsilon$-Bayes-Nash equilibrium in which agents make relatively few proposals. Specifically, there is an equilibrium in which no agent proposes more than $O(\ln^2 n)$ times and all but the bottommost $O((\ln n/n)^{1/3})$ fraction of the agents make only $O(\ln n)$ proposals. Here $\eps=\Theta(\ln n/n^{1/3})$.

We use the idea of considering a run of DA with cuts just as in the proof of Theorem~\ref{thm:basic-bdd-deriv-result}; in addition, the proposal receiving side will impose reservation thresholds based on their public rank. We also apply the distribution bound on losses described in (4) in the previous subsection. The resulting analysis is somewhat involved (see Appendix~\ref{sec::eps-BN-equil-full-match}).

\subsection{Lower Bounds}

\smallskip\noindent
1. The lower bound complementing the one-to-one upper bound.\\
The main idea is to show by a direct computation that for each woman, with probability at least $1/n^{1/8}$,
all her incident edges provide a loss of more than $L$ to either her or her partner.
We will need to exclude some low-probability events in which the number of agents in an interval is far from
its expectation, and also eliminate the agents with public ratings less than $\tfrac 32 L$.
The net effect is that with probability at least $1/4n^{1/8}$ some woman has no incident $(L,\tfrac32 L)$-acceptable edge, where $L=\tfrac 18(\ln n/n)^{1/3}$,
and hence with this probability there is no matching using solely $(L,\tfrac32 L)$-acceptable edges.
Consequently, in order to obtain a stable matching with high probability, we need to increase the value of $L$.

\smallskip\noindent
2. The lower bound complementing the general utilities model upper bound.\\
To show that no sub-constant loss bound (such as $(\ln n/n)^{1/3}$) is possible,
we consider a loss bound that is shrinking (slowly) as a function of $n$.
For a given $n$, this can be expressed as a loss bound $\eps(n)$.
We provide two similar constructions as there are two separate derivative bounds. 

Our first construction uses a utility function $U(s,v)=\tfrac12(s+g(v))$, with $g(1)=1$ and $g(\cdot)$ being
unboundedly rapidly growing as $v\ra 1$.
$g$ is designed to ensure that with high probability the edges to the women with public ratings $s\ge 1-2\eps(n)$ 
all have private scores less than $1-\eps(n)$.
This will ensure that with high probability $m_1$, the man with the highest public ranking, will have no edge providing him a loss of at most $\eps(n)$.
However slowly $\eps(n)$ decreases as a function of $n$, we show that we can construct
a corresponding $g$ that grows suitably quickly.
This construction demonstrates that the parameter $\epsilon$ needs to be constant. Notice that our construction actually shows that, in the general setting, w.h.p, there is not only no stable matching where all high public rating agents face sub-constant losses, but in fact no perfect matching. 
\section{Proofs of the Remaining Upper Bound Results}
\label{sec::upper_full_proofs}
\begin{proof}
(Of Theorem~\ref{thm:basic-bdd-deriv-result}.)
We now consider what changes occur when we are no longer restricted to the linear separable model with $\lambda=\tfrac 12$.

First, we need to determine the values for $\beta$, $\gamma$ and
$\Lbar$ implied by the bounded derivative parameters $\rho$ and $\mu$. 
We show the following values, 
$\beta=\alpha \rho$, $\gamma=\alpha\rho$ and $L= 4\alpha\mu$, satisfy \eqref{eqn::beta-constraint} and
\eqref{eqn::gamma-constraint}.

For by the definition of $\rho$,
$V(r-\alpha,1) \le V(r,1-\alpha\rho)=V(r,1-\beta)$, satisfying
\eqref{eqn::beta-constraint}.
And by the definition of $\rho$ and $\mu$,
$U(r,1)-U(r-3\alpha,1-\gamma) \le U(r,1)-U(r-3\alpha-\gamma/\rho,1)\le (3\alpha+\gamma/\rho)\mu=4\alpha\mu=L$,
satisfying \eqref{eqn::gamma-constraint}.

To complete the argument, it suffices to determine the failure probability on setting $L=\Lbar$ when running the double-cut DA.
Recall that the failure probability (summed over the $2n$ men and women) is given by:
\begin{align*}
    p_f&=2n\cdot \exp(-\alpha(n-1)/12)
+2n\cdot \exp(-\alpha (n-1)/24)
    +2n\exp(-\alpha\beta n/8)
    +2n\cdot \exp(-\alpha\beta\gamma n/2)\\
    &\le 2n\cdot \exp(-\alpha(n-1)/12)
+2n\cdot \exp(-\alpha (n-1)/24)
+2n\exp(-\alpha^2\rho n/8)
    +2n\cdot \exp(-\alpha^3\rho^2 n/2).
\end{align*}
We note that $\alpha = \Lbar/4\mu$,
and set $\Lbar=[128(c+2)\mu^3\ln n/(\rho^2 n)]^{1/3}$. 
For large enough $n$ this ensures a failure
probability of at most $n^{-c}$.
\end{proof}

\begin{proof} (Of Theorem~\ref{thm::low_loss_bounded}.)
We now need to consider smaller intervals of men and women below $m_i$ and $w_i$ respectively.

We set $\sigma=\sigmabar/t$, where $t\ge 1$.
We then set $\alpha=\sigma/4$ and $\beta=\alpha \rho$ as before, but to keep the most significant term in the
probability bound unchanged ($2n\cdot \exp(-\alpha\beta\gamma n/2)$), we increase $\gamma$ by a factor of $t^2$.
We also set $L=\Lbar t^2$.

The failure probability continues to be at most $n^{-c}$ for large enough $n$ so long as $\alpha\beta n=\Omega(c\ln n)$; this holds for $\sigma=\Omega((\ln n)/n)^{1/2}$.

Now let's consider what happens when there are $p$ men
and $n$ women. We start with the case $p\ge n$.
Our key lemma is stated w.r.t.\ the rank $i$ man $m_i$ and the rank $i$ woman $w_i$, and requires $r_{w_i}\ge 3\alpha$ when the bottom of the rating range is 0 for both men and women. 

It is convenient to have the range of ratings for the men be $[0,p/n]$ and for the women be $[p/n-1,p/n]$.
The effect is that the expected values for $r_{w_i}$ and $r_{m_i}$ are equal.
The condition for $m_i$ to have a loss of at most $L$ becomes $r_{w_i} \ge \tfrac pn +3\alpha$ (i.e.\ $w_i$ has a rating at least $3\alpha$ greater than the bottommost possible rating for the women).
But the condition for $w_i$ to have a loss of at most $L$ remains
$r_{m_i}\ge 3\alpha$ (i.e.\ $m_i$ has a rating at least $3\alpha$ greater than the bottommost possible rating for the men).

Symmetric bounds apply when $n\ge p$.
\end{proof}

\begin{proof} (Of Theorem~\ref{thm::low_loss_general})
We set $\alpha= \sigma/3$ and $L=\eps$.
Again, we need to satisfy \eqref{eqn::beta-constraint} and \eqref{eqn::gamma-constraint}.

To define $\beta$
we begin by specifying a parameter $\beta(r,\alpha)$.
There are two cases.
If $V(r-\alpha,1) \le V(r,0)$, then $\beta(r,\alpha) = 1$.
Otherwise, as $V$ is continuous and strictly increasing,
there must be a value $\beta(r,\alpha)>0$ such that $V(r-\alpha,1)=V(r,1-\beta(r,\alpha))$.
Now, we define $\beta= \min_{r\in[\alpha,1]}\{\beta(r,\alpha)\}$.
As this is the minimum of strictly positive values on a compact set, it follows that $\beta> 0$, also.
Note that $V(r-\alpha,1) \le V(r,1-\beta)$ for all $r\in [\alpha,1]$, satisfying \eqref{eqn::beta-constraint}. 
Also, $\beta= \Theta(1)$ if $\alpha = \Theta(1)$.

Similarly, to define $\gamma$
we begin by specifying a parameter $\gamma(r,\alpha,\eps)$.
Again, there are two cases.
If $U(r,1) - U(r-3\alpha,0) \le L = \eps$ then $\gamma(r,\alpha,\eps) =1$.
Otherwise, as $U$ is continuous and strictly increasing,
there must be a value $\gamma(r,\alpha,\eps)>0$ such that
$U(r,1) - U(r-3\alpha,1-\gamma(r,\alpha,\eps)) =  \eps$.
Now, we define $\gamma\triangleq \min_{r\in[3\alpha,1]}\{\gamma(r,\alpha,\eps)\}$.
Again, as this is a minimum of strictly positive values on a compact set, $\gamma > 0$ also.
Note that $U(r,1) - \eps \le U(r-3\alpha,1-\gamma)$ for all $r\in [3\alpha,1]$,
satisfying \eqref{eqn::gamma-constraint}. Also, $\gamma= \Theta(1)$ if $\alpha = \Theta(1)$.

As $\sigma=\Theta(1)$, all of $\alpha,\beta,\gamma=\Theta(1)$.
Therefore, by Lemma~\ref{lem::key_bounded}, the failure probability is
$\exp(-\Theta(n))$.
\end{proof}

Because the many-to-one setting is non-symmetric it is actually simpler to analyze the
many-to-many setting, many-to-one being just a special case of this.
We will use the terminology of workers and companies, for want of a better alternative.
(One could think of these workers as being consultants or gig workers who seek multiple tasks at a time.)

In this setting there are $n_c$ companies $c_1,c_2,\ldots,c_{n_c}$,  and $n_w$ workers, $w_1,w_2,\ldots,w_{n_w}$, both ordered by their public ranks. Each company has $d_c$ tasks, and each worker desires $d_w$ tasks. 
For simplicity, we suppose $n_c\cdot d_c = n_w\cdot d_w \triangleq n$. 
We let $\nmax = \max\{n_c,n_w\}$.
There will be two loss parameters, $\Lbar_c$, for the companies, and $\Lbar_w$, for the workers. 
Finally, we use the notation $C(I)$ and $W(I)$, where $I$ is an interval of public ratings, to denote, respectively, the companies and workers with public ratings in the interval $I$.

\begin{defn} [Alignment]
Suppose company $c$ has rank $i$. Let $w$ be the worker with rank $\ceil{d_c\cdot i/d_w}$. Then $w$ is aligned with $c$.
Likewise, suppose worker $w'$ has rank $j$.
Let $c'$ be the company with rank $\ceil{d_w\cdot j/d_c}$.
Then $c'$ is aligned with $w'$.
\end{defn}

\begin{defn}[company-acceptable edges]
Let $0 < \sigma_c, \sigma_w < 1$, $0<\Lbar_c, \Lbar_w<1$ be parameters.
An edge $(c,w)$ is company-acceptable if either $c \in C[0,\sigma_c)$, or the utility $c$ gets from this match is at least $U(r_{w'},1)- \Lbar_c$, where $w'$ is the worker aligned with $c$. Worker-acceptability requires either $c \in [0,\sigma_w)$, or utility at least $V(r_{c'},1)- \Lbar_w$, where $c'$ is the company aligned with $w$. An edge is acceptable if it is both company and worker-acceptable. (Strictly speaking, the definition is w.r.t.\ the four parameters $\sigma_c$, $\sigma_w$, $\Lbar_c$, and $\Lbar_w$, but for the sake of readability, we omit them from the terms company- and worker-acceptable.)
\end{defn}

\begin{defn}[DA stops]
The workers \emph{stop at public rating $r$} if in each worker's preference list all the edges with utility 
less than $V(r,1)$ are removed. 
The workers \emph{stop at company $c$} if in each worker's preference list all the edges following their edge to $c$ are removed.
The workers \emph{double cut at $c$ and public rating $r$}, if they each stop at $c$ or $r$, whichever comes first.
Companies stopping and double cutting are defined similarly.
\end{defn}

\begin{theorem}
\label{thm::many-to-many}
Suppose that $d_w/d_c=O((n_w/\ln \nmax)^{2/3})$. Then,
in the bounded derivatives model,
for any given constant $k>0$, with probability at least $1-n^{-k}$,
in every stable matching, every company $c_i$, for which the aligned worker $w_{j_i}$ has public rating
at least $\sigma_w= \Theta(\Lbar_w)$, suffers loss at most
\begin{align*}
    \Lbar_c =\left\{\begin{array}{ll}
                    \Theta([(\ln \nmax)/n_w]^{1/3})\hspace*{0.5in} & d_c=O(\ln \nmax)\\
                    \Theta([d_c/n_w]^{1/3})& \ln\nmax \le d_c= O(d_w (n_w/\ln \nmax)^{2/3})\\
                   \Theta(d_c\ln \nmax/[d_w n_w]) & d_c =\Omega(d_w (n_w/\ln \nmax)^{2/3})
                 \end{array}
         \right.
\end{align*}
and a corresponding symmetric bound for the workers' loss.

\end{theorem}
\begin{proof} 
We need to take account of the fact that each company
seeks to fill $d_c$ positions and each worker seeks $d_w$ positions.
So we slightly redefine the double-cut DA to state that each worker who is not fully matched, i.e., who has fewer than $d_w$ matches, keeps trying to match, stopping when she runs out of proposals, or she is fully matched,
or her next proposal is to $c_i$.

First, to avoid rounding issues, we assume $\alpha$ is chosen so that $\alpha (n_w-\tfrac52)$ is an integer for the argument bounding $\Lbar_c$, and similarly $\alpha (n_c-\tfrac52)$ is an integer for the argument bounding $\Lbar_w$.

We introduce one more index: $j_i = \ceil{d_c\cdot i/d_w}$. 
We then define $\ell_i=\floor{d_c(i +h_i -1)/d_w} +\alpha (n_w-\tfrac 52) - j_i$ (this is where we use the assumption that $\alpha (n_w-\tfrac52)$ is an integer as $\ell_i$ has to be an integer).
This will ensure that after running the double cut DA, the number of not fully matched workers is at least
$\alpha (n_w-\tfrac 52)$.
To see this, note that the number of available positions is $d_c(i+h_i-1)$ (remember $c_i$ is not matched in Step 1); therefore, the number of fully matched workers
is at most $\floor{d_c(i+h_i-1)/d_w}$ 
and therefore the number of not-fully matched workers is at least $\ell_i+j_i - \floor{d_c(i+h_i-1)/d_w} \ge
\floor{d_c(i +h_i -1)/d_w} +\alpha (n_w-\tfrac 52) - \floor{d_c(i+h_i-1)/d_w} \ge \alpha (n_w-\tfrac 52)$.

We need to make small changes to Claims~\ref{clm::span-man-interval}--\ref{clm::one-man-accept-prop} and to their proofs.
It seems simplest to restate and, as necessary, reprove the claims.

\begin{claim}
\label{clm::span-man-interval-mtm}
Let ${\mathcal B}_1$ be the event that for some $i$,
$h_i=\big|C[r_{c_i}-\alpha,r_{c_i})\big| \ge  \tfrac 32 \alpha (n_c-1)$.
${\mathcal B}_1$ occurs with probability at most $n_c\cdot \exp(-\alpha(n_c-1)/12)$.
The only randomness used in the proof are the choices of the companies' public ratings.
An analogous bound applies to the workers.
\end{claim}
Its proof is unchanged.
We just replace $n$ by $n_c$.

\begin{claim}
\label{clm::span-woman-interval-mtm}
Let ${\mathcal B}_2$ be the event that for some $i$, $\ell_i=\big|W[r_{w_i}-3\alpha,r_{w_i})\big| \le \tfrac 52 \alpha (n_w-1)$.
Then ${\mathcal B}_2$ occurs with probability at most $n_w\cdot \exp(-\alpha (n_w-1)/24)$.
%
The only randomness used in the proof are the choices of the workers' public ratings.
An analogous bound applies to the companies.
\end{claim}
Its proof is unchanged.
We just replace $n$ by $n_w$.
\hide{
\begin{proof}
We prove the bound for an arbitrary worker $w$.
Recall that $\ell_i =\floor{d_c(i+h_i-1)/d_w}+\alpha (n_w-\tfrac 52) -j_i$, and $j_i =\ceil{d_c\cdot i/d_w}$.

As ${\mathcal B}_1$ does not occur, by
Claim~\ref{clm::span-man-interval-mtm},
$h_i \le  \tfrac 32 \alpha (n_c-1)$.
Next, we bound the rating range $\tau$ such that
$\big|W[r_{w_{j_i}}-\tau,r_{w_{j_i}}]\setminus\{w_{j_i}\}\big|\ge \ell_i$ with high probability.

For any $\tau>0$, the expected number $N_w$ of workers other than $w$
in $W[r_{w_{j_i}}-\tau,r_{w_{j_i}}]\setminus\{w_{j_i}\}$ is $\tau(n_w-1)$.
Thus, by a Chernoff bound, 
\begin{align*}
    \Pr\big[N_w \le \tau(n_w-1) - \tfrac 12 \alpha(n_w-1)\big]  \le \exp(-\alpha^2(n_w-1)/8\tau).
\end{align*}
We want to ensure $N_w \ge \ell_i$.
This holds if 
\begin{align*}
  \tau(n_w-1) - \tfrac 12 \alpha(n_w-1) \ge   \floor{d_c(i+h_i-1)/d_w}+\alpha (n_w-\tfrac 52) -j_i.
\end{align*}
$\tau(n_w-1) \ge \tfrac 12 \alpha(n_w-1) + \tfrac 32 \alpha (n_c-1)d_c/d_w  +\alpha (n_w-\tfrac 52)$, suffices, and therefore
$\tau = 3\alpha$ suffices.
This yields a probability bound of $\exp(-\alpha(n_w-1)/24)$.
A union bound over all over all $i$ yields the result, i.e.\ over the minimum of all $n_c$ companies,
or all $j_i$, i.e.\ over all $n_w$ workers.
\end{proof}
}

\begin{claim}
\label{clm::many-man-acceptable-props-mtm}
Let ${\mathcal B}_3$ be the event that between them, the workers with
rank at most $j_i+\ell_i$ make at least
$\tfrac 12 \alpha\beta (n_w-5/2)$ Step 2 proposals to $c_i$.
If events ${\mathcal B}_1$ and ${\mathcal B}_2$
do not occur, then ${\mathcal B}_3$ occurs with probability at most $\exp(-\alpha\beta (n_w-5/2)/8)$.
\end{claim}
It's proof is largely unchanged.
The first issue is that now in the run of the DA algorithm
placing a block on a worker $w$ corresponds to $w$ having
matched $d_w$ times. The proof is otherwise unchanged
as any unblocked worker will run through her full utility range as before.
However, the calculations change as follows.
The number of not-fully matched workers is at least
\begin{align*}
d_w(j_i+\ell_i) - d_c(i-1 +h_i) 
&\ge d_w \cdot \frac 52 \alpha(n_w-1) + d_c - d_c\cdot \frac 32 \alpha(n_c-1)\\
& \ge \alpha d_wn_w -\frac 52\alpha d_w \ge \alpha d_w(n_w-\frac 52).
\end{align*}

This causes the replacement of $n$ by $n_w- 5/2$ in the bounds.

\begin{claim}
\label{clm::one-man-accept-prop-mtm}
If none of the events ${\mathcal B}_1$, ${\mathcal B}_2$, or
${\mathcal B}_3$
occur, then
at least $\tfrac 14\alpha\beta\gamma (n_w-5/2)$ of the Step 2 proposals to $c_i$ will
each cause $c_i$ a loss of at most $\Lbar_c$ with probability at least
$1 -\exp(-\alpha\beta\gamma (n_w-5/2)/16)$.
\end{claim}

To ensure $c_i$ receives at least $d_c$ proposals that each cause it a loss at most $\Lbar_c$, by Claim~\ref{clm::one-man-accept-prop-mtm},
we need that 
\begin{align}
\label{eqn::dc-condition}
\tfrac 14\alpha\beta\gamma \big(n_w -\tfrac 52\big) \ge d_c.
\end{align}

Then the overall failure probability summed over all companies is at most
\begin{align*}
   & n_c\cdot \exp(-\alpha(n_c-1)/12)
    +n_w\cdot \exp(-\alpha (n_w-1)/24)
    +n_c\exp(-\alpha\beta (n_w-5/2)/8)\\
   & +n_c\exp(-\alpha\beta\gamma (n_w-5/2)/16).
\end{align*}

In the bounded derivative setting, we continue to set $\beta=\alpha\rho$, $\gamma=\alpha\rho$ and $\Lbar_c=4\alpha\mu$. Then, for large enough $n_c, n_w$, 
with $\Lbar_c\ge 4\mu\cdot[16(k+2)\ln \nmax/\rho^2 (n_w-5/2)]^{1/3}$ 
and $\Lbar_c \ge 48\mu(k+2) \ln \nmax/(n_c-1)\ge 47\mu (d_c/d_w) [(k+2) \ln \nmax/(n_w-1)]$,  the overall failure probability is at most $\nmax^{-k}$.
The first of the two bounds on $\Lbar_c$ dominates if
$d_c/d_w=O((n_w/\ln \nmax)^{2/3})$.
In addition, with $\Lbar_c\ge 4\mu\cdot(4d_c/\rho^2(n_w-\tfrac 52))^{1/3}$, \eqref{eqn::dc-condition} is satisfied.
Thus, the overall condition is that $\Lbar_c=\Omega(\max\{ d_c, \ln \nmax \} /n_w)^{1/3})$.

The corresponding bound $\Lbar_w=\Omega(\max\{ d_w, \ln \nmax \} /n_c)^{1/3})$ can be deduced using the company-proposing DA.

It remains to prove Claim \ref{clm::one-man-accept-prop-mtm}, which we do below.
\end{proof}

\begin{proof}
(Of Claim~\ref{clm::one-man-accept-prop-mtm}.)
As $\mathcal{B}_3$ does not occur, by Claim~\ref{clm::many-man-acceptable-props-mtm}, there
are at least $\tfrac 12\alpha\beta (n_w-5/2)$ Step 2 proposals to $c_i$. 
As explained in observation (4) of the sketch proof,
each Step 2 proposal has independent probability at least $\gamma$ of causing $c_i$ a loss of at most $\Lbar_c$
(the independence is because this is due to the private score of $c_i$ for this proposal).
In expectation, there are at least $\tfrac 12\alpha\beta\gamma (n_w-5/2)$ of the proposals causing $c_i$
a loss of at most $\Lbar_c$, and by a Chernoff bound at least
$\tfrac 14\alpha\beta\gamma (n_w-5/2)$ such proposals to $c_i$ with failure probability at most
$\exp(-\alpha\beta\gamma (n_w-5/2)/16)$.
\end{proof}

\subsection{Range of Public Ratings for Acceptable Edges}
\label{sec::cone-size}

Here we prove that in the one-to-one bounded derivative
setting, with high probability, for each $i$, the acceptable edges from women $w_i$ are to men with public rating in the range $[r_{m_i}-4\alpha,r_{m_i}+5\alpha]$, which we call $w_i$'s \emph{cone}.
A symmetric bound applies to the men.

We then obtain a similar bound for the many-to-one setting.

\begin{theorem}
\label{thm::cone-bdd-deriv}
In the one-to-one bounded derivative setting with $n$ men and $n$ women, for large enough $n$, with probability $1-n^{-c}$, for each $i$, the acceptable edges from women $w_i$ are to men with public rating in the range $[r_{m_i}-4\alpha,r_{m_i}+5\alpha]$, where $\alpha = \Lbar/4\mu$.
A symmetric bound applies to the men.
\end{theorem}
\begin{proof}
Theorem~\ref{thm:basic-bdd-deriv-result} bounds the loss by $\Lbar$ for non-bottommost agents with probability $1-n^{-c}$ for large enough $n$.
Therefore $w_i$ will not be interested in matching with any man with public rating less than
$r_{m_i} -\Lbar/\mu$ (for any such man would give a loss
greater than $\Lbar$). 

The situation to higher rated men needs a little more calculation.  Let $m_g$ be such a man.
Then what matters is whether $m_g$ incurs a loss of more than $\Lbar$ if matched to $w_i$. This happens if
$r_{w_g} -r_{w_i} > \Lbar/\mu=4\alpha$.
We now show that to obtain $r_{w_g} -r_{w_i} \le 4\alpha$, w.h.p.\ we must have $r_{m_g}-r_{m_i} \le 5\alpha$.

We prove this in two steps: first, we show that
w.h.p., if  $r_{w_g} -r_{w_i} \le 4\alpha$
then $i-g < 4\alpha(n-1) +\tfrac 12 \alpha (n-1)$.
Second, we show that w.h.p., $r_{m_g}-r_{m_i} \le (4\alpha+\tfrac 12\alpha) + \tfrac 12 \alpha = 5\alpha=\tfrac 54 L/\mu$.

The expected number of women in $W[r_{w_g}-4\alpha,r_{w_g}]$ other than $w_g$ is at most $4\alpha (n-1)$; and so by a Chernoff bound this number is at least $4\alpha (n-1)+\tfrac 12 \alpha (n-1)$ with probability at most
$\exp(-\alpha (n-1)/48)$. Call this bad event ${\mathcal B}_4$.
Note that by assumption, $w_i \in W[r_{w_g}-4\alpha,r_{w_g}]$, and so if ${\mathcal B}_4$ does not occur $i-g< 4\tfrac12 \alpha (n-1)$.

Now suppose ${\mathcal B}_4$ does not occur, and
consider the set $M[r_{m_g}-5\alpha,r_{m_g}]$.
In expectation, other than $m_g$, it contains $5\alpha (n-1)$ men.
By a Chernoff bound, it contains at most $4\tfrac 12 \alpha (n-1)$ men other than $m_g$
with probability at most $\exp(-\alpha(n-1)/40)$.
Call this bad event ${\mathcal B}_5$.

If neither ${\mathcal B}_4$ nor ${\mathcal B}_5$ occur, as $i-g< 4\tfrac12 \alpha (n-1)$,
$m_i \in M(r_{m_g}-5\alpha,r_{m_g}]$, and therefore $r_{m_g}-r_{m_i} < 5\alpha \le \tfrac 54 \Lbar/\mu$.

A union bound over the $n$ women and $n$ men gives
a failure probability of $2n\cdot \exp(-\alpha (n-1)/48) + 2n \cdot \exp(-\alpha (n-1)/40)$, plus the failure probability from the proof of Theorem~\ref{thm:basic-bdd-deriv-result}, which
was actually at most
\begin{align*}
    2n\cdot \exp(-\alpha(n-1)/12)
+2n\cdot \exp(-\alpha n/24)
+2n\exp(-\alpha^2\rho n/8)
    +2n\cdot \exp(-\alpha^3\rho^2 n/2).
\end{align*}
Even adding the new terms, for large enough $n$ it still suffices to set $\Lbar=[128(c+2)\mu^3\ln n/(\rho^2 n)]^{1/3}$ to achieve an overall $n^{-c}$ failure probability.
\end{proof}

We now extend the result to the many-to-many setting.

\begin{theorem}
In the many-to-many bounded derivative setting with $n_c$ companies and $n_w$ workers, for large enough $n_c$ and $n_w$, with probability $1-n^{-k}$,
for each $i$, the acceptable edges from worker $w_i$ are to companies with public rating in the range $[r_{c_{j_i}}-L_w/\mu,r_{c_{j_i}}+5L_c n_c/4\mu(n_c-1)]$,
where $c_{j_i}$ is the company aligned with $w_i$.
A symmetric bound applies to the companies.
\end{theorem}
\begin{proof}
We need to adapt the previous proof to account for the fact that there are $n_c$ companies and $n_w$ workers.

The argument demonstrating the lower limit is unchanged,
except we replace $\Lbar$ with $L_w$. For the upper limit, we change the argument as follows. Now, we replace $\Lbar$ by $L_c$ .

We first observe that the number of workers in $W[r_{w_g}-4\alpha,r_{w_g}]$ other than $w_g$
is at least $4\alpha (n_w-1)+\tfrac 12 \alpha (n_w-1)$ with probability at most
$\exp(-\alpha (n_w-1)/48)$. Call this bad event ${\mathcal B}_4$.

Second, if ${\mathcal B}_4$ does not occur, the set $M[r_{c_g}-5\alpha n_c/(n_c-1),r_{c_g}]$ contains at most $4\tfrac 12 \alpha n_c$ companies other than $c_g$
with probability at most $\exp(-\alpha n_c/40)$.
Call this bad event ${\mathcal B}_5$.

Suppose neither ${\mathcal B}_4$ nor ${\mathcal B}_5$ occur.
Then, the number of positions sought by the workers
in $M[r_{w_{i-1}},r_{w_g}]$ is $d_w(i-g) < 4\tfrac12 d_w\alpha (n_w-1)= 4\tfrac12 \alpha d_c n_c - 4\tfrac12 d_w\alpha$,
while the number of positions available in
$C[r_{c_g}-5\alpha,r_{c_g}]$ 
is more than
$4\tfrac12 d_c\alpha n_c= 4\tfrac12 \alpha d_c n_c$. 
Thus the number of available positions is at least the number sought, and therefore $c_i \in C[r_{c_g}-5\alpha n_c/(n_c-1),r_{c_g}]$.

It remains to revisit the probability bounds.
The failure probability from the proof of Theorem~\ref{thm::many-to-one} summed over all workers is
\begin{align*}
   & n_w\cdot \exp(-\alpha(n_w-1)/12)
    +\min\{n_c,n_w\}\cdot \exp(-\alpha (n_c-1)/24)\\
   & +n_w\exp(-\alpha\beta (n_c-5/2)/8)
    +n_w\exp(-\alpha\beta\gamma (n_c-5/2)/16),
\end{align*}
where $\beta=\gamma=\alpha\rho$.
There is an analogous bound for the companies.
Again, for large enough $n_c$ and $n_w$, we can use the same values for $L_c$ and $L_w$ as before while maintaining
the total failure probability at $n^{-k}$.
\end{proof}


\subsection{Distribution Bound on Losses (Proof of Theorem~\ref{thm::distr-bound})}
\label{sec::dist-bound}

Recall event ${\mathcal B}_3$ from the proof of Lemma \ref{lem::key_bounded} (see Appendix \ref{sec::main_lemma_full_proof}), that the women with
rank at most $i+\ell_i$ make fewer than
$\tfrac 12 \alpha\beta n$ Step 2 proposals to $m_i$.
Claim~\ref{clm::many-man-acceptable-props} shows that
if ${\mathcal B}_1$ and ${\mathcal B}_2$
do not occur then ${\mathcal B}_3$ occurs with probability at most $\exp(-\alpha\beta n/8)$.

Now, consider a man $m$ and the aligned woman $w$, where
$r_w\ge 4\alpha$.
Let $\beta=\gamma=\alpha \rho$.
We will bound the probability that $m$ has a loss of more than $L^m_\alpha\triangleq U(r_w,1) - (r_w-4\alpha,1)$.

If, in addition, $r_w\ge \sigmabar$ and $\alpha = \sigmabar/(4\cdot 2^h)=\Lbar/(4\mu \cdot 2^h)$,
as $U(r,1)-U(r-4\alpha,1)\le 4\alpha\mu=\Lbar/2^h$, this implies a loss of at most $\Lbar/2^h$.

\begin{lemma}
\label{lem::loss-distribution}

Let $m$ be a man and let $w$ be the aligned woman.
Suppose we run the DA algorithm cutting at $m$ and $r_m-\alpha$.
Then the probability that every Step 2 proposal to $m$ gives him a loss of more than $L^m_\alpha$ is at most $\exp(-\alpha\beta\gamma n/2)$.

\end{lemma}

\begin{proof}
Let $m$ be a man in $M[r_m,r_m+\delta]$.
As ${\mathcal B}_3$ does not occur,
$m$ receives at least $\tfrac 12 \alpha\beta n$ Step 2 proposals.
As shown in Observation 4 of the proof sketch,
each proposal gives a loss of more than $L^m_\alpha$ with probability at most $\gamma$.
Thus, the probability that every one of these proposals give $m$ a loss of more than $L^m_\alpha$
is at most
\begin{align*}
    (1 - \gamma)^{\alpha\beta n/2}\le \exp(-\alpha\beta\gamma n/2).
\end{align*}

\end{proof}

\begin{cor}
\label{cor::loss-distribution}
Suppose we run the DA algorithm cutting at $m$ and $r_m-\alpha$. Let ${\mathcal B}_6^h$ be the event that at least $2[1+(n-1)\delta]\cdot \exp(-\alpha\beta\gamma n/2)$ men in $M[r_m,r_m+\delta]$ suffer loss greater
than $L^m_\alpha$, where $0< \delta \le 1$ and $\alpha=\sigmabar/(4\cdot 2^h)$. 
If $r_w\ge \sigmabar$, $\beta=\gamma=\alpha \rho$,
and  none of ${\mathcal B}_1$--${\mathcal B}_3$ occur,
then ${\mathcal B}_6^h$ occurs with probability at most $n^{-(c+2)}$, where $\tfrac 13 \log\big(\frac{(c+2)\ln n}{\ln[\delta n/(3(c+2)\ln n)]}\big)\le h \le \tfrac 13 \log[(c+2)\ln n]$.
\end{cor}

\begin{proof}
Consider a man $m$ in $M[r_m,r_m+\delta]$. By Lemma \ref{lem::loss-distribution} it follows that the probability that $m$ experiences a loss of more than $L^m_\alpha$ is at most $\exp(-\alpha\beta\gamma n/2)$.

This bound depends only on $m$'s private scores for the Step $2$ proposals made to him.
Thus the outcomes for the different men in $M[r_m,r_m+\delta]$ are independent.

In expectation, at most $(1+(n-1)\delta)\cdot \exp(-\alpha\beta\gamma n/2)$ men in $M[r_m,r_m+\delta]$
suffer a loss of more than $L^m_\alpha$, and by a Chernoff bound, at most
$2(1+(n-1)\delta)\cdot\exp(-\alpha\beta\gamma n/2)$ men suffer such a loss
with probability $\exp(-\delta n\cdot\exp(-\alpha\beta\gamma n/2)/3)$.

Now $\exp(-\alpha\beta\gamma n/2) = \exp(-\Lbar^3 n \rho^2/[128\mu^3 2^{3h}])=\exp(-(c+2)\ln n/2^{3h})$.
Let $h= \tfrac13\log[(c+2)\ln n]-\tfrac 13\log g$.
Then $\exp(-\alpha\beta\gamma n/2)= \exp(-g)$.
In sum, at most $2(1+(n-1)\delta)/\exp(g)$ men in 
$M[r_m,r_m+\delta]$ have a loss of more
than $L^m_\alpha$ with probability at least
$1- \exp(-\delta n\cdot \exp(-g)/3)$.
So the failure probability is at most
$n^{-(c+2)}$ if $g\le \ln(\delta n/[3(c+2)\ln n])$.
\end{proof}

\begin{proof} (Of Theorem~\ref{thm::distr-bound})~
We apply Corollary~\ref{cor::loss-distribution} with $\delta=1$.
Over all the men and women whose aligned partners have public score at least $\sigmabar$, this yields that at most
$2n \cdot \exp(-\alpha\beta\gamma n/2) = 2n \cdot \exp(-(c+2)\ln n/2^{3h})$ men suffer a loss of more than $\Lbar/2^h$, and likewise for the women, with failure probability at most $2n^{-(c+1)}\log\log n$, for integer $h$ in the range
stated in the lemma.

Applying the prior analysis, if none of ${\mathcal B}_1$--${\mathcal B}_3$ occur, then the outcome is a stable matching with the bounds on the losses as stated in the previous paragraph. For large enough $n$, the failure probability will be at most $n^{-c}$.
\end{proof}

\section{Lower Bounds}
\label{sec::lower_full_proofs}
\subsection{A Lower Bound in the Linear Model}

The following theorem shows that the upper bound we obtained
is the best possible up to a constant factor.
The intuition is as follows: the expected number of acceptable edges per agent is $\Theta(\ln n)$, excluding the agents with public ratings of less than $L$. 
So long as the constant is small enough, the variance
in the number of these edges over all the agents will
be sufficient to ensure a good probability that at least one agent will have no incident acceptable edge.

For the lower bound we set $\lambda = \tfrac 12$.
We begin by identifying and bounding the probability of some bad events, denoted by ${\mathcal B}_4$ and ${\mathcal B}_5$. We then perform
an analysis for the case that ${\mathcal B}_4$ and ${\mathcal B}_5$ do not occur.

To do this, we need some additional notation.
In the following lemmas, we let $(m,w)$ and $(m',w')$ be two pairs of men and women with equal public ranks, and suppose their public ratings are $r_m, r_w, r_{m'},r_{w'}$, respectively.
We let $x = r_m - r_{m'}$ and $y = r_w - r_{w'}$.
Note that $\text{sign}(x)=\text{sign}(y)$.

\smallskip
\noindent
{\bf Event ${\mathcal B}_4$}.
Let ${\mathcal E}_4$ be the following event:
If $|x| \le 4 L$, then
the number of men with public ratings in the range $[r_m,r_{m'}]$ lies in the range $(2+ x\cdot(n-2) -L(n-2), 2+ x\cdot(n-2)+ L(n-2))$, and similarly for the women.
${\mathcal B}_4$ is the (bad) event that ${\mathcal E}_4$ does not occur.

\begin{lemma} \label{lem::index-to-range}
${\mathcal B}_4$ occurs with probability at most $2n^2\cdot n^{-L(n-2)/12\ln n}$.
\end{lemma}
\begin{proof}
Not counting $m$ and $m'$, the expected number of men
with public ratings in the range $[r_m,r_{m'}]$ is $|x|(n-2)$. By a Chernoff bound, this number lies outside the range
$(|x|(n-2) -   L(n-2), |x|(n-2) +   L(n-2))$ with probability at most $e^{- L^2(n-2)/2|x|} + e^{- L^2(n-2)/3|x|}\le 2e^{-L(n-2)/12}$, as $|x|\le 4 L$ by assumption.

The same bound applies to the women.
Now we apply a union bound to all $\tfrac 12 n(n-1)$ pairs $(m,w)$ and $(m',w')$ to obtain the result.
\end{proof}

\smallskip
\noindent
{\bf Event ${\mathcal B}_5$}.
This is the event that for some pairs $(m,w)$ and $(m',w')$,
either (i) $|y|= \big|r_{w'} - r_w\big|> 4L$ 
and the number of women
in the range $[r_w,r_{w'}]$ is at most
$2+ 3 L(n-2)$,
or (ii)  $|x|= \big|r_{m'} - r_m\big|> 4L$ and the number of men
in the range $[r_m,r_{m'}]$ is at most
$2+ 3 L(n-2)$.

\begin{lemma} \label{lem::prob-Bfour}
${\mathcal B}_5$ occurs with probability at most
$n^2\cdot n^{-L(n-2)/8\ln n}$.
\end{lemma}
\begin{proof}
We obtain a bound in case (i).
Excluding $w$ and $w'$, the expected number of women in the range $[r_w,r_{w'}]$
is at least $|y|(n-2)$.
By a Chernoff bound it is at most $y(n-2) - (|y|-3L)\cdot(n-2)$
with probability at most $\exp(-(|y|-3L)^2(n-2)/2|y|)\le \exp(-(|y|/4)^2(n-2)/2|y|)\le \exp(-L(n-2)/8)$.

The same bound holds in case (ii).
Now we apply a union bound to all $\tfrac 12 n(n-1)$ pairs $(m,w)$ and $(m',w')$ to obtain the result.
\end{proof}

\begin{theorem}
If $n\ge 32,000$ and $L= \tfrac 18 (\ln n/n)^{1/3}$
then with probability at least $\tfrac 14 n^{-1/8}$ there is no perfect matching, let alone stable matching,
in which every edge is $(L,\tfrac 32 L)$-acceptable.
\end{theorem}
\begin{proof}
Suppose that ${\mathcal B}_4$ and ${\mathcal B}_5$ do not occur.
Then, we will show that the expected number of women with
no acceptable incident edge is greater than or equal to $\tfrac12 n^{7/8}$.
As there are $n$ women, it immediately follows that 
with probability at least $\tfrac 12 n^{-1/8}$ there is no matching in which every edge is acceptable.
The result now follows if the probability of ${\mathcal B}_4 \cup {\mathcal B}_5$ is at most $\tfrac 14 n^{-1/8}$, i.e.\ that $2n^2\cdot n^{-L(n-2)/12\ln n} +n^2\cdot n^{-L(n-2)/8\ln n}  \le \tfrac 14 n^{-1/8}$; 
$n\ge 32,000$ suffices.

Lemma~\ref{lem::lower-bound} below shows that in expectation there are at least $n^{7/8}$ women (and men)
such that every possible proposal to one of these women would cause at least one of the two parties a loss
greater than $L$.
Recall that every edge to a woman with a public rating less than $\tfrac 32 L$ is woman-acceptable.
Let $w'$ be the topmost such woman (i.e.\ the one with the highest public rating). 
Let $w$ be the woman with the lowest public rating equal to or greater than $r_{w'}+2L$, and let $m$ be $w$'s aligned partner. Then the edge $(m,w')$ gives $m$ a loss greater than $L$, and thus every edge $(m,w'')$ that is man-acceptable to $m$ will be acceptable to $w''$ only if it gives $w''$ a loss of at most $L$.
So for men with public rating at least $r_m$,
an edge is acceptable only if it gives both partners a loss of at most $L$.
We show that there are at most $5Ln$ such men if
${\mathcal B}_4$ does not occur.
For if this event does not occur, then
the number of women in the range $[r_{w_n},r_w]$,
and hence the number of men in the range $[r_{m_n},r_m]$,
is at most $2+\tfrac92 L(n-2)\le 5Ln$, if $2\le \tfrac{1}{16}(n-2)\ln^{1/3}n/n^{1/3}$; $n\ge 256$ suffices. 
The same bound applies to the women.

%
Thus, 
there are at least $n^{7/8}-  5 Ln$ women
who do not have acceptable matches.
So long as $\tfrac 12 n^{7/8} \ge 5 Ln = \tfrac 5{16} n^{2/3}\ln^{1/3}n$,
this implies that the number of women with no acceptable match is at least $\tfrac 12 n^{7/8}$.
This condition holds when $n\ge 1$.
\end{proof}

\begin{lemma} \label{lem::score-span-diff}
Suppose that ${\mathcal B}_4$ and ${\mathcal B}_5$ do not occur.
Further suppose that either $|x|\le 2L$
or $|y|\le 2L$.
Then,
$\big|x-y| <  2L$.
\end{lemma}
\begin{proof}
We consider the case that $x\le 2L$.
The proof for the other case is symmetric.

Let $h$ be the
the number of men in the range $[r_m,r_{m'}]$; $h$ is also
the number of women in the range $[r_w,r_{w'}]$. By Lemma~\ref{lem::index-to-range}, $h\in 2+\big(|x|(n-2) - L(n-2), |x|(n-2) +  L(n-2)\big)$ and
if $|y|\le 4L$,
$h\in 2+\big(|y|(n-2) - L(n-2), |y|(n-2) +  L(n-2)\big)$. Consequently $|x-y| <  2L$ if $|y| \le 4L$.

If $|y|> 4L$, as ${\mathcal B}_5$ does not occur, the number of women in the range $[r_w,r_{w'}]$ is more than $2+ 3 L(n-2)$. But as ${\mathcal B}_4$ does not occur, the number of men is at most $2+ 3 L(n-2)$. These numbers are supposed to be equal, and therefore $|y|>4L$ cannot happen.
\end{proof}

\begin{lemma} \label{lem::acceptable-prob}
Suppose that ${\mathcal B}_4$ and ${\mathcal B}_5$ do not occur.  Further suppose that $y=r_w -r_{w'}\ge 0$. Then, 
the probability that edge $(m,w')$ causes a loss of at most $L$ to both $m$ and $w'$ is at most
$(2L-y)\cdot(4L+y)\le 8L^2$.
A symmetric bound applies if $\overline{x}=r_{m'}-r_m\ge 0$.
\end{lemma}
\begin{proof}
We show the proof for the first bound. The argument for the second bound is identical.
$m$ has a loss of at least $y$ on edge $(m,w')$.
Therefore, for $(m,w')$ to be acceptable to $m$, we need $y\le 2L$.
The probability that $(m,w')$ is acceptable
is $(2L-y) \cdot (2L+x)$, and by Lemma~\ref{lem::score-span-diff}, this is at most 
$(2L-y)\cdot(2L+y + 2L)$.
\end{proof}

\begin{lemma}
\label{lem::not-acceptable-due-to-public-score}
Consider an edge $(m,w')$. If $x<-2L$ or $y>2L$ then
$(m,w')$ is not acceptable.
\end{lemma}
\begin{proof}
If $y>2L$ then $m$ has a loss of more than $2L$,
and if $x<-2L$ then $w'$ has a loss of more than $2L$.
\end{proof}

\begin{defn}
\label{def::public-score-test}
Let $(m,w')$ be an edge.
If $x \ge -2L$ and $y\le 2L$ we say $(m,w')$ passes
the public rating test, and otherwise it fails the test.
\end{defn}

\begin{lemma}
\label{lem::bound-on-no-acceptable-edges}
Suppose that ${\mathcal B}_4$ and ${\mathcal B}_5$ do not occur. Then apart from at most $3+6L(n-2)$ edges $(m,w')$ 
all other edges incident on $m$ fail the public rating test.
\end{lemma}
\begin{proof}
Let $w'$ be the lowest rated woman in $W[r_w-2L,r_{w}]$. By Lemma~\ref{lem::index-to-range},
there are at most $2+3L(n-2)$ women in $W[r_{w'},r_{w}]\subseteq W[r_w-2L,r_w]$.
By Lemma~\ref{lem::not-acceptable-due-to-public-score},
for any woman $w''$ with a lower rating than $w'$, $(m,w'')$ will fail the public rating test (as for $w''$, $y>2L$).

Now let $m'$ be the highest rated man in $M[r_m,r_m+2L]$. By Lemma~\ref{lem::index-to-range},
for any woman $w''$ with a higher rating than $w'$, $(m,w'')$ will fail the public rating test (as for $w''$, $x<-2L$). 
By Lemma ~\ref{lem::index-to-range},
there are at most $2+3L(n-2)$ men in $M[r_m,r_m+2L]$,
and therefore there are at most $2+3L(n-2)$ women
in $W[r_w,r_{w'}]$.

$w$ belongs to both these sets of women.
So there are at most $3+6L(n-2)$ women who pass the public rating test.
\end{proof}

\hide{
\begin{lemma} \label{lem::prob-single-edge-low-loss}
Suppose that ${\mathcal B}_4$ and ${\mathcal B}_5$ do not occur.
The probability that a random edge $(m,w')$
with $w'\ne w$ causes a loss of at most $L$ to both $m$ and $w'$ is at most $20 L^3$.
\end{lemma}
\begin{proof}
We average the bound from Lemma~\ref{lem::acceptable-prob}
over all $y$ in the range $[0,2L]$, covering the
case that $y=r_w -r_{w'}\ge 0$, and then average over all
all $\overline{x}$ in the range $[0,2L]$, covering the
case that $\overline{x}=r_{m'}-r_m\ge 0$.
\begin{align*}
    \int_{0}^{2L} (2L-y) \cdot (4L+y) ~\text{d}y= 16 L^3 - 4L^3 - \frac {8}3 L^3 \le 9\tfrac23 L^3.
\end{align*}
Thus the overall bound is $19\tfrac 13 L^3\le 20L^3$.

\end{proof}
}


\begin{lemma} \label{lem::lower-bound}
Suppose that ${\mathcal B}_4$ and ${\mathcal B}_5$ do not occur and $n\ge 100$. 
If $L\le\tfrac 18 (\ln n/n)^{1/3}$, and there are equal numbers of men and women, then
the expected number of unmatched men (and women) is at least $n^{7/8}$.
\end{lemma}
\begin{proof}
Consider an arbitrary man $m$.
By Lemma~\ref{lem::acceptable-prob}, each edge which passes the public rating test is acceptable with probability at most $8L^2$.
By Lemma~\ref{lem::bound-on-no-acceptable-edges},
there are at most $3+6L(n-2)$ such edges incident on $m$.
Therefore
the probability that all $n$ edges incident on $m$ cause one or both parties
a loss of more than $L$ is at least
\begin{align*}
    \big(1-8L^2\big)^{3+6L(n-2)} = e^{(3+(6L(n-2))\ln (1- 8L^2)}  \ge n^{-1/8},
\end{align*}
as we argue next.

For this to hold, it suffices
that
\begin{align*}
&-8\cdot \frac{3+6L(n-2)} {\ln n} \Big(-8L^2- \frac 12(8L^2)^2 -\frac 13(8L^2)^3 -\ldots\Big) \le 1\\
\text{or that}\hspace*{0.4in} & 8 \cdot \frac{7Ln}{\ln n}\cdot\frac {8L^2}{1-8L^2}\le 1\hspace*{0.5in}\text{(if $3\le Ln$)}\\
\text{or that}\hspace*{0.4in}    &56 \cdot\frac {n^{2/3}}{8\ln^{2/3} n}\cdot \frac{\ln^{2/3}n}{8n^{2/3}-\ln^{2/3}n}\le 1\\
\text{or that}\hspace*{0.4in} &   \frac78 \cdot\frac 1{1-\ln^{2/3}n/8n^{2/3}}\le 1
\end{align*}
which holds if $\ln^{2/3}n/8n^{2/3}\le 1/8$, which holds
for $n\ge 1$.
Our other condition, $3\le Ln$, or $3\le \tfrac 18 n^{2/3}\ln^{1/3}n$, holds if $n\ge 100$. 
\hide{
$\tfrac{20}{512}\tfrac{\ln n}{n})\Big) \le 1$; in turn it suffices that
$8\big(\tfrac{20}{512} +\tfrac 12\big( \tfrac{20}{512}\big)^2 + \ldots \big) \le 1$,
and this is readily seen to be true.

Note that edge $(m,w)$ causes both $m$ and $w$ a loss of at most $L$ is $4L^2$.
We deduce from Lemma~\ref{lem::prob-single-edge-low-loss}
that
the probability that all $n$ edges incident on $w$ cause one or both parties
a loss of more than $L$ is at least
\begin{align*}
    \Big(1-4L^2)(1 -  20 L^3\big)^(n-1) \le e^{n\ln (1- 20L^3)}  \ge n^{-1/8},
\end{align*}
as we argue next.
For this to hold, it suffices
that
$-8\tfrac n{\ln n}\ln \Big(1 - \tfrac{20}{512}\tfrac{\ln n}{n})\Big) \le 1$; in turn it suffices that
$8\big(\tfrac{20}{512} +\tfrac 12\big( \tfrac{20}{512}\big)^2 + \ldots \big) \le 1$,
and this is readily seen to be true.
}

Thus the expected number of women having all incident edges causing a loss of more than $L$ to both parties
is at least $n^{7/8}$.
\end{proof}


\subsection{Lower Bound on Performance for the General Utility Model}
\label{sec::general_lower}

Now we show that without the bounds on the derivatives, no sub-constant loss is achievable in general.

\begin{defn}[Sub-constant function]
A function $f(x):{\mathbb R}\rightarrow {\mathbb R}^+$ is sub-constant if for every choice of constant $c>0$, there exists an $\xbar$ such that for all $x \ge\xbar$, $f(x) \le c$.
\end{defn}

We first examine what happens if the derivatives w.r.t.\ private scores are not bounded,
but the derivatives w.r.t.\ public ratings are bounded;
this implies there is no lower bound on the ratio
of the derivatives of the utility functions w.r.t.\ public ratings and private scores (recall Definition~\ref{def::bdd-deriv}).

\begin{lemma}\label{obs::limits_of_result}
Let $f:N\rightarrow {\mathbb R}^+$ be a continuous, strictly decreasing sub-constant function, and let $\sigma,\delta \in (0,1)$ be constants. Suppose the public ratings and private scores of the $n$ men and $n$ women are drawn uniformly and independently from $[0,1]$. Then there exist continuous and strictly increasing utility functions $U(.,.)$ and $V(.,.)$ 
having derivatives w.r.t.\ their first variables that are bounded by a constant, but for (at least) one of which the derivatives w.r.t.\ their second variables are not bounded by any constant, having the following property:
for some $\nbar>0$,
for all $n\ge \nbar$, with probability at least $1 - \delta$, in every perfect matching, some rank $i$ man $m_i$ or woman $w_i$ with public rating at least $\sigma$ receives utility less than $U(r_{w_i},1)-f(n)$ or $V(r_{m_i},1)-f(n)$, respectively.
\end{lemma}
\begin{proof}
We will give an example where, with probability at least $1-\delta$, in every perfect matching, man $m_1$
receives utility less than $U(r_{w_1},1)-f(n)$.

Observe that proving the result for a more slowly decreasing $f$ implies it for faster decreasing functions.
In what follows, at times we will need to assume $f$ decreases sufficiently slowly, but given the just made observation, we can do so WLOG.

Now we define $U(r,s) = r + g(s)$ where $g(s)$ is a continuous, strictly increasing function, and for which 
\begin{align*}
    g\Big( 1 - \frac {\delta} {8n\cdot f(n)}\Big) = g(1) - f(n).
\end{align*}
The reason for this condition will become clear in due course.
We will first demonstrate that there is such a $g$.
To this end, define
\begin{align*}
 \hspace*{1in}   k(y) &= \frac { 8y\cdot f(y)}{\delta}&\text{\rm for } y \ge 0 \hspace*{1in} \\
    g(s) &= g(1) - f(k^{-1}(1/(1-s))),~~&\text{if}~s<1\hspace*{1in} \\
    g(1) &= f(k^{-1}((1))) & \text{(so $g(0) = 0$).}\hspace*{1in}
\end{align*}
We will want $k$ to be strictly increasing and unbounded. This is true if $f$ is sufficiently slowly decreasing.
Next, note that as $f$ is continuous, so is $k$. Therefore  $k^{-1}$ is continuous and strictly increasing, and therefore so is $g$,
except possibly at $s=1$.
For $g$ to be continuous at $s=1$ we need $\lim_{s\ra 1}f(k^{-1}(1/(1-s))) = 0$, which happens as $\lim_{x\ra \infty}f(x) = 0$, which happens since $f$ is sub-constant.

Setting $s = 1-\frac {\delta} {8n\cdot f(n)}$ gives
\begin{align}
\label{eqn:g-defn}
    g(s) = g(1) - f\Big(k^{-1}\Big(\frac {8n\cdot f(n)}  {\delta}\Big)\Big) = g(1) -f(n),
\end{align}
as desired.

Strictly speaking, we should rescale the utility so that its range is $[0,1]$ rather than the actual $\big[0, 1 + g(1)]\subset\big[0,1+f(0)\big]$. 
Note that although $1+g(1)$ is a function of $\delta$,
it is always bounded by $1+f(0)$, a constant, and so the rescaling
does not affect the result stated in the lemma.
We omit performing the rescaling to avoid unnecessary clutter.

For $m_1$ to face a loss of at most $f(n)$, he must match with a woman having public rating at least $r_{w_1}-f(n)$.

The probability that no woman has a public rating in the range
$[1-\ln (4/\delta )/n,1]$ is at most 
\begin{align*}
\Big(1- \frac{\ln(4/\delta)}{n}\Big)^n \le \exp(-\ln(4/\delta))=\tfrac 14 \delta.
\end{align*}
Otherwise, $r_{w_1} \ge 1-\ln (4/\delta )/n$.
If $f$ is sufficiently slowly decreasing, for large enough $n$,
$\ln (4/\delta )/n \le  f(n)$. Therefore, for such large enough $n$,
with probability at least $1-\tfrac 1{4}\delta$,
$r_{w_1} - f(n) \ge 1 - 2 f(n)$.
Call the probability $\tfrac14\delta$ event ${\mathcal B}_1$.

The same analysis shows that, with failure probability at most
$\tfrac 1{4} \delta$, $r_{m_1} \ge 1 -f(n)$,
and as $f$ is a sub-constant function, for large enough $n$,
$r_{m_1} \ge 1 -f(n) \ge \sigma$.
Call the probability $\tfrac14\delta$ event ${\mathcal B}_2$.

The expected number of women other than $w_1$ in  $W[r_{w_1}-f(n), r_{w_1}]$ is $(n-1) \cdot f(n)$. Let  $n_w$ be the actual number of women other than $w_1$ in this range.
By a Chernoff bound,
\begin{align*}
 \Pr\Big[n_w \ge  (n-1) \cdot f(n)+ \sqrt{3(n-1)\cdot f(n)\cdot\ln(4/\delta)}\Big] \le e^{-\ln(4/\delta)} = \tfrac 1{4} \delta.
\end{align*}
If $f$ is decreasing sufficiently slowly, then for sufficiently large $n$,
$1+\sqrt{3(n-1)\cdot f(n)\cdot\ln(4/\delta)} \le n\cdot f(n)$;
therefore, in addition, $1+n_w \le 2n \cdot f(n)$ with probability at least $1 -\tfrac 1{4} \delta$.
Note that $1+n_w$ is the number of women in 
$W[r_{w_1}-f(n),1]$.
Call the probability $\tfrac14\delta$ event ${\mathcal B}_2$.

Next, consider an edge $(m_1,w)$ for which $m_1$'s private score is $s$.
For this edge to cause more than $f(n)$ loss to $m$, it suffices that $g(s) < g(1) -f(n)=g(1 -\delta/[8n\cdot f(n)])$ by \eqref{eqn:g-defn}.
This occurs with probability $\delta/[8n\cdot f(n)]$.

We now lower bound the probability that every match with a woman in $W[r_{w_1}-f(n),1]$
causes $m$ a loss of more than $f(n)$ if none of the events ${\mathcal B}_1$--${\mathcal B}_3$ occur. This probability
is at least:
\begin{align*}
    \Big(1 - \frac{\delta}{8n\cdot f(n)}\Big)^{2n\cdot f(n)}
    \ge 1 - \tfrac 14 \delta.
\end{align*}

Thus, by a union bound, modulo an overall failure probability of at most $\delta$,
$m_1$ has a loss of more than $f(n)$ on every incident edge, and hence in every perfect matching some agent ($m_1$ actually) incurs a loss of more than $f(n)$.
\end{proof}

We now consider the case where the derivatives w.r.t.\ the first variable are bounded, but there is no bound on the derivatives w.r.t.\ the second variable.

\begin{lemma}\label{obs::limits_of_result_2}
Let $f:N\rightarrow {\mathbb R}^+$ 
be a continuous, strictly decreasing sub-constant function,
and let $\sigma,\delta \in (0,1)$ be constants. Suppose the public ratings and private scores of the $n$ men and $n$ women are drawn uniformly and independently from $[0,1]$. Then there exist continuous and strictly increasing utility functions $U(.,.)$ and $V(.,.)$ having derivatives w.r.t.\ their second variables that are bounded by a constant, but for (at least) one of which the derivative w.r.t.\ their first variable is not bounded by any constant, having the following property:
for some $\nbar>0$,  
for all $n\ge \nbar$, 
with probability at least $1 - \delta$, in every perfect matching, some man $m_i$ or woman $w_i$ with public rating at least $\sigma$ receives utility less than $U(r_{w_i},1)-f(n)$ or $V(r_{m_i},1)-f(n)$, respectively.
\end{lemma}

\begin{proof}
We will give an example where, in every stable matching,  man $m_1$ receives  utility less than $U(r_{w_1},1)-f(n)$. 
The analysis has the same thrust as the one for the preceding lemma.

Let $U(r,s)=\widetilde{g}(r)+s$, where $\widetilde{g}$ is defined below is a very similar way to the $g$ in the proof of Lemma~\ref{obs::limits_of_result}.
\begin{align*}
 \hspace*{1in}   \widetilde{k}(y) &= \frac { 8y\cdot f(y)}{\delta}& \\
    \widetilde{g}(r) &= \widetilde{g}(1) - 3\cdot f(\widetilde{k}^{-1}(1/(1-r))),~~&\text{if}~r<1\hspace*{1in} \\
    \widetilde{g}(1) &= 3\cdot f(\widetilde{k}^{-1}((1))) & \text{(so $\widetilde{g}(0) = 0$).}\hspace*{1in}
\end{align*}
Now, setting $r = 1-\frac {\delta}{8n\cdot f(n)}\triangleq 1 -\nu$, gives $\widetilde{g}(r) = \widetilde{g}(1) - 3\cdot f(n)$.
Again, strictly speaking, we should rescale the utility so that its range is $[0,1]$. 

Next, we will show that $\widetilde{g}(1-\nu/3) -\widetilde{g}(r) \ge f(n)$.
As $\widetilde{g}(r) = \widetilde{g}(1) - 3\cdot f(n)$, it suffices to show that $\widetilde{g}(1-\nu/3) \ge \widetilde{g}(1) -2\cdot f(n)$, which we do as follows:
\begin{align*}
    \widetilde{g}(1-\nu/3) &= \widetilde{g}(1) -  3\cdot f\Big(\widetilde{k}^{-1}\Big(\frac{3}{\nu}\Big)\Big) \\
    &= \widetilde{g}(1) -  3\cdot f\Big(\widetilde{k}^{-1}\Big(3\cdot \frac{8n\cdot f(n)}{\delta}\Big)\Big)\\
    & = \widetilde{g}(1) -  3\cdot f\Big(\widetilde{k}^{-1}\Big(\frac 32\cdot \frac{f(n)}{f(2n)} \cdot \frac{8\cdot 2n \cdot f(2n)}{\delta}\Big)\Big)
\end{align*}
If $f$ is sufficiently slowly decreasing, for all $n$,
$\frac 32\cdot \frac{f(n)}{f(2n)}\ge 1$, and as $\widetilde{k}^{-1}$ is increasing and $f$ is
decreasing, the RHS of the above expression is at least
\begin{align*}
    \widetilde{g}(1) -  3\cdot f\Big(\widetilde{k}^{-1}\Big( \frac{8\cdot 2n \cdot f(2n)}{\delta}\Big)\Big)=\widetilde{g}(1) - 3\cdot f(2n)
    \ge \widetilde{g}(1) -2\cdot f(n),
\end{align*}
as $3\cdot f(2n) \ge 2\cdot f(n)$, if $f$ decreases sufficiently slowly.

The probability that no woman has a public rating in the range
$[1-\ln (4/\delta)/n,1]$ is at most $\tfrac 14 \delta$.
As $f$ is a sub-constant function, for large enough $n$,
$\ln (4/\delta)/n \le  \tfrac 13 \cdot \delta/[8n\cdot f(n)]= \tfrac 13\nu$. 

Therefore, for large enough $n$,
with probability at least $1-\tfrac 14\delta$, 
\begin{align*}
\widetilde{g}(r_{w_1}) - \widetilde{g}(1-\nu)\ge \widetilde{g}(1 - \nu/3) - \widetilde{g}(1-\nu) \ge f(n).
\end{align*}
Thus, with probability at least $1-\tfrac 14\delta$, all women with public rating less than $(1-\nu)$ will cause $m_1$ a loss of more than $f(n)$.
Call the probability $\tfrac14\delta$ event ${\mathcal B}_1$.

Let $n_w$ be the actual number of women, aside $w_1$, with public rating at least $(1-\nu)$.
$\text{E}\big[n_w\big]\le(n-1)\nu$.
By a Chernoff bound,
\begin{align*}
 \Pr\Big[n_w \ge (n-1)\nu + \sqrt{3(n-1)\nu \cdot \ln(4/\delta)}\Big] \le \exp(-\ln(4/\delta))= \tfrac 14\delta.
\end{align*}
However slowly $f$ is decreasing, for large enough $n$, $1+\sqrt{3(n-1)\nu \cdot \ln(4/\delta)}\le n \nu$,
which implies $1+n_w \le 2n\nu$.
Call the probability $\tfrac14\delta$ event ${\mathcal B}_2$.

The same analysis as in the proof of Lemma~\ref{obs::limits_of_result}
shows that, with failure probability at most
$\tfrac 14 \delta$, $r_{m_1} \ge 1 -f(n)$,
and as $f$ is a sub-constant function, for large enough $n$,
$r_{m_1} \ge 1 -f(n) \ge \sigma$.
Call the probability $\tfrac14\delta$ event ${\mathcal B}_3$.

Next, note that an edge $(m_1,w)$ causes $m_1$ a loss of more
than $f(n)$ based on the private score alone with probability $1-f(n)$.

Thus, if none of the events ${\mathcal B}_1$--${\mathcal B}_3$ occur, the probability that every edge incident on $m_1$ causes it a loss of more than $f(n)$ is at least
\begin{align*}
 \big(1 - f(n)\big)^{2n\nu}=\big(1 - f(n)\big)^{\delta/[4\cdot f(n)]}\ge 1 - \delta/4,
\end{align*}
if $\delta \le 1$.

Therefore, by a union bound, modulo an overall failure probability of at most $\delta$,
$m_1$ has a loss of more than $f(n)$ on every incident edge, and hence in every perfect matching some agent ($m_1$ actually) incurs a loss of more than $f(n)$.
\end{proof}


\section{$\eps$-Bayes-Nash Equilibria}
\label{sec::eps-BN-equil-full-match}

In this section, we demonstrate that there is a $\eps$-Nash equilibrium in which with high probability all agents have low losses.
To obtain this result, we need the stronger bounded-derivatives condition, namely we need both lower and upper bounds for the two derivative expressions
(see Definition~\ref{def::strong-bdd-deriv}).
We will assume that both $U$ and $V$ satisfy the strong bounded derivative property.

\hide{
\begin{defn} \label{def::strong-bdd-deriv}
A function $f(x,y):{\mathbb R}^2\rightarrow {\mathbb R}^+$ has \emph{$(\rho_{\ell},\rho_u,\mu_{\ell},\mu_u)$-bounded derivatives} if for all $(x,y)\in {\mathbb R}^2$,
\begin{align*}
   &\text{\rm The ratio bound:}~~\rho_{\ell} \le\pdv{f}{x}\Big/\pdv{f}{y}\le \rho_u.\\ 
&\text{\rm The first derivative bound:}~~\mu_{\ell}\le \pdv{f}{x}\leq \mu_u.
\end{align*}
Then $f$ is said to have the \emph{strong bounded derivative} property. 
\end{defn}
}

Our analysis here will repeatedly use weak stochastic dominance to justify the application of Chernoff bounds.
To avoid repetition, we summarize the technique here.
Suppose ${\mathcal X} =\{X_1,X_2,\ldots,X_m\}$ is a collection of not-necessarily independent binary random variables. Suppose that $\Pr[X_i=1 | {\mathcal X}\setminus\{X_i\}] \le p_i$.
Let $Y_i$ be a binary variable with $\Pr[Y_i=1]=p_i$, with the $Y_i$ being independent.
Clearly, for all $z$, $\Pr[\sum_i X_i\ge z] \le \Pr[\sum_i Y_i\ge z]$.
Thus, if by means of a Chernoff bound, we show that
$\Pr[\sum_i Y_i\ge z] \le q$, then $\Pr[\sum_i X_i\ge z]\le q$ also.
Henceforth, we will justify this application of a Chernoff bound to ${\mathcal X}$ by saying it uses \emph{stochastic dominance}.

Our analysis will build on the bound shown in Theorem~\ref{thm::many-to-one}.
It will be helpful to review the randomness that was used.
The probability that events ${\mathcal B}_1$ and ${\mathcal B}_2$ do not occur is based on the public ratings of the men and women.
The bound on the probability that event ${\mathcal B}_3$ occurs for a particular man $m_i$
is based on the private scores of the edges to $m_i$,
namely the private score of the woman $w_j$ for the man $m_i$, for each such edge $(m_i,w_j)$.
Note that ${\mathcal B}_3$ is the bad event in Claim~\ref{clm::many-man-acceptable-props}.
The bound on the final error term is based on the private scores of $m_i$ for their edges to the women $w_j$, as given in Claim~\ref{clm::one-man-accept-prop}.
We will let ${\mathcal B}_7$ denote the bad event in Claim~\ref{clm::one-man-accept-prop}, namely that all the proposals to $m_i$ cause him too large a loss.
Symmetric bounds apply to the women.

In the analysis that follows, we will identify additional bad events concerning there being too few or too many agents in a range of public ratings; these will depend on the range. We will also bound the probability of losses for bottommost women and men using the private scores of proposals to these agents; these private scores
will be disjoint from the ones used in the bounds mentioned in the previous paragraph.

In the remainder of this section, $m$ and $w$ are always aligned, as are $m'$ and $w'$, $m''$ and $w''$, etc.

In addition, in order to improve some of the bounds, we will restate losses in terms of public ratings and private scores.
A quick inspection of the proof of Theorem~\ref{thm:basic-bdd-deriv-result}
shows that the high probability bound on the loss for a man $m$, whose aligned woman $w$ has public rating $r_w\ge \sigmabar$, is at most $U(r_w,1) - (r_w-\sigmabar,1)$ (recall that $\alpha = \tfrac 14 \sigmabar)$.
Also note that is suffices to set $\sigmabar= [128(c+2)\ln n/(\rho_l^2 n)]^{1/3}$.
Similarly, if $r_w \ge \sigmabar/t$, where $t>1$,
the bound on the loss is at most
$L^m_t\triangleq U(r_w,1) - U(r_w-\sigmabar t^2,1)$ (see the proof of Theorem~\ref{thm::low_loss_bounded}).
Analogously, for a woman $w$,
if $r_m \ge \sigmabar/t$, where $t>1$,
the bound on the loss is at most
$L^w_t\triangleq V(r_m,1) - V(r_m-\sigmabar t^2,1)$.

As already noted, these bounds use the private scores of proposals to men and women with public ratings of at least $\sigmabar/t$.

Shortly, we will specify maximum values $t_m$ and $t_w$ of $t$ for the men and women, respectively.
We will demonstrate the existence of a stable match in which w.h.p.\ every man $m$
has a loss of at most 
$L^m_{t_m}$, and every
woman $w$ has a loss of at most
$L^w_{t_w}$.
For this to be meaningful when $r_w-\sigmabar t_m^2 < 0$, we extend the definition of $U$ to this domain as follows. For $r<0$, $\pdv{U(r,s)}{r}=\mu_{\ell}$ and
$\pdv{U(r,s)}{s}=\rho_{\ell}$.
We proceed analogously to handle the case that
$r_m-\sigmabar t_w^2 < 0$.

We define $t_m$ and $t_w$ using  suitable constants $\eta>1$ and
$0<\nu<1$, which we will specify later. 
We set $\sigma_m = \nu/ n^{1/3}$ and $\sigma_w= \eta/n^{1/3}$.
We then define $t_m= \sigmabar/\sigma_m$ and
$t_w= \sigmabar/\sigma_w$.
Note that $\sigma_w/\sigma_m= \eta/\nu$.

The maximum loss will occur only to some of the agents with low public ratings. We identify the potentially high-loss agents as follows.

\begin{defn}
\label{def::bottom-zone}
Let $w'$ be the bottommost woman with a public rating of at least $\sigma_m$ and let $m'$ be aligned
with $w'$. Then the \emph{bottom zone} of men comprises
the set $B_M\triangleq M[0,r_{m'})$, and the top zone $T_M$ comprises $M[r_{m'},1]$.
Similarly, let $m''$ be the bottommost man with a rating of at least $\sigma_w$ and let $w''$ be aligned
with $m''$. Then the \emph{bottom zone} of women comprises
the set $B_W\triangleq W[0,r_{w''})$,
and the top zone $T_W$ comprises $W[r_{w''},1]$.
\end{defn}

Note that by Theorem~\ref{thm::low_loss_bounded}, in any stable match, every man $m \in T_M$ has loss at most $L_{t_m}^m$.
Likewise, every woman $w \in T_W$ has loss at most $L_{t_w}^w$.

We also want to distinguish those edges which yield men a utility of at least $U(0,1)$ and women a utility of at least $V(0,1)$.

\begin{defn}
\label{def::high-edges}
An edge $(m_i,w_j)$ is \emph{man-high} if $U(r_{w_j},s_{m_i}(w_j))\ge U(0,1)$, and otherwise it is \emph{man-low}; 
it is \emph{woman-high} if $V(r_{m_i},s_{w_j}(m_i)) \ge V(0,1)$, and otherwise it is \emph{woman-low}.
\end{defn}

We begin by identifying two bad events ${\mathcal B}_8$
and ${\mathcal B}_9$ and bounding the probabilities they occur.

\smallskip
\noindent
{\bf Event ${\mathcal B}_8$}.
Let ${\mathcal E}_8$ be the event that the number of men in $B_M$ lies in the range $(\tfrac 12 \sigma_m \cdot n,2\sigma_m\cdot n)$, and each of these men has public rating less than $3\sigma_m$, together with the corresponding event for women.
Let ${\mathcal B}_8$ be the complementary event.

\begin{lemma}
\label{lem::max-score-bottom-zone}
${\mathcal B}_8$ occurs with probability at most $\exp(-\sigma_m\cdot n/3)+ \exp(-\sigma_m\cdot n/6)+\exp(-\sigma_m\cdot n/8)+\exp(-\sigma_w\cdot n/3)+ \exp(-\sigma_w\cdot n/6)+\exp(-\sigma_w\cdot n/8) \le 6\exp(-\sigma_m\cdot n/8)$.
This bound is based on the independent random choices of public ratings for the men and women.
\end{lemma}
\begin{proof}
In expectation, there are $\sigma_m\cdot n$ women with public rating less than $\sigma_m$. These choices are based on the women's independent public scores. Hence, by a Chernoff bound,
the probability that there are at least $2\sigma_m\cdot n$ women with public rating less than $\sigma_m$ is at most
$\exp(-\sigma_m\cdot n/3)$,
and the probability that there are at most
$\tfrac 12 \sigma_m\cdot n$ women with public rating less that $\sigma_m$ is at most $\exp(-\sigma_m\cdot n/8)$.
But these are the women aligned with the men in $B_M$. Hence these bounds also apply to the number of men in $B_M$.

Now we bound the probability that there are at most $2\sigma_m\cdot n$ men in the public rating range $[0,3\sigma_m)$.
The expected number of men in this range is $3\sigma_m\cdot n$.
This is based on their independent public ratings.
Then, by a Chernoff bound, there are at most $2\sigma_m\cdot n$ men in this range with probability at most
$\exp(-\sigma_m\cdot n/6)$.

Analogous bounds apply to the women.
\end{proof}

\noindent
{\bf Event ${\mathcal B}_9$}. This is the event that
$r_{m'} < 4\sigma_m$.

\begin{lemma}
\label{lem::smprime-bound}
If ${\mathcal B}_8$ does not occur, then
${\mathcal B}_9$ occurs with probability at most $\exp(-\sigma_m\cdot n)$. 
This bound is based on the independent random choices of public ratings for the men.
\end{lemma}
\begin{proof}
As ${\mathcal B}_8$ does not occur, by Lemma~\ref{lem::max-score-bottom-zone},
$B_M\subset M[0,3\sigma_m)$.
Therefore, if $r_{m'}\ge 4\sigma_m$,
$M[3\sigma_m,4\sigma_m)$ is empty.
But the probability that $M[3\sigma_m,4\sigma_m)$ is empty is at most $(1-\sigma_m)^n\le \exp(-\sigma_m\cdot n)$, and it follows that this is the probability that
$r_{m'}\ge 4\sigma_m$.
\end{proof}

The desired stable match will be found by running the woman-proposing DA
when each man $m$, whose aligned woman $w$ has public rating less than $\sigma_m$, applies a truncation strategy
of refusing proposals that provide a loss greater than $L^m_{t_m}$.
No truncation is applied by men with higher public ratings, but we already know their losses
are bounded by $L^m_{t_m}$.
The women apply a symmetric truncation, meaning that a woman $w$ will only propose edges that provide a loss of at most $L^w_{t_w}$.

Our analysis considers the result of running woman-proposing DA on the truncated edge set.
We begin by observing that every man in $T_M$ is matched, and similarly every woman in $T_W$ is matched.
We then argue that every woman in $B_W$ will be matched,
from which we deduce that every man in $B_M$ must also be matched.

As we showed in Theorem~\ref{thm::low_loss_bounded},
with failure probability $O(n^{-(c+1)})$, in every stable match, every man $m$ in $T_M$ will have a loss of at most 
$L^m_{t_m}$.
Furthermore, this match is achieved with the edge set cut as in Lemma~\ref{lem::key_bounded}.
As the men in $T_M$ do not truncate any edges, all the edges required for Lemma~\ref{lem::key_bounded} remain present despite the men's truncations.
Also, all the edges used by this lemma are women-high,
and the women do not truncate such edges.
Thus the result of Lemma~\ref{lem::key_bounded} continues to apply as does Theorem~\ref{thm::low_loss_bounded}.

A symmetric argument shows that with failure probability $O(n^{-(c+1)})$, in every stable match, every woman $w$ in $T_W$ will have a loss of at most 
$L^w_{t_w}$.

To analyse what happens to the women in $B_W$ we proceed as follows.

\hide{
Furthermore, the construction in Theorem~\ref{thm:basic-bdd-deriv-result}
ensures that all these men receive a Phase 1 proposal that is woman-high, for these are edges the women leave untruncated. Thus, in Phase 1, w.h.p., all the men in $T_M$ are matched.

Our analysis considers a run of woman-proposing DA which proceeds in two phases.
\RJC{I don't believe we actually use the phases}

\smallskip
\noindent
{\bf Phase 1}.
Recall that $w'$ is the woman with the smallest public rating greater than or equal to $\sigma_m$, 
and that $m'$ is aligned with $w'$. 
Phase 1 applies a cut at $r_{m'}-\tfrac 14\sigma_m$.

\smallskip
\noindent
{\bf Phase 2}.
Each unmatched women keeps proposing the next edge on her preference list. 

\smallskip

In every stable match, every man $m$ in $T_M$ will have a loss of at most 
$L^m_{t_m}$.
These men do not truncate any edges.
Furthermore, the construction in Theorem~\ref{thm:basic-bdd-deriv-result}
ensures that all these men receive a Phase 1 proposal that is woman-high, for these are edges the women leave untruncated. Thus, in Phase 1, w.h.p., all the men in $T_M$ are matched.

Consequently, it suffices to show that during Phase 2 every man in $B_M$ becomes matched.
For this implies all the men are matched and therefore all the women are matched.
}
We observe that w.h.p.:\\
i. The men in $B_M$ receive at most $\tfrac14|B_M|$ proposals which are both man-high and woman-high.\\
ii. The men in $B_M$ receive at most $\tfrac14 |B_M|$ proposals which are both man-high and woman-low.\\
iii. We conclude that at most half the men in $B_M$ will receive a man-high proposal.\\
iv. The proposals from $B_W$ to $B_M$ that are both man-low and woman-low behave in the same way as in the uniform random model, up to a constant factor. This will mean that it suffices that the women in $B_W$ have $\Theta(\ln^2 n)$ man-low and women-low edges to the men in $B_M$ (which they do), and ensures that each man in $B_M$ receives at least one proposal.

Our analysis will also be concerned with the following subsets $W_h$ of women, for integer $h\ge 0$; $W_h$ comprises the women aligned with the men in $M[2^h\sigma_w,2^{h+1}\sigma_w)$.

\smallskip\noindent
{\bf Event ${\mathcal B}_{10}$}.
${\mathcal B}^h_{10}$ is the event that  $|W_h|\ge 2^{h+1}\eta\cdot n^{2/3}$.
And ${\mathcal B}_{10} = \cup_{h\ge 0}{\mathcal B}^h_{10}$.

\begin{lemma}
\label{lem::women-in-intermed-zone}
${\mathcal B}^h_{10}$ occurs with probability at most $\exp(-2^{h}\eta \cdot n^{2/3}/3)$.
And ${\mathcal B}_{10}$ occurs with probability at most
$2\exp(-\eta \cdot n^{2/3}/3)$, if $n^{2/3} \ge 3$.
These bounds are based on the independent random choices of public ratings for the men.
\end{lemma}
\begin{proof}
The expected number of men in $M[2^h\eta  /n^{1/3},2^{h+1}\eta/n^{1/3})$ is $2^h\eta n^{2/3}$, and these choices are based on the men's public ratings.
Thus, by a Chernoff bound, there are at least $2^{h+1}\eta n^{2/3}$ men in this range with
probability at most 
$\exp(-2^h\eta n^{2/3}/3)$.
This is also the bound on the number of women aligned with these men.

The second claim follows on summing the probability bound over $h$, using the assumption that $n^{2/3} \ge 3$.
\end{proof}

\begin{lemma}
\label{lem::num-matches-WM-BL}
Suppose that none of ${\mathcal B}_1$--${\mathcal B}_{10}$ occur.
Then there are at most $\tfrac 14|B_M|$ matches between women in $T_W$ and men in $B_M$, with failure probability at most $\exp(-|B_M|/24)$,
if $\eta\ge 6\nu$ and $4(\eta/\nu) \cdot \exp(-(\eta/2)^3\rho_{\ell}^2/128) \le \tfrac 1{10}$.
\end{lemma}
\begin{proof}
We will consider the sets $W_h$ of women aligned with $M[2^h\eta/n^{1/3},2^{h+1}\eta/n^{1/3})$, for $h\ge 0$. The union of the sets forms $T_W$.

If a woman $w_j$ in $W_h$ is matched to a man $m_i$ in $B_M$ the difference in public scores between $m_j$ and $m_i$ is at least
\begin{align*} 
r_{m_j} -r_{m_i} \ge (2^h\eta  - 3\nu)/n^{1/3}\ge 2^{h-1}\eta/n^{1/3} \triangleq g_h,
\end{align*}
as $\eta\ge 6\nu$ and ${\mathcal B}_8$ does not occur (and hence $r_{m_i} \le 3 \sigma_m$).

We will apply Lemma~\ref{lem::loss-distribution}, swapping the roles of the men and women, with
$\alpha = \frac 14 g_h$, $\beta=\gamma = \alpha\rho_{\ell}$,
to bound the probability $p_h$ that $w_j$
sustains a loss of more than
$V(r_{m_j},1) - V(r_{m_j}-g_h,1)$.
To match with any man $m_i$ in $B_M$, $w_j$ must sustain such a loss.
Therefore, with probability at least $1-p_h$, $w_j$ does not match with a man in $B_M$.

As none of ${\mathcal B}_1$--${\mathcal B}_{10}$ occur,
By Lemma~\ref{lem::loss-distribution}, the probability that
in the man-proposing DA with cuts at $w_j$ and $r_{w_j} - \alpha $, gives her a loss of more that $L^m_h$
is at most 
$\exp(-2^{3(h-1)}(\eta/n^{1/3})^3\rho_{\ell}^2 n/128)$,
and this bound depends only on the private scores of the proposals between the woman and the men in $T_M$.
But if this does not occur, $L^m_h$ is also a bound
on $w_j$'s loss in the woman proposing DA.
As $L^m_h \le g_h$, this implies $m_j$ is matched to a man
in $T_M$.

As ${\mathcal B}_{10}$ does not occur, by Lemma~\ref{lem::women-in-intermed-zone},
$|W_h|\le 2^{h+1}\eta n^{2/3}$.
Also, as ${\mathcal B}_8$ does not occur,
by Lemma~\ref{lem::max-score-bottom-zone}, $|B_M|\ge \tfrac 12 \sigma_m \cdot n=\tfrac 12 \nu n^{2/3}$.
Finally, recall that 
$\sigmabar^3 = 128(c+2)\ln n/(\rho_{\ell}^2 n)$.
Thus the expected number of matches between women in $T_W$ and men in $B_M$ is at most
\begin{align*}
    &(2^{h+1} \eta n^{2/3}) \cdot \exp(-2^{3(h-1)}(\eta/n^{1/3})^3\rho_{\ell}^2 n/128)  \le 2^{h+2}(\eta/\nu)\cdot(\tfrac 12 \nu n^{2/3})\cdot \exp(-2^{3h}(\eta/2)^3\rho_{\ell}^2/128).
\end{align*}
As $4(\eta/\nu) \cdot \exp(-(\eta/2)^3\rho_{\ell}^2/128) \le \tfrac 1{10}$,
we see that $\exp(-(\eta/2)^3\rho_{\ell}^2/128)\le \nu/40\eta$, and therefore the bound on the number of matches is at most
\begin{align*}
    \frac{2^h \cdot(\tfrac 12 \nu n^{2/3})} {10\big( \frac{40\eta} {\nu}\big)^{2^{3h}-1}}.
\end{align*}
Summing over all $h\ge 0$, we obtain that
the expected number of matches is at most
$\tfrac{1}{8}\cdot(\tfrac 12 \nu n^{2/3})$.
By a Chernoff bound, the number of matches is at most
$\tfrac 14 \cdot(\tfrac 12 \nu n^{2/3})\le\tfrac 14 |B_M|$, with failure probability at most
$\exp(\nu n^{2/3}/48)$.

Next, we argue that this use of a Chernoff bound is justified by stochastic dominance. The expectation is the product of two terms: a bound on $|W_h|$, which follows from the assumption that ${\mathcal B}_{10}$ does not occur, and a bound on the probability that an arbitrary woman $w$ in $W_h$ has a small loss and therefore cannot be proposing to any man in $B_M$. The upper bound on the latter probability depends only on the men's and women's private scores for the proposals from $w$ to  the men in $T_M$, and so we can safely apply stochastic dominance.
\end{proof}

\begin{lemma}
\label{phase2-man-high-prob}
Suppose that that neither ${\mathcal B}_8$ nor ${\mathcal B}_9$ occur.
Then, the probability that a proposal from a woman in $B_W$ to a man in $B_M$ is man-high is at most $\eta\nu^2\rho_u\rho_{\ell} /[ 32 (c+2)\ln n]$, if $t_m\ge 2$.
\end{lemma}
\begin{proof}
Since ${\mathcal B}_8$ does not occur, by Lemma~\ref{lem::max-score-bottom-zone},
every woman in $B_W$ has rating at most $3\sigma_w$.
We now use this to bound the probability that an edge from woman $w_j\in B_W$ to man $m_i\in B_M$ is man-high.
For the edge to be man-high, we need
$U(r_{w_{j}},s_{m_i}(w_j)) \ge U(0,1)$.
Now, $U(0,s_{m_i}(w_j)+r_{w_{j}}\cdot \rho_u) \ge U(r_{w_{j}},s_{m_i}(w_j))$, so the edge is man-high with probability at most $r_{w_{j}}\cdot \rho_u \le 3\sigma_w \cdot \rho_u$.

Because of the truncation, the edge is low if
$U(r_{w_{i}}-\sigmabar\cdot t_m^2,1) \le U(r_{w_{j}},s_{m_i}(w_j))< U(0,1)$.
Because $r_{w_i}\le \sigma_m = \sigmabar/t_m$,
the edge is low if
$U(\sigmabar/t_m-\sigmabar\cdot t_m^2,1) \le U(r_{w_{j}},s_{m_i}(w_j))< U(0,1)$.
Consequently
the probability that the edge is low
is at least
$\rho_{\ell}(\sigmabar t_m^2 - \sigmabar/t_m) \ge \tfrac 34 \rho_{\ell}\sigmabar t_m^2$, as $t_m\ge 2$.

Therefore the probability that a proposal is man-high is at most
\begin{align*}
    \frac{3\sigma_w \cdot\rho_u} {\frac 34\sigmabar t_m^2 \rho_{\ell}} = \frac{4\sigma_w \cdot\rho_u} {\sigmabar t_m^2 \rho_{\ell}}.
\end{align*}

Recall that $\sigmabar = \sigma_m\cdot t_m$ and $\sigmabar^3=128(c+2)\ln n/(\rho_l^2 n)$. Thus, the probability bound is
\begin{align*}
    \frac{4\sigma_w \cdot\rho_u 
    \cdot \sigma_m^2} {\sigmabar^3\rho_{\ell}} = \frac{4\eta\nu^2\rho_u\rho_{\ell}} {128(c+2)\ln n}=
    \frac{\eta\nu^2\rho_u\rho_{\ell}} {32(c+2)\ln n}.
\end{align*}
\end{proof}

We will now analyze the women-low proposals. Note that once a women makes one such proposal, all her subsequent proposals will be woman-low. We now state two assumptions regarding the proposals by women in $B_W$.
They will be demonstrated later.

\begin{asspt}
\label{ass::limits-on-w-props}
i. The edges proposed by each woman in $B_W$ have private score at least $\tfrac 12$.\\
ii. Each woman in $B_W$ proposes to at most half the men in $B_M$.
\end{asspt}

\begin{lemma}
\label{lem::Phase3-almost-random-man-chosen}
Let $w$ be a woman in $B_W$, who is now proposing woman-low edges.
For her next proposal, let $p_{\min}$ be the minimum probability that she selects a particular man in $B_M$, and let $p_{\max}$ be the maximum probability, over the men she has not yet proposed to.
Then $p_{\max}/p_{\min}\le 2\mu_u/\mu_{\ell}$.
\end{lemma}
\begin{proof}
Suppose $w$'s most recent proposal provided her a utility of $u$. Consider the utility interval
$(u,u-\delta u]$.
The probability that she selects a man providing utility in this interval is given by the private score decrease that reduces the utility $u$ to $u-\delta u$ divided by the remaining available private score, which includes the range $[0,\tfrac 12]$ by assumption.
Thus the probability that she selects a particular man in $B_M$ varies between $\delta u \cdot\mu_{\ell}$ and $2\delta u \cdot\mu_u$. 
\end{proof}

\begin{cor}
\label{cor::bound-woman-low-props}
There are at most $d\ln n|B_M|$ woman-low proposals to men in $B_M$, where $d=2(c+2)\mu_u /\mu_{\ell}$, with failure probability at most $n^{-(c+1)}$.
\end{cor}
\begin{proof}
Suppose $m\in B_M$ does not receive a woman-high proposal.
The probability that $m$ receives no proposals among $d\ln n|B_M|$ woman-low proposals
is at most 
\begin{align*}
    \Big(1 - \frac {\mu_{\ell}} {2\mu_u |B_M|}\Big)^{d\ln n|B_M|}
    \le \exp(-d \mu_{\ell} \ln n / 2\mu_u).
\end{align*}
As $d = 2(c+2)\mu_u/\mu_{\ell}$, the probability is at most $n^{-(c+2)}.$
A union bound over the men in $B_M$ yields the claim.
\end{proof}

\begin{lemma}
\label{lem::mhigh-wlow-bound}
The number of man-high proposals
from women in $B_W$ to men in $B_M$ is at most
$d\eta \nu^2 \rho_u\rho_{\ell} |B_M|/[16(c+2)]$,
with failure probability at most $\exp(-d\eta \nu^2 \rho_u\rho_{\ell} |B_M|/[48(c+2)])$.
\end{lemma}
\begin{proof}
By Lemma~\ref{phase2-man-high-prob},
the probability that a proposal is man-high is at most
$\eta\nu^2\rho_u\rho_{\ell}/[32(c+2)\ln n]$.
Over $d\ln n|B_M|$ proposals, this yields an expected
$d\eta\nu^2\rho_u\rho_{\ell}|B_M|/[32(c+2)]$ proposals.
By a Chernoff bound, there are at most
$d\eta \nu^2 \rho_u\rho_{\ell} |B_M|/[16(c+2)]$ such proposals with failure probability at most
$\exp(-d\eta \nu^2 \rho_u\rho_{\ell} |B_M|/[48(c+2)])$.
\end{proof}

\begin{lemma}
\label{lem::bound-on-low-low-props}
Suppose neither ${\mathcal B}_8$ nor ${\mathcal B}_9$ occur.
Then, over the course of the first $d\ln n |B_M|$ woman-low proposals from women in $B_W$ to men in $B_M$,
assuming each woman proposes to at most half the men in $B_M$, no man in $B_M$ receives more than $e_1\ln n$ of these proposals, with failure probability $n\cdot \exp(-4d\ln n (\mu_u/\mu_{\ell})/3)$, where $e_1=8d(\mu_u/\mu_{\ell})$.
\end{lemma}
\begin{proof}
Let $m$ be a man in $B_M$.
First, we bound
the probability that a proposal is to man $m$.
By Lemma~\ref{lem::Phase3-almost-random-man-chosen},
the ratio of probabilities for the proposals to men in $B_M$ is bounded by $2\mu_u/\mu_{\ell}$, and by assumption, as least $\tfrac 12 |B_M|$ have not yet been proposed to.
Therefore, the probability that a proposal is to man $m$
is at most
$2(\mu_u/\mu_{\ell})\cdot (2/|B_M|)$.
Thus the expected number of woman-low proposals $m$  receives is at most 
\begin{align*}
    \frac{2\mu_u} {\mu_{\ell}} \cdot \frac{2}{|B_M|}
     \cdot d\ln n |B_M|
    \le \frac{4\mu_u} {\mu_{\ell}}\cdot d \ln n.
\end{align*}
The upper bounds on these probabilities are based on the women's private scores for $m$, and therefore we can use stochastic dominance to justify applying a Chernoff bound. 
Thus,  the number of these proposals is at most $8d(\mu_u/\mu_{\ell}) \ln n$
with failure probability at most
$\exp(-4d\ln n (\mu_u/\mu_{\ell})/3)$.
A union bound over the men in $B_M$ yields the final result.
\end{proof}

Let $B_{M,h}$ denote the set of men in $B_M$ who eventually receive a man-high proposal, and
$B_{M,\ell}$ denote the set $B_M\setminus B_{M,h}$.

\begin{lemma}
\label{lem::prob-low-low-prop}
If a woman $w_i\in B_W$ is currently matched with a man $m_j$ in $B_{M,\ell}$, the probability that the next woman-low and man-low proposal is to $m_j$ is at most $4(\mu_u/\mu_{\ell})/|B_M|$.
\end{lemma}
\begin{proof}
By Lemma~\ref{lem::Phase3-almost-random-man-chosen},
the ratio of probabilities for the proposals to men in $B_M$ is bounded by $2\mu_u/\mu_{\ell}$, and by assumption, as least $\tfrac 12 |B_M|$ men have not yet been proposed to.
Therefore, the probability that a proposal is to $m_j$
is at most
$2(\mu_u/\mu_{\ell})\cdot (2/|B_M|)$.
\end{proof}

\begin{lemma}
\label{lem::prob-women-is-matched-at-end}
Suppose that $8d(\mu_u/\mu_{\ell})\ln n$ is an integer.
Let $w\in B_W$. If $w$ has at least $e_2(\ln n)^2$ man-low and woman-low edges to men in $B_M$, then the probability that she is unmatched after $d\ln n|B_M|$ man and women-low proposals to $B_M$ is at most $\exp(-\frac 83d \cdot e_1  \big(\frac{\mu_u}{\mu_{\ell}}\big)^2\cdot (\ln n)^2)$,
where $e_2 = 16d \cdot e_1 \big( \frac{\mu_u} {\mu_{\ell}} \big)^2$.
\end{lemma}
\begin{proof}
Suppose $w$ is currently matched to a man in $B_{M,\ell}$. By Lemma~\ref{lem::prob-low-low-prop},
the probability that she is bumped (i.e.\ loses her current match) by the next man and women-low proposal to $B_M$ is at most $4(\mu_u/\mu_{\ell})/|B_M|$.

Therefore, over the course of $d\ln n|B_M|$ such proposals, she is bumped at most an expected
$4d(\mu_u/\mu_{\ell})\ln n$ times.
Using stochastic dominance, we can apply 
a Chernoff bound, which shows
she is bumped at most
$8d(\mu_u/\mu_{\ell})\ln n -1$ times
with failure probability at most $\exp(-\tfrac 43d(\mu_u/\mu_{\ell})\ln n)$.

We now bound the probability that $w$ tentatively matches with that man.
By Lemma~\ref{lem::bound-on-low-low-props},
$m$ receives at most $e_1\ln n$ proposals (including the current proposal by $w$). Each proposal has probability
at most $\mu_u \Delta$ and at least $\mu_{\ell} \Delta$ of being in a $\Delta$ range of loss for the man, and therefore $w$'s proposal produces the least loss among these up to $e_1\ln n$ proposals with probability at least
\begin{align*}
    \frac{\mu_{\ell}}{(e_1\ln n -1)\cdot \mu_u +\mu_{\ell} } \ge  \frac{1}{e_1\ln n} \cdot \frac {\mu_{\ell}}{\mu_u}.
\end{align*}

Note that the bounds for each man are independent as they depend on the private scores of that man for the proposals he has received.

Therefore, to end up matched after these $d\ln n |B_M|$ proposals, it suffices that $w$ make an expected
\begin{align}
    8d\frac{\mu_u}{\mu_{\ell}}\ln n \cdot {e_1\ln n} \cdot \frac {\mu_u} {\mu_{\ell}}=
    8d \cdot e_1 \Big(\frac{\mu_u}{\mu_{\ell}}\Big)^2 \cdot (\ln n)^2~~\text{proposals.}
\end{align}

Then, by a Chernoff bound, she makes at most 
$16d \cdot e_1 \big(\frac{\mu_u}{\mu_{\ell}}\big)^2\cdot (\ln n)^2$ proposals
with failure probability at most
$\exp(-\frac 83 d \cdot e_1  \big(\frac{\mu_u}{\mu_{\ell}}\big)^{2}\cdot (\ln n)^2)$.
\end{proof}

\begin{lemma}
\label{lem::enuf-edges}
Each woman in $B_W$ has at least $e_2(\ln n)^2$ man and woman-low edges to $B_M$ with failure probability at most $\exp(-e_2 \ln^2 n/3)$,
if 
$\nu \le \frac{1}{8\eta^2}\Big(\frac{\mu_{\ell}}{\rho_{\ell}}\Big)^2 \Big(\frac{\mu_u}{\mu_{\ell}}\Big)^5$,
$t_m\ge 2$, and $t_w\ge 2$.
\end{lemma}
\begin{proof}
As in the proof of Lemma~\ref{phase2-man-high-prob},
the probability that an edge from a woman $w_j\in B_W$ to a man $m_i\in B_M$ is  man-low is at least $\tfrac 34 \sigmabar t_w^2\rho_{\ell}$ and this depends on the man's private score for this edge; similarly the probability that it is woman low is at least $\tfrac 34\sigmabar t_m^2 \rho_{\ell}$ and depends on the woman's private score for the edge.
As ${\mathcal B}_8$ does not occur, by Lemma~\ref{lem::max-score-bottom-zone},
$|B_M|\ge \tfrac 12 \sigma_m \cdot n$.
Thus, the expected number of man and woman-low edges from $w_j$ to $B_M$ is at least
\begin{align*}
    \frac 34\sigmabar t_w^2\rho_{\ell} \cdot \frac 34\sigmabar t_m^2 \rho_{\ell} \cdot \tfrac 12 \sigma_m \cdot n
    &\ge \frac 14 \cdot\frac{\sigmabar^6 \cdot n}{\sigma_w^2\sigma_m} (\rho_{\ell})^2\\
    & \ge 32 \cdot 128(c+2)^2  \Big(\frac{\mu_{\ell}}{\rho_{\ell}}\Big)^2\frac{(\ln n)^2}{\eta^2\nu}\\
    &\ge 2e_2(\ln n)^2,~~~~~\text{if}
\end{align*}
\begin{align*}
    2e_2=16de_1 \Big(\frac{\mu_u}{\mu_{\ell}}\Big)^2
   & =8\cdot 16\Big[2(c+2)\Big(\frac{\mu_u}{\mu_{\ell}}\Big)\Big]^2\cdot \Big(\frac{\mu_u}{\mu_{\ell}}\Big)^3 \\
   &  \le 32 \cdot 128(c+2)^2  \Big(\frac{\mu_{\ell}}{\rho_{\ell}}\Big)^2\frac{1}{\eta^2\nu}\\
   \text{i.e., if~~} 8\Big(\frac{\mu_{\ell}}{\rho_{\ell}}\Big)^2 \ge \Big(\frac{\mu_{\ell}}{\mu_u}\Big)^5 \eta^2\nu.
\end{align*}

Applying stochastic dominance, by a Chernoff bound, the number of edges is at least
$e_2 \ln^2 n$ with probability at most $\exp(-e_2 \ln^2 n/3)$.
\end{proof}

\begin{lemma}
\label{lem::TW-matched}
All the women in $T_W$ are matched with failure probability at most $O(n^{-(c+1)})$.
\end{lemma}
\begin{proof}
Let $w_j$ be a woman in $T_W$.
If $r_{m_j}\ge \sigmabar$, then the truncation does not remove any of the acceptable edges to $w_j$ and so the previous analysis shows $w_j$ is matched with failure probability $O(n^{-(c+2)})$.

So now suppose that
$r_{m_j}< \sigmabar$.
Consider a run of man-proposing DA with the edge set cut at 
$w_j$ and $\frac 14 r_{w_j}$.
Now, the acceptable edges are all woman-high.
Furthermore, the acceptable edges cause a man $m_i$ a loss
of at most $U(r_{w_i},1) -U(r_{w_i}-\sigmabar t_w^2,1)$, and these are edges that are not truncated by $m_i$. The proof of Theorem~\ref{thm::low_loss_bounded} shows that such a woman $m_j$ is matched using these edges with failure probability $O(n^{-(c+2)})$.
\end{proof}

\begin{lemma}
\label{lem::assump-holds}
If ${\mathcal B}_8$ does not occur, then, for large enough $n$,
Assumption~\ref{ass::limits-on-w-props} holds with failure probability $\exp(-|B_M|/8)$.
\end{lemma}
\begin{proof}
Assumption (i) holds if $\sigmabar t_w^2 \rho_u \le \tfrac 12$, i.e.\ if $\sigmabar^3/\sigma_w^2 \rho_u \le \tfrac 12$, i.e.\ if $128(c+2)\ln n/(\rho_{\ell}^2 n) \cdot n^{2/3}/\eta^2 \cdot \rho_u \le \tfrac 12$;
this holds if $(n^{1/3}/\ln n)\ge 256 (c+2) \rho_u/[(\rho_{\ell}^2 \eta^2]$, which is true for large enough $n$.

Assumption (ii) holds if each woman in $B_W$ has at most $\tfrac 12 |B_M|$ untruncated edges to men in $B_M$.
The probability that an edge $(m_i,w_j)$ is not truncated by $m_i$ is at most $(r_{w_j} + \sigmabar t_m^2)\rho_u\le 2\sigmabar t_m^2\rho_u$,
and the probability that it is not truncated by $w_j$ is at most $(r_{m_i} + \sigmabar t_w^2)\rho_u \le 2\sigmabar t_w^2\rho_u$.
Thus the expected number of untruncated edges from a woman $w\in B_W$ to the men in $B_M$ is at most
\begin{align*}
    2 \sigmabar t_w^2\rho_u \cdot 2 \sigmabar t_m^2\rho_u \cdot |B_M|
   &\le 4 \frac{\sigmabar^6}{\sigma_w^2 \sigma_m^2} \rho_u^2 |B_M|
    \le \frac{(256 (c+2) \ln n\rho_u)^2}{\rho_{\ell}^4\eta^2 \nu^2 n^{2/3}}|B_M|
    \le \tfrac 14 |B_M|,
\end{align*}
if $n$ is large enough.

Note that the bounds on the probabilities are due to the men's and women's independent private scores for these edges. Thus, using stochastic dominance, by means of a Chernoff bound, we obtain that the number of these edges is at most $\tfrac 12 |B_M|$ with failure probability
$\exp(-|B_M|/8)$.
\end{proof}

\begin{lemma}
\label{lem::constraints and error bound}
The run of woman-proposing DA with the truncated edge sets matches every woman (and man) with failure probability $n^{-c}$ if $n$ is large enough, if
$\nu =\frac{64}{\eta^2}\cdot\Big(\frac{1}{\rho_{\ell}}\big)^2\cdot\Big(\frac{\mu_{\ell}}{\mu_u}\Big)^4$ and
$\eta$ satisfies
$4\eta^3 \rho_{\ell}^4(\mu_u/\mu_{\ell})^4\cdot \exp(-\eta^3\rho_{\ell}^2/128) \le \tfrac 1{10}$.
\end{lemma}
\begin{proof}
As any unmatched woman in $B_W$ will keep proposing until she runs out of proposals, we deduce from
Lemmas~\ref{lem::prob-women-is-matched-at-end} and~\ref{lem::enuf-edges} that
all the women in $B_W$ are matched, modulo the lemma's failure probability.
By Lemma~\ref{lem::TW-matched}, all the women in $T_W$ are matched, modulo the lemma's failure probability. Thus all the women are matched.

This entails the following constraints, from Lemmas~\ref{lem::women-in-intermed-zone},
\ref{lem::num-matches-WM-BL},
\ref{phase2-man-high-prob},
Corollary~\ref{cor::bound-woman-low-props}, Lemmas~\ref{lem::bound-on-low-low-props}, \ref{lem::prob-women-is-matched-at-end}, \ref{lem::enuf-edges}, \ref{lem::num-matches-WM-BL}, respectively.
\begin{align*}
    n^{2/3} &\ge 3\\
    \eta &\ge 6\nu\\
    t_m &\ge 2\\
    d&=2(c+2)\mu_u /\mu_{\ell}\\
    e_1&=16d(\mu_u/\mu_{\ell}) = 8(c+2)(\mu_u/\mu_{\ell})^2\\
    e_2 &= 8d \cdot e_1 ( {\mu_u}/ {\mu_{\ell}} )^2 = 8(c+2)^2(\mu_u/\mu_{\ell})^4\\
    t_w &\ge 2\\
    \nu &\le \frac{1}{8\eta^2}\cdot\Big(\frac{\mu_{\ell}}{\rho_{\ell}}\Big)^2\cdot\Big(\frac{\mu_{u}}{\mu_{\ell}}\Big)^5\\
    \tfrac 1{10} & \ge 4(\eta/\nu) \cdot \exp(-(\eta/2)^3\rho_{\ell}^2/128)  \\
\end{align*}
We set $\nu =\frac{1}{8\eta^2}\cdot\Big(\frac{\mu_{\ell}}{\rho_{\ell}}\Big)^2\cdot\Big(\frac{\mu_{u}}{\mu_{\ell}}\Big)^5$.
The final constraint becomes
\begin{align*}
     \eta^3 (\rho_{\ell}/\mu_{\ell})^2(\mu_{\ell}/\mu_{u})^5\cdot \exp(-(\eta/2)^3\rho_{\ell}^2/128) \le \tfrac 8{10}.
\end{align*}
In addition, we need to satisfy $\eta \ge 6\nu$.
Clearly, $\eta=O(1)$ suffices.

Finally, to ensure $t_w\ge 2$ it suffices to have
\begin{align*}
    \Big(\frac{128(c+2) \ln n} {\rho_{\ell}^2}\Big)^{1/3} \ge 2 \eta,
\end{align*}
and clearly this holds if $n$ is large enough.
As $t_m> t_w$, this also ensures that $t_w\ge 2$.

We also assume that 
$8d(\mu_u/\mu_{\ell})\ln n=8(c+2)(\mu_u/\mu_{\ell})^2\ln n$ is an integer
(in Lemma~\ref{lem::prob-women-is-matched-at-end}).
This can be achieved by increasing $\mu_u$ slightly.
 
The overall failure probability obtained by summing the terms in Lemmas~\ref{lem::enuf-edges}, \ref{lem::prob-women-is-matched-at-end},
\ref{lem::bound-on-low-low-props},
\ref{lem::mhigh-wlow-bound},
\ref{lem::num-matches-WM-BL},
\ref{lem::women-in-intermed-zone},
\ref{lem::smprime-bound},
\ref{lem::max-score-bottom-zone},
\ref{lem::TW-matched},
\ref{lem::assump-holds},
and Corollary~\ref{cor::bound-woman-low-props},
plus ruling out ${\mathcal B}_1$--${\mathcal B}_3$,
is at most
\begin{align*}
    &\exp(-e_2 (\ln n)^2/3)
    +\exp(-\tfrac 83 d \cdot e_1 ({\mu_u}/{\mu_{\ell}})^2\cdot (\ln n)^2)
    + n\cdot \exp(-4d\ln n (\mu_u/\mu_{\ell})/3)\\
    & +\exp(-d\eta \nu^2 \rho_u\rho_{\ell} |B_M|/[48(c+2)])
     +\exp(-|B_M|/24)
    +2\exp\big(-\eta \cdot n^{2/3}/3\big) \\
    &+\exp(-\sigma_m\cdot n)
    +6\exp(-\sigma_m\cdot n/8)
    +\exp(-|B_M|/8)
    +O(n^{-(c+1)}).
\end{align*}
This totals $O(n^{-(c+1)})$, which is bounded by $n^{-c}$ for large enough $n$.
\end{proof}

\begin{proof}(of Theorem~\ref{thm::eq-BN})
Lemma~\ref{lem::constraints and error bound} shows that, with probability at least $1-n^{-c}$, there exists a stable matching, in which every man and woman obtains a match with a loss of less than $L^{m}_{t_m}$ and $L^{w}_{t_w}$, respectively; it results from the men with public rating $\sigmabar t$ implementing reservation strategies with reservation thresholds $L^{m}_{t}$, for $t< t_m$,
and the remaining men using the reservation threshold $L^m_{t_m}$. The edges meeting this constraint are the
acceptable edges for this run of DA.
By Theorem~\ref{thm:basic-bdd-deriv-result}, w.h.p, no man $m$ gets utility greater than $U(r^m,1)+\Theta([\ln n/n]^{1/3})$, and an analogous bound applies to the women. Thus, the most a man could gain by deviating from the equilibrium strategy, in terms of his expected utility, is 
\begin{align*}
n^{-c}\cdot2+(1-n^{-c})\cdot(\Theta([\ln n/n]^{1/3})+L^{m}_{t_m}).
\end{align*}
Since $L^{m}_{t_m}=\Theta(\ln n/n^{1/3})$, this is an $\eps$-Bayes-Nash equilibrium with $\eps=\Theta(\ln n/n^{1/3})$.

Further notice that, for each agent, the number of acceptable edges is at most $\Theta(\ln^2 n)$; furthermore, this bound improves to at most $\Theta(\ln n)$ for all agents outside the bottom $\Theta([\ln n/n]^{1/3})$ fraction of agents.
\end{proof}

\section{Additional Numerical Simulations and Discussion}
\label{sec::more_numerics}
Here we provide another set of the experiments, 
but for $n=\text{1,000}$ instead of 2,000. The relative weight of public ratings and private scores is unchanged ($\lambda=0.8$).

\subsection{Numbers of Available Edges}
\subsubsection{One-to-one} $n=\text{1,000}$, $\lambda=0.8$, $L=0.15$, 100 runs.

\begin{figure}[!ht]
    \centering
    \subfloat[Number of edges in the acceptable edge set for each woman.]{	\includegraphics[width=0.47\linewidth,height=4cm]{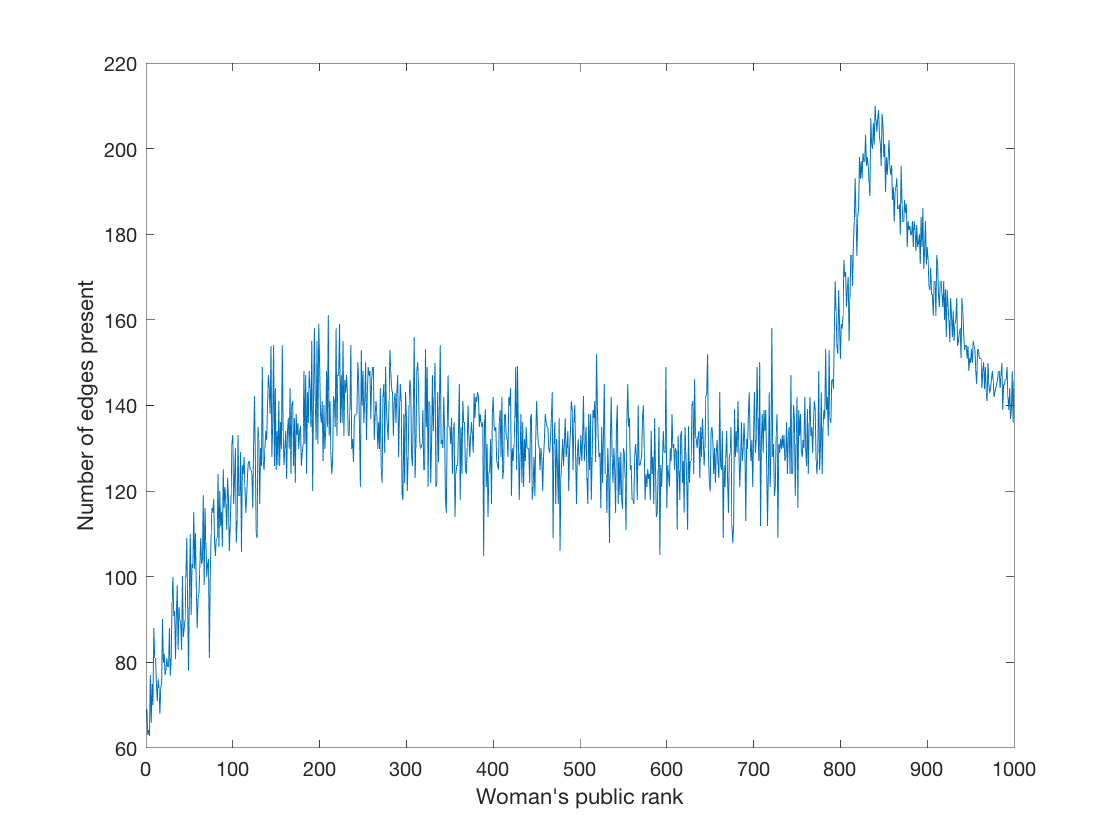}}
    \qquad
    \subfloat[Number of edges in the acceptable edge set  proposed by each woman.
    ]{\includegraphics[width=0.47\linewidth,height=4cm]{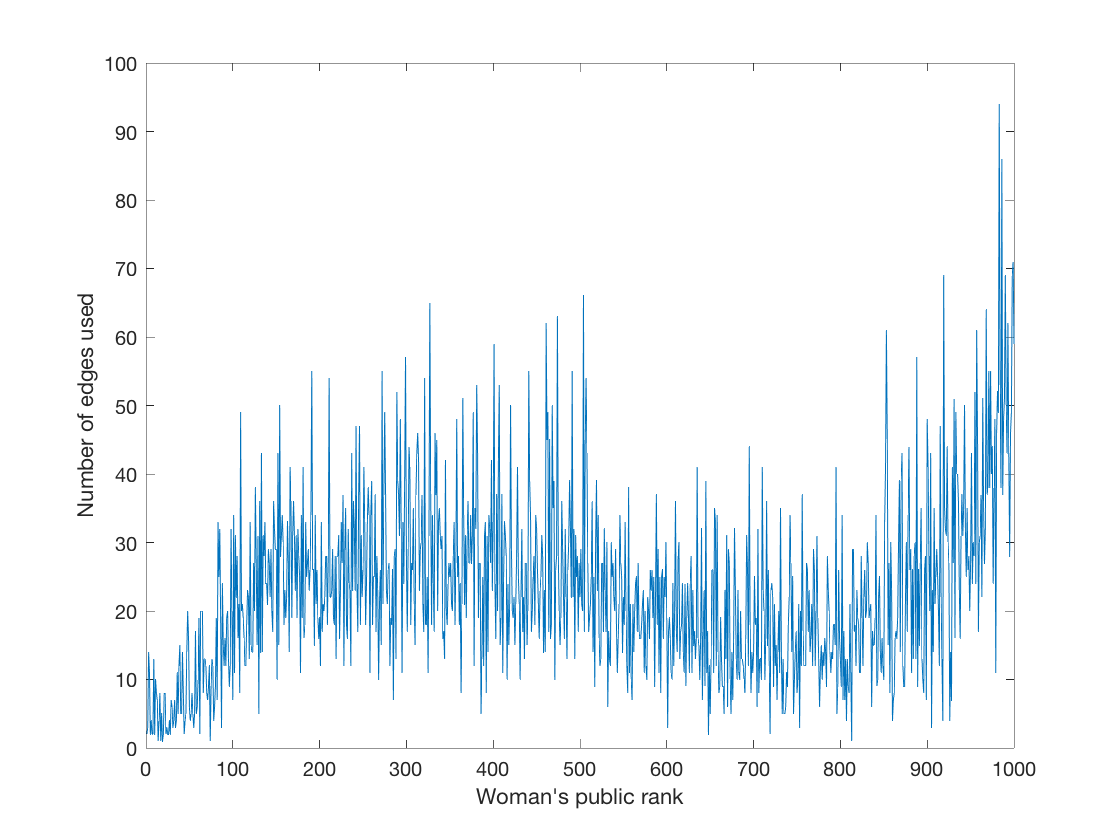}}
    \caption{One-to-one case: Outcome in a typical run.}
\label{fig:edges_typ_1000}
\end{figure}

\begin{figure}[!ht]
    \centering
    \subfloat[Number of edges in the acceptable edge set, per woman, by decile; average in blue with circles, minimum in red with stars.]{	\includegraphics[width=0.47\linewidth,height=4cm]{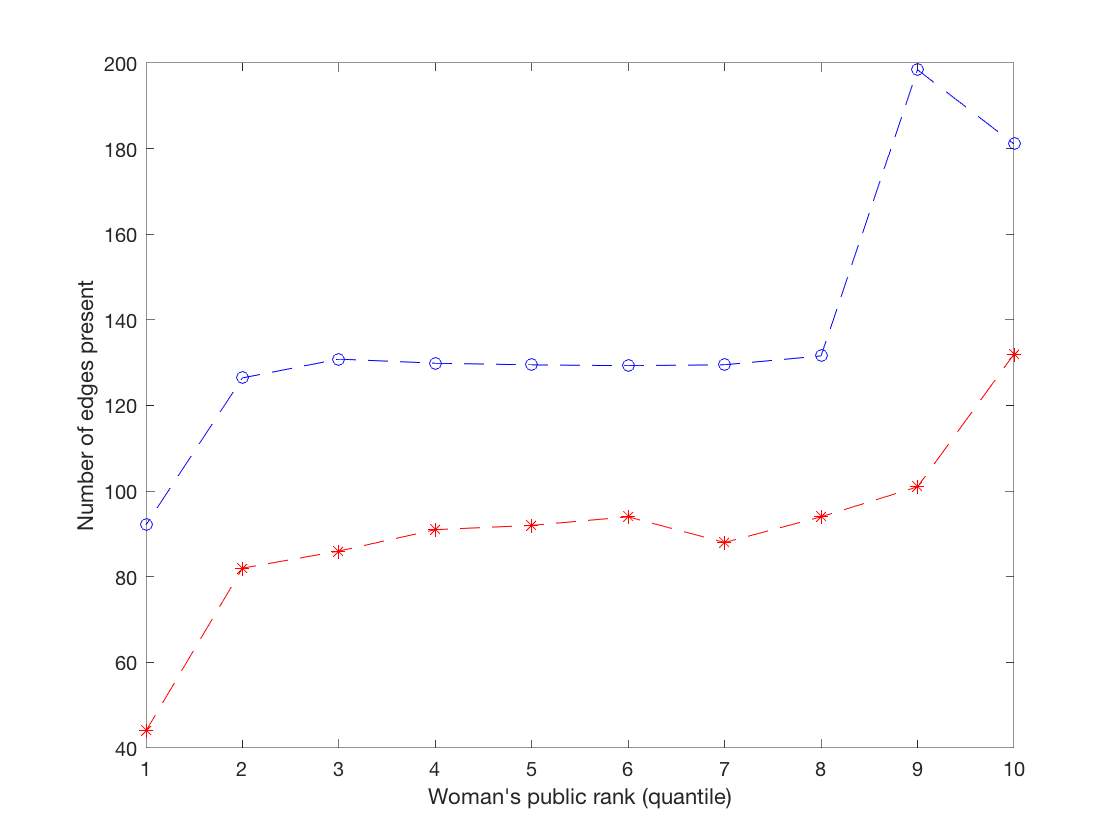}}
    \qquad
    \subfloat[Number of edges in the acceptable edge set proposed during the run of DA, per women, by decile; average in blue with circles, maximum in red with stars.]{\includegraphics[width=0.47\linewidth,height=4cm]{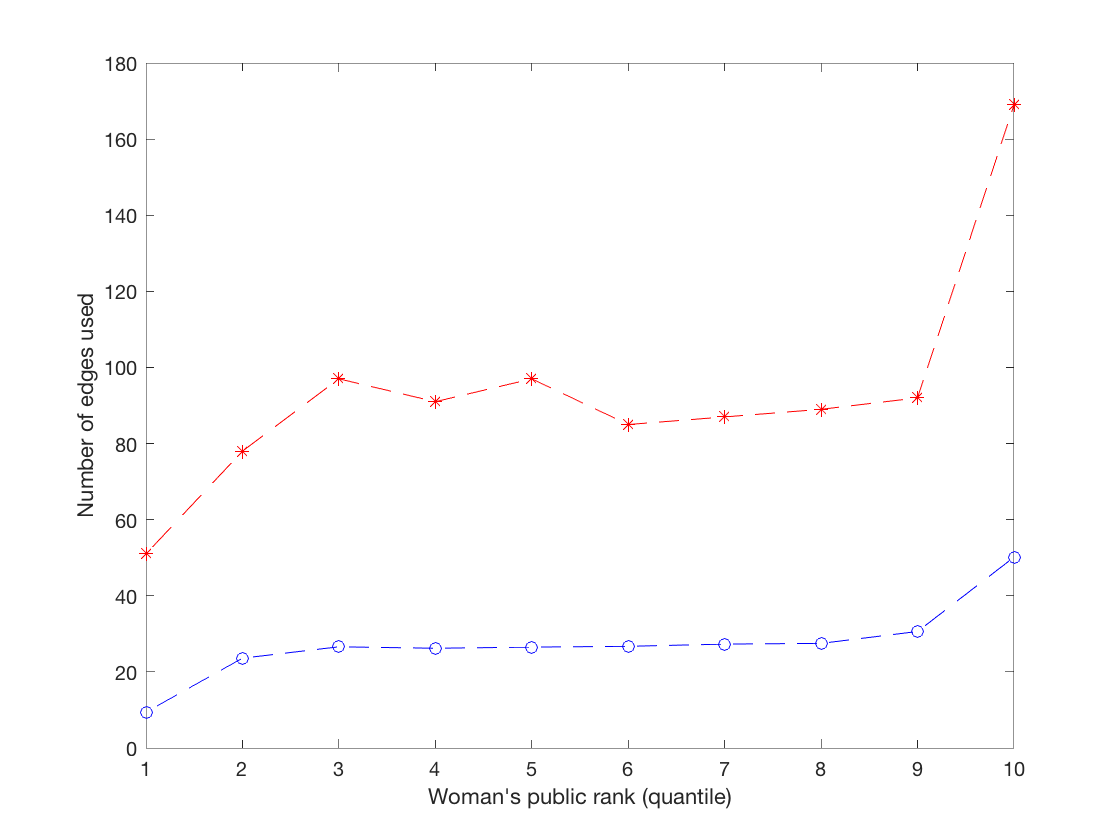}}
    \caption{One-to-one case: summary statistics.}
\label{fig:edges_avg_1000}
\end{figure}

\subsubsection{Many-to-one} $n=\text{1,000}$, $\lambda=0.8$, $d=4$, $L_c=0.16$, $L_w=0.25$, 100 runs.

We chose to present the results for $d=4$ rather than $8$ (as used in the $n=\text{2,000}$ experiments) because the needed value for $L_w$ with $d=8$ leads to very large acceptable edge sets, which we do not consider an interesting case.

\begin{figure}[H]
    \centering
    \subfloat[Number of edges in the acceptable edge set, per woman, by decile; minimum in red with stars, average in blue with circles. ($n_w=\text{1,000}$, $d=4$, $\lambda=0.8$, $L_c = 0.15$, $L_w=0.25$.)]{	\includegraphics[width=0.47\linewidth,height=4cm]{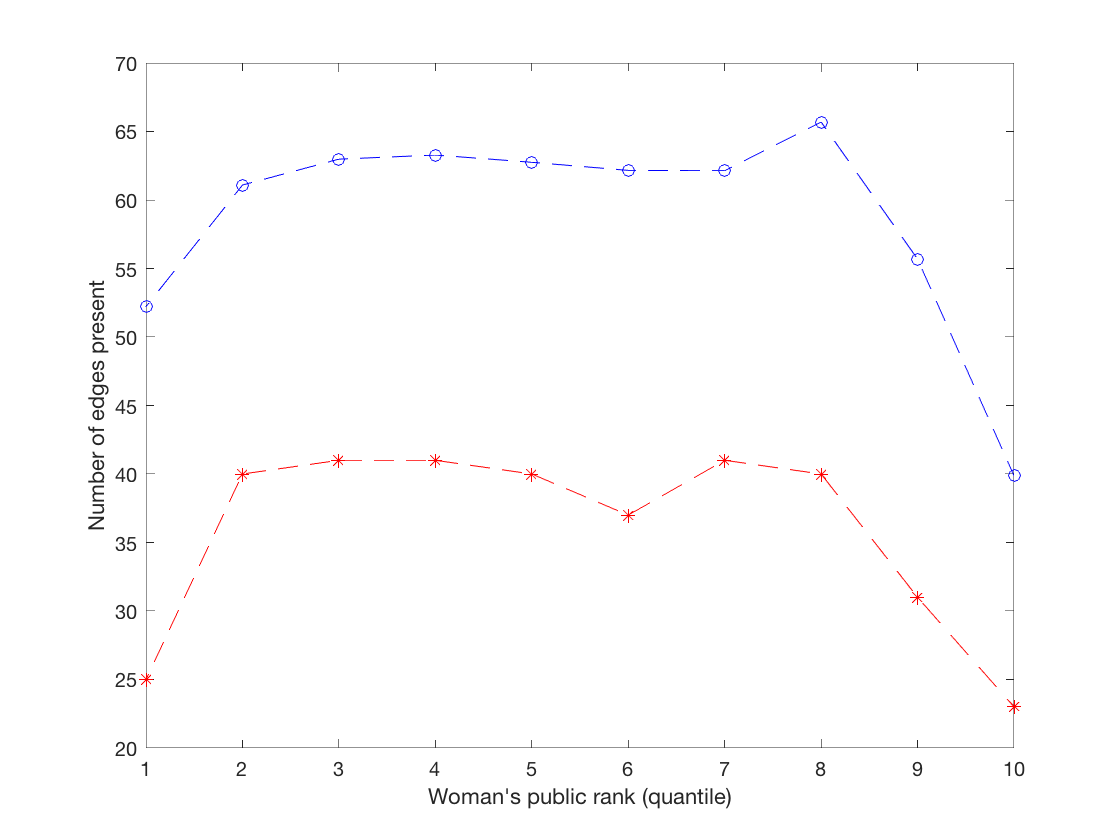}}
    \qquad
    \subfloat[Number of edges in the acceptable edge set proposed during the run of DA, per  woman, by decile; maximum in red with stars, average in blue with circles. 
    ]{\includegraphics[width=0.47\linewidth, height=4cm]{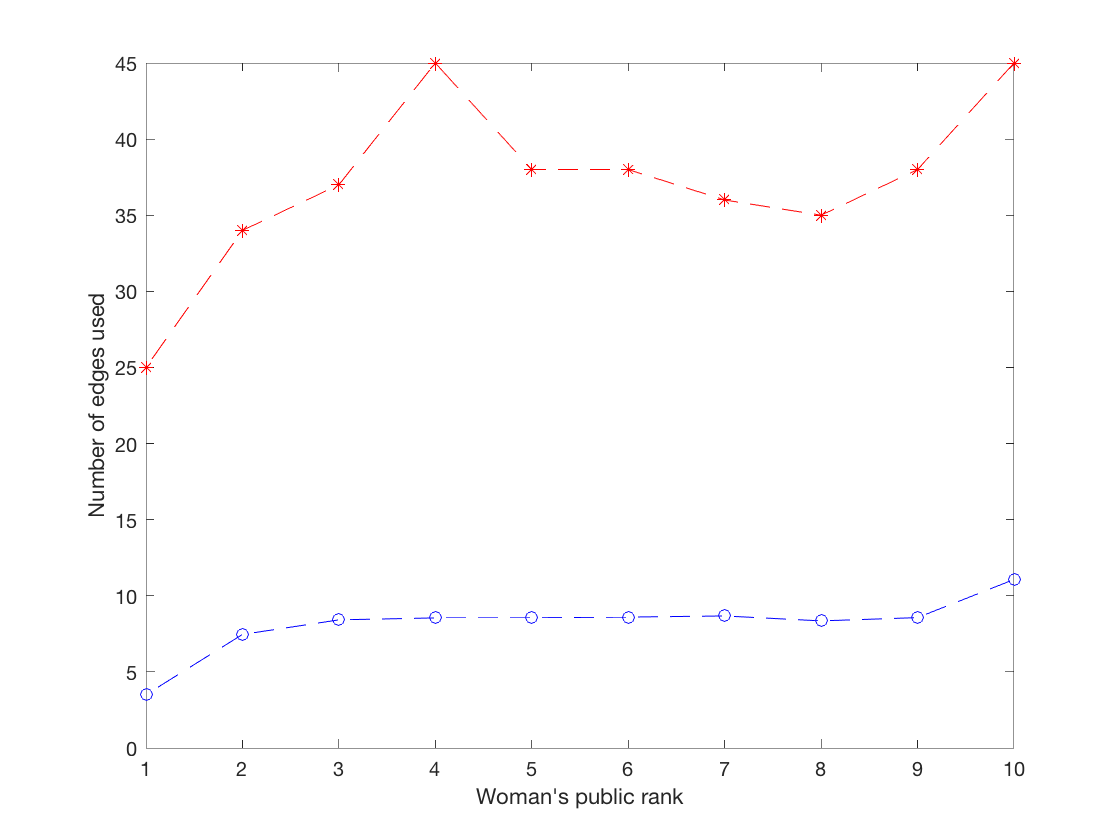}}
    \caption{Many to One Setting}
\label{fig:edges_typ_4_1000}
\end{figure}

\hide{
We explain the reason for the perhaps surprising drop at the ends of the curves in Figure~\ref{fig:edges_typ_4_1000}(a). We observe that for a worker $w$ the range of company public ratings for its acceptable edges is
$[s_c-\Theta(L_w),s_c+\Theta(L_c)]$, and at the ends a portion of this range will be cut off, reducing the number of acceptable edges, with the effect more pronounced for low public ratings. 
}

\subsection{Unique Stable Partners}

100 runs; 38 men have multiple stable partners in the typical run shown.
\begin{figure}[h!]
    \centering
    \subfloat[Public rank of men with multiple stable partners in a typical run.]{	\includegraphics[width=0.47\linewidth]{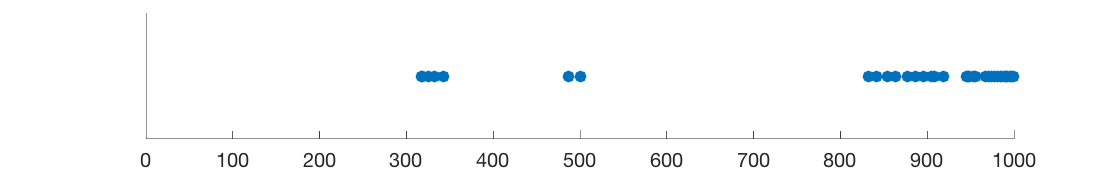}}
    \qquad
    \subfloat[Average numbers of men with multiple stable partners, by decile.]{\includegraphics[width=0.47\linewidth,height=4cm]{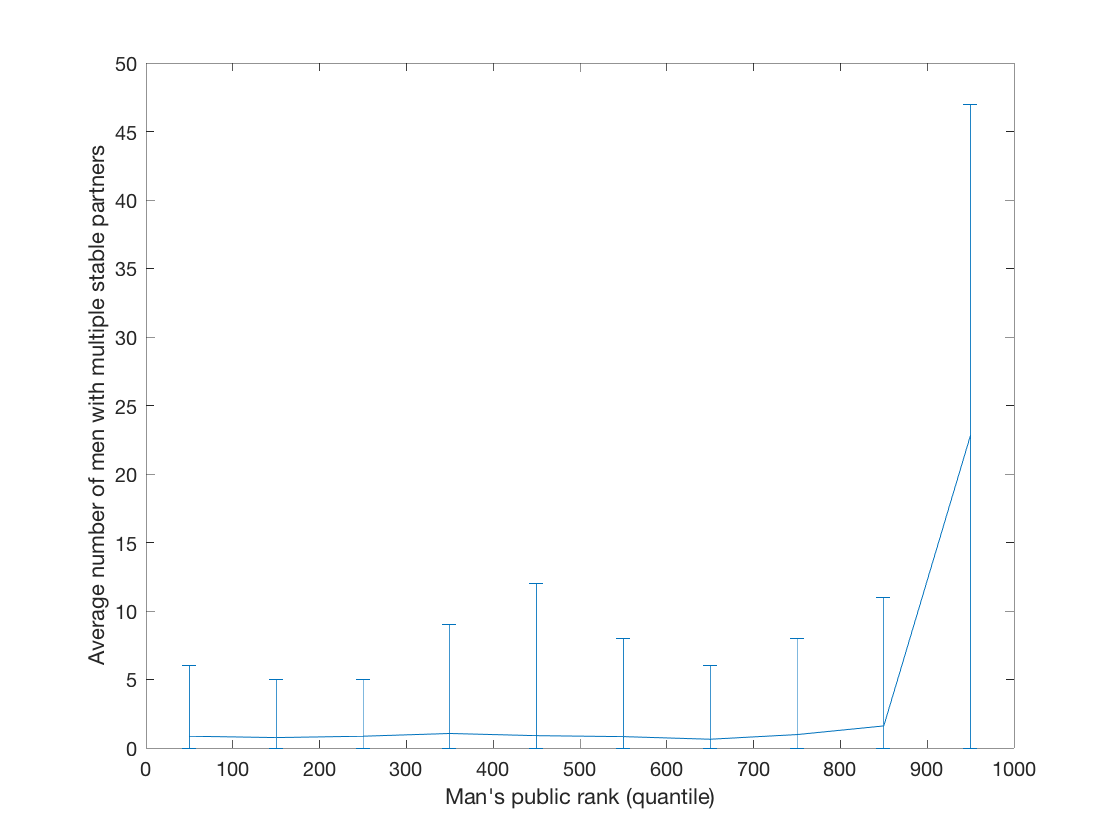}}
    \caption{Unique stable partners, one-to-one setting, $n=\text{1,000}$.}
\label{fig:multiple_stable_1000}
\end{figure}

\subsection{Constant Number of Proposals}

$r=0.19$, $q=0.60$, company capacity $=4$, 100 runs.
\begin{figure}[H]
    \centering
    \subfloat[Public ranks of unmatched workers in a typical run.]{	\includegraphics[width=0.29\linewidth]{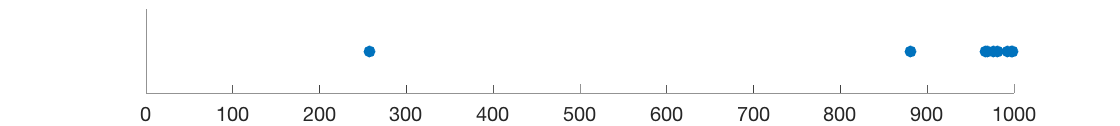}}
    \qquad
    \subfloat[Average number of unmatched workers, by decile.]{\includegraphics[width=0.29\linewidth]{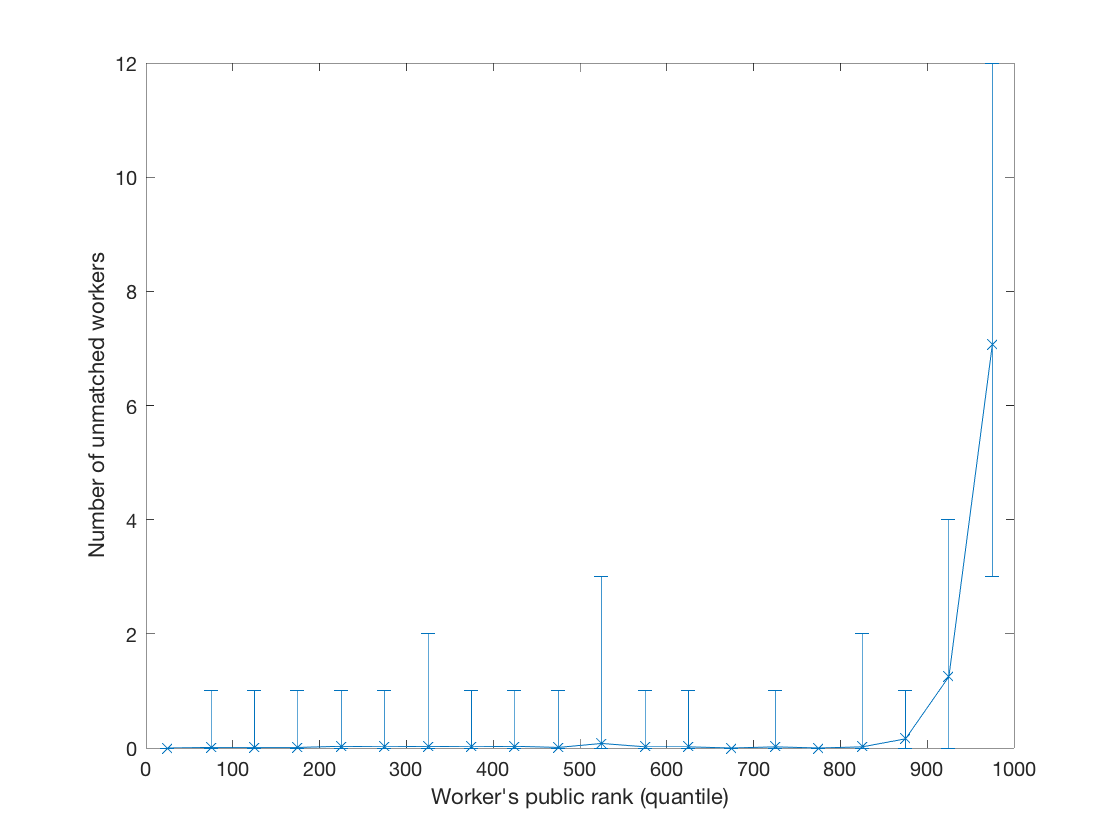}}
    \qquad
    \subfloat[Distribution of workers' utilities with worker-proposing DA: $\text{(full edge set result)}$ $- \text{(Interview edge set result)}$]{\includegraphics[width=0.29\linewidth]{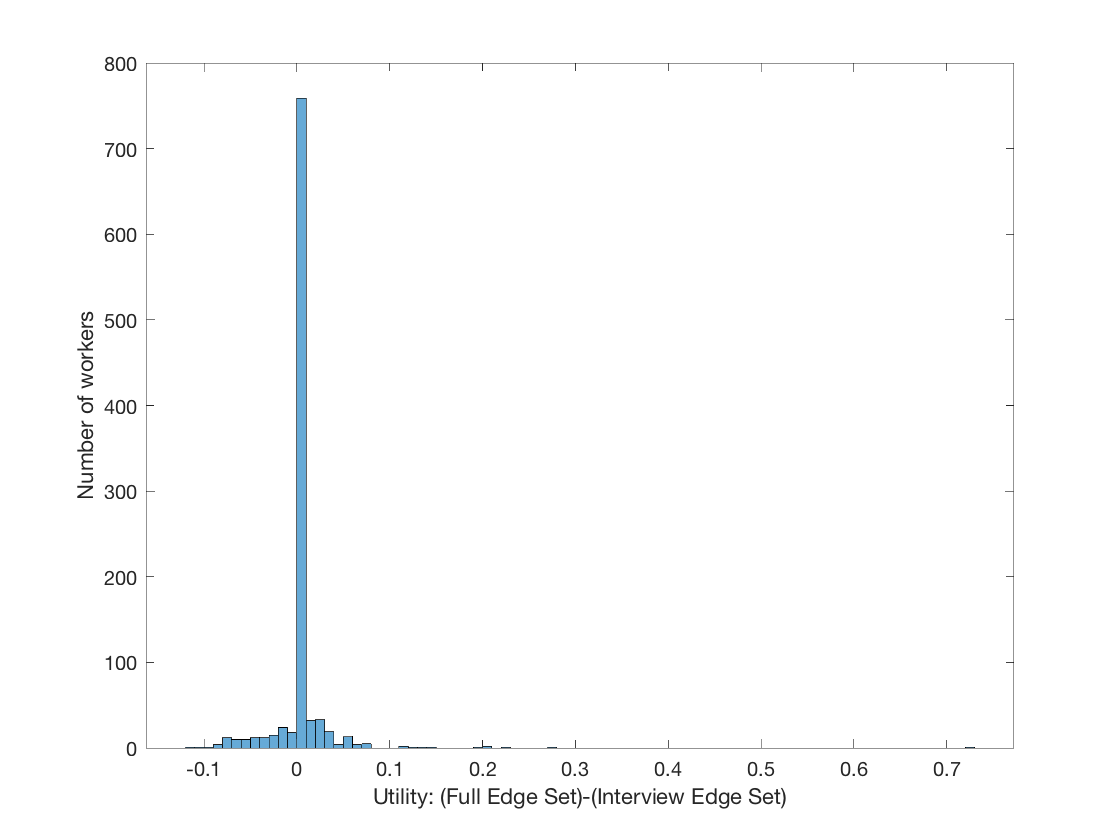}}
    \caption{Constant number of proposals, $n=\text{1,000}$.}
    
\label{fig:constant_1000}
\end{figure}

\printbibliography

@article{APR05,
Author = {Abdulkadiroğlu, Atila and Pathak, Parag A. and Roth, Alvin E.},
Title = {The New York City High School Match},
Journal = {American Economic Review},
Volume = {95},
Number = {2},
Year = {2005},
Month = {May}, 
Pages = {364-367},
DOI = {10.1257/000282805774670167},
URL = {https://www.aeaweb.org/articles?id=10.1257/000282805774670167}}

@article{ABKS19,
author ={Ashlagi, Itai and Braverman, Mark and Kanoria, Yash and Shi, Peng},
year ={2019},
title = {Clearing Matching Markets Efficiently: Informative Signals and Match Recommendations},
journal={Management Science},
volume={66},
number={5},
pages={2163--2193},
doi={10.1287/mnsc.2018.3265}
}

@article{AKL17,
author={Ashlagi, Itai and Kanoria, Yash and Leshno, Jacob D.},
title={Unbalanced Random Matching Markets: The Stark Effect of Competition},
journal={Journal of Political Economy},
Volume={125},
Number={1},
year={2017},
doi={10.1086/689869}
}

@article{CKN13,
Author = {Coles, Peter and Kushnir, Alexey and Niederle, Muriel},
Title = {Preference Signaling in Matching Markets},
Journal = {American Economic Journal: Microeconomics},
Volume = {5},
Number = {2},
Year = {2013},
Month = {May},
Pages = {99-134},
DOI = {10.1257/mic.5.2.99},
URL = {https://www.aeaweb.org/articles?id=10.1257/mic.5.2.99}}

@article{GS,
 ISSN = {00029890, 19300972},
 URL = {http://www.jstor.org/stable/2312726},
 author = {D. Gale and L. S. Shapley},
 journal = {The American Mathematical Monthly},
 number = {1},
 pages = {9--15},
 publisher = {Mathematical Association of America},
 title = {College Admissions and the Stability of Marriage},
 volume = {69},
 year = {1962}
}

@inproceedings{GNKR19,
author = {Gonczarowski, Yannai A. and Nisan, Noam and Kovalio, Lior and Romm, Assaf},
title = {Matching for the Israeli "Mechinot" Gap-Year Programs: Handling Rich Diversity Requirements},
year = {2019},
isbn = {9781450367929},
publisher = {Association for Computing Machinery},
address = {New York, NY, USA},
url = {https://doi.org/10.1145/3328526.3329620},
doi = {10.1145/3328526.3329620},
booktitle = {Proceedings of the 2019 ACM Conference on Economics and Computation},
pages = {321},
numpages = {1},
location = {Phoenix, AZ, USA},
series = {EC '19}
}

@inproceedings{GNOR15,
author = {Gonczarowski, Yannai A. and Nisan, Noam and Ostrovsky, Rafail and Rosenbaum, Will},
title = {A Stable Marriage Requires Communication},
year = {2015},
publisher = {Society for Industrial and Applied Mathematics},
address = {USA},
booktitle = {Proceedings of the Twenty-Sixth Annual ACM-SIAM Symposium on Discrete Algorithms},
pages = {1003–1017},
numpages = {15},
location = {San Diego, California},
series = {SODA '15}
}

@article{HRS17,
Author = {Hassidim, Avinatan and Romm, Assaf and Shorrer, Ran I.},
Title = {Redesigning the Israeli Psychology Master's Match},
Journal = {American Economic Review},
Volume = {107},
Number = {5},
Year = {2017},
Month = {May},
Pages = {205-09},
DOI = {10.1257/aer.p20171048},
URL = {https://www.aeaweb.org/articles?id=10.1257/aer.p20171048}}

@article{IM15,
author = {Immorlica, Nicole and Mahdian, Mohammad},
title = {Incentives in Large Random Two-Sided Markets},
year = {2015},
issue_date = {June 2015},
publisher = {Association for Computing Machinery},
address = {New York, NY, USA},
volume = {3},
number = {3},
issn = {2167-8375},
url = {https://doi.org/10.1145/2656202},
doi = {10.1145/2656202},
journal = {ACM Trans. Econ. Comput.},
month = {jun},
articleno = {14},
numpages = {25}
}

@inproceedings{KMQ21,
author = {Kanoria, Yash and Min, Seungki and Qian, Pengyu},
title = {In Which Matching Markets Does the Short Side Enjoy an Advantage?},
year = {2021},
isbn = {9781611976465},
publisher = {Society for Industrial and Applied Mathematics},
address = {USA},
pages = {1374–-1386},
numpages = {13},
location = {Virtual Event, Virginia},
series = {SODA '21}
}

@book{Knuth76,
title={Mariages stables et leurs relations avec d'autres problèmes combinatoires : introduction à l'analyse mathémathique des algorithmes},
author={Knuth, Donald E.},
year={1976},
publisher={Les Presses de l'Université de Montréal},
}

@book{Knuth96,
author={Knuth, Donald E.},
title={Stable Marriage and Its Relation to Other Combinatorial Problems: An Introduction to the Mathematical Analysis of Algorithms},
year={1996},
publisher={CRM Proceedings \& Lecture Notes},
volume={10},
}

@article{KP09,
Author = {Kojima, Fuhito and Pathak, Parag A.},
Title = {Incentives and Stability in Large Two-Sided Matching Markets},
Journal = {American Economic Review},
Volume = {99},
Number = {3},
Year = {2009},
Month = {June},
Pages = {608-27},
DOI = {10.1257/aer.99.3.608},
URL = {https://www.aeaweb.org/articles?id=10.1257/aer.99.3.608}}

@inproceedings{Kupfer20,
  title={The Influence of One Strategic Agent on the Core of Stable Matchings},
  author={Ron Kupfer},
  booktitle={WINE},
  year={2020}
}

@article{Lee16,
  title={Incentive compatibility of large centralized matching markets},
  author={Lee, SangMok},
  journal={The Review of Economic Studies},
  volume={84},
  number={1},
  pages={444--463},
  year={2016},
  publisher={Review of Economic Studies Ltd}
}

@article{Mertens05,
author = {Mertens, Stephan},
year = {2005},
month = {10},
pages = {},
title = {Random Stable Matchings},
volume = {2005},
journal = {Journal of Statistical Mechanics: Theory and Experiment},
doi = {10.1088/1742-5468/2005/10/P10008}
}

@article{Pittel92,
author = {Boris Pittel},
title = {{On Likely Solutions of a Stable Marriage Problem}},
volume = {2},
journal = {The Annals of Applied Probability},
number = {2},
publisher = {Institute of Mathematical Statistics},
pages = {358 -- 401},
year = {1992},
doi = {10.1214/aoap/1177005708},
URL = {https://doi.org/10.1214/aoap/1177005708}
}

@article{RLPC21,
author={Rios, Ignacio and Larroucau, Tomás and Parra, Giorgiogiulio and Cominetti, Roberto},
title={Improving the Chilean College Admissions System},
journal={Operations Research},
volume={69},
number={4},
pages={1186--1205},
year={2021},
doi={10.1287/opre.2021.2116}
}

@article{RP99,
Author = {Roth, Alvin E. and Peranson, Elliott},
Title = {The Redesign of the Matching Market for American Physicians: Some Engineering Aspects of Economic Design},
Journal = {American Economic Review},
Volume = {89},
Number = {4},
Year = {1999},
Month = {September},
Pages = {748-780},
DOI = {10.1257/aer.89.4.748},
URL = {https://www.aeaweb.org/articles?id=10.1257/aer.89.4.748}}

@inproceedings{Shorrer19,
author = {Shorrer, Ran I.},
title = {Simultaneous Search: Beyond Independent Successes},
year = {2019},
isbn = {9781450367929},
publisher = {Association for Computing Machinery},
address = {New York, NY, USA},
url = {https://doi.org/10.1145/3328526.3329599},
doi = {10.1145/3328526.3329599},
booktitle = {Proceedings of the 2019 ACM Conference on Economics and Computation},
pages = {347–348},
numpages = {2},
location = {Phoenix, AZ, USA},
series = {EC '19}
}

@misc{NRMP,
  author =       "nrmp.org",
  year =         "2021",
  month 	= May,
  title =        "Results and Data, 2021 Main Residency Match",
  lastaccessed = "January 31, 2022",
  url =          "https://www.nrmp.org/match-data-analytics/residency-data-reports/",
}

@article{PSV07,
  title={On the number of fixed pairs in a random instance of the stable marriage problem},
  author={Pittel, Boris and Shepp, Larry and Veklerov, Eugene},
  journal={SIAM Journal on Discrete Mathematics},
  volume={21},
  number={4},
  pages={947--958},
  year={2008},
  publisher={SIAM}
}

@article{KMRP90,
  title={Stable husbands},
  author={Knuth, Donald E and Motwani, Rajeev and Pittel, Boris},
  journal={Random Structures \& Algorithms},
  volume={1},
  number={1},
  pages={1--14},
  year={1990},
  publisher={Wiley Online Library}
}

@article{Pittel89,
  title={The average number of stable matchings},
  author={Pittel, Boris},
  journal={SIAM Journal on Discrete Mathematics},
  volume={2},
  number={4},
  pages={530--549},
  year={1989},
  publisher={SIAM}
}

@article{Pittel18,
  title={On likely solutions of the stable matching problem with unequal numbers of men and women},
  author={Pittel, Boris},
  journal={Mathematics of Operations Research},
  volume={44},
  number={1},
  pages={122--146},
  year={2019},
  publisher={INFORMS}
}

@article{GMM19,
  title={Incentives in Popularity-based Random Matching Markets},
  author={Gimbert, Hugo and Mathieu, Claire and Mauras, Simon},
  year={2019}
}

\end{document}